\documentclass[11pt]{article}

\usepackage[margin = 1in]{geometry}

\usepackage{amsmath,amssymb,amsthm,hyperref,color,mathtools,fixmath}
\hypersetup{colorlinks=true,citecolor=blue, linkcolor=blue, urlcolor=blue}

\usepackage{graphicx}
\usepackage{color}

\newtheorem{theorem}{Theorem}

\newtheorem{lemma}[theorem]{Lemma}
\newtheorem{proposition}[theorem]{Proposition}

\newtheorem{observation}[theorem]{Observation}
\newtheorem{corollary}[theorem]{Corollary}
\newtheorem{remark}[theorem]{Remark}
\newtheorem{definition}[theorem]{Definition}

\newcommand{\eee}{\mathrm{e}}
\newcommand{\eps}{\epsilon}

\newcommand{\Tmix}{T_\mathrm{mix}}

\newcommand{\degree}{\mathrm{deg}}
\newcommand{\fptas}{\mathsf{FPTAS}}
\newcommand{\fpras}{\mathsf{FPRAS}}
\newcommand{\FFF}{\mathcal{F}}
\newcommand{\poly}{\mathrm{poly}}

\newcommand{\R}{\mathbf{R}}

\newcommand{\indicator}[1]{{\mathbf{\large 1}\left\{{#1}\right\}}}
\newcommand{\TreeD}{\mathbb{T}_{\Delta}}

\newcommand{\Expectation}{\mathbb{E}}
\newcommand{\Exp}[1]{{\Expectation\left[{#1}\right]}}

\newcommand{\ExpCond}[2]{{\Expectation\left[{#1} \mid {#2} \right]}}

\newcommand{\ExpCondSub}[3]{{\Expectation_{#1}\left[{#2} \mid {#3} \right]}}

\newcommand{\Probability}{\mathbf{Pr}}
\newcommand{\Prob}[1]{{\Probability\left[{#1}\right]}}
\newcommand{\ProbSub}[2]{{\Probability_{#1}\left[{#2}\right]}}

\newcommand{\ProbCondSub}[3]{{\Probability_{#1}\left[{#2} \mid {#3} \right]}}

\newcommand{\I}{\mathbf{U}}
\newcommand{\SOS}{\mathbf{S}}
\newcommand{\W}{\mathbf{W}}
\newcommand{\dtv}{d_{\mathsf{TV}}}

\newcommand\blfootnote[1]{%
  \begingroup
  \renewcommand\thefootnote{}\footnote{#1}%
  \addtocounter{footnote}{-1}%
  \endgroup
}

\title{Convergence of MCMC and Loopy BP in the Tree Uniqueness Region for
the Hard-Core Model}

\author{
Charilaos Efthymiou\thanks{Goethe University, Frankfurt am Main, Germany. Email:
efthymiou@gmail.com. Research supported by DFG grant EF 103/11.}
\and
Thomas P. Hayes\thanks{Department of Computer Science, University of New Mexico,
Albuquerque, NM 87131. Email: hayes@cs.unm.edu.}
\and
Daniel \v{S}tefankovi\v{c}\thanks{
Department of Computer Science, University of Rochester,
Rochester, NY 14627.  Email: stefanko@cs.rochester.edu. Research
supported in part by NSF grant CCF-1318374.}
\and 
Eric Vigoda
\thanks{School of Computer Science, Georgia
Institute of Technology, Atlanta GA 30332. Email:
ericvigoda@gmail.com. Research supported in part by
NSF grants CCF-1217458 and CCF-1555579.}
\and Yitong Yin\thanks{State Key Laboratory for Novel Software Technology, Nanjing University, China. Email: yinyt@nju.edu.cn. Research supported by NSFC grants 61272081 and 61321491.}
}

\begin{document}

\maketitle

\thispagestyle{empty}

\begin{abstract}
We study the hard-core (gas) model defined on independent sets 
of an input graph where the independent sets are weighted
by a parameter (aka fugacity) $\lambda>0$.  
For constant $\Delta$, previous work of Weitz (2006) 
established an FPTAS for the partition function
for graphs of maximum degree $\Delta$ when $\lambda<\lambda_c(\Delta)$.
Sly (2010) showed that  there is no FPRAS, unless NP=RP, when $\lambda>\lambda_c(\Delta)$.
The threshold $\lambda_c(\Delta)$ is the critical point for the statistical physics
phase transition for uniqueness/non-uniqueness on the infinite $\Delta$-regular tree.
The running time of Weitz's algorithm is exponential in $\log{\Delta}$.
Here we present an FPRAS for the partition function whose running time is $O^*(n^2)$.
We analyze the simple single-site Markov chain known as
the Glauber dynamics for sampling from the associated
Gibbs distribution.
We prove there exists a constant $\Delta_0$
such that for all graphs with maximum degree $\Delta\geq\Delta_0$ 
and girth $\geq 7$ (i.e., no
cycles of length $\leq 6$), the mixing time of the Glauber
dynamics is $O(n\log{n})$ when $\lambda<\lambda_c(\Delta)$.
Our work complements that of Weitz which applies for small constant $\Delta$
whereas our work applies for all $\Delta$ at least a
sufficiently large constant $\Delta_0$ (this includes $\Delta$ depending
on $n=|V|$).

Our proof utilizes loopy BP (belief propagation) which is a widely-used
algorithm for inference in graphical models.
A novel aspect of our work is 
using the principal eigenvector for the BP operator to design a 
distance function which contracts in expectation for pairs of states that
behave like the BP fixed point.
We also prove that the Glauber dynamics behaves locally like
loopy BP.  As a byproduct we obtain that the Glauber dynamics,
after a short burn-in period,  converges close to the BP fixed point, and this implies that the
fixed point of loopy BP is a close approximation to the Gibbs distribution.
Using these connections we establish that loopy BP quickly
converges to the Gibbs distribution 
when the girth $\geq 6$ and $\lambda<\lambda_c(\Delta)$.

\end{abstract}


\blfootnote{This work was done in part while all of the authors were 
visiting the Simons Institute for the Theory of Computing.
}

\newpage

\setcounter{page}{1}

\section{Introduction}

\subsection*{Background}
The hard-core gas model is a natural combinatorial problem that has played
an important role in the design of new approximate counting algorithms and
for understanding computational connections to statistical physics phase transitions.
For a graph $G=(V,E)$ and a fugacity $\lambda>0$,  
the hard-core model is defined on the set $\Omega$ of independent sets of $G$
where $\sigma\in\Omega$ has weight $w(\sigma)=\lambda^{|\sigma|}$.
The equilibrium state of the system is described by the
Gibbs distribution $\mu$ in which an independent set $\sigma$ has
probability $\mu(\sigma) = w(\sigma)/Z$.  The partition function 
$Z=\sum_{\sigma\in\Omega} w(\sigma)$.

We study the closely related problems of efficiently approximating the partition function
and approximate sampling from the Gibbs distribution.  
These problems are important for Bayesian inference in graphical models
where the Gibbs
distribution corresponds to the posterior or likelihood distributions.
Common approaches used in practice are Markov Chain Monte Carlo (MCMC)
algorithms and message passing algorithms, such as loopy BP (belief propagation),
and one of the aims of this paper is to prove fast convergence of these algorithms.

Exact computation of the partition function is \#P-complete \cite{Valiant}, 
even for restricted input classes~\cite{Greenhill}, hence the focus is on
designing an efficient approximation scheme, either a deterministic
$\fptas$ or randomized 
$\fpras$.   The existence of an $\fpras$ for the partition
function is polynomial-time inter-reducible to approximate sampling from 
the Gibbs distribution.

A beautiful connection has been established: 
there is a computational phase
transition on graphs of maximum degree $\Delta$
that coincides with the statistical physics phase transition on $\Delta$-regular
trees.  The critical point for both of these phase transitions is 
$\lambda_c(\Delta):=(\Delta-1)^{\Delta-1}/(\Delta-2)^{\Delta}$.  
In statistical physics, $\lambda_c(\Delta)$ is
the critical point for the uniqueness/non-uniqueness phase transition 
on the infinite $\Delta$-regular tree $\TreeD$ \cite{Kelly} (roughly speaking, this is the
phase transition for
the decay versus persistence of the influence of the leaves on the root).
For some basic intuition about the value of this critical point, note its
asymptotics $\lambda_c(\Delta) \sim e/(\Delta-2)$
and the following basic property: $\lambda_c(\Delta)>1$ for
$\Delta\leq 5$ and $\lambda_c(\Delta)<1$ for $\Delta\geq 6$. 

Weitz \cite{Weitz} showed, for all constant $\Delta$,
an $\fptas$ for the partition function for all
graphs of maximum degree $\Delta$ when $\lambda<\lambda_c(\Delta)$.
To properly contrast the performance of our algorithm with Weitz's algorithm
let us state his result more precisely: for all $\delta>0$, 
there exists constant $C=C(\delta)$, for all $\Delta$,
all $G=(V,E)$ with
maximum degree $\Delta$, all $\lambda<(1-\delta)\lambda_c(\Delta)$, all $\eps>0$,
there is a deterministic algorithm to approximate $Z$ within a factor $(1\pm\eps)$
with running time $O\left((n/\eps)^{C\log{\Delta}}\right)$.
An important limitation of Weitz's 
result is the exponential dependence on $\log{\Delta}$ in the running time.  
Hence it is polynomial-time only for constant $\Delta$,
and even in this case the running time is unsatisfying.

On the other side, Sly \cite{Sly} (extended in \cite{GGSVY,GSV:ising,SlySun,GSV:colorings})
has established that, unless $NP=RP$, for all $\Delta\geq 3$,  there exists $\gamma>0$,
for all $\lambda>\lambda_c(\Delta)$,
there is no polynomial-time algorithm  for triangle-free $\Delta$-regular graphs 
to approximate the partition function within a factor $2^{\gamma n}$. 

Weitz's algorithm was extremely influential: many works have built upon his algorithmic 
approach to establish efficient algorithms for a variety of problems
(e.g., \cite{RSTVY,SST,LLY12,LLY13,SSY13,VVY,LLL14,SSSY,LL15}).    
One of its key conceptual contributions was showing how  decay of correlations properties
on a $\Delta$-regular tree are connected to the existence of an efficient algorithm
for graphs of maximum degree $\Delta$.  
We believe our paper enhances this insight by connecting these same decay of correlations 
properties on a $\Delta$-regular tree to the analysis of widely-used  Markov Chain Monte Carlo 
(MCMC) and message passing  algorithms.  
%

\subsection*{Main Results}

As mentioned briefly earlier on, there are two widely-used approaches
for the associated approximate counting/sampling problems, namely
MCMC and message passing approaches.
A popular MCMC algorithm is the simple single-site update Markov chain known 
as the Glauber dynamics.  
The Glauber dynamics is a Markov chain $(X_t)$ on $\Omega$ whose
transitions $X_t\rightarrow X_{t+1}$ are defined by the following process:
\begin{enumerate}
\item Choose $v$ uniformly at random from $V$.
\item If $N(v)\cap X_t = \emptyset$ then let 
\[
X_{t+1} = \begin{cases}  
X_t\cup\{v\} & \mbox{ with probability } \lambda/(1+\lambda) \\
X_t\setminus\{v\} & \mbox{ with probability } 1/(1+\lambda) 
\end{cases}
\]
\item If $N(v)\cap X_t \neq \emptyset$ then let $X_{t+1} = X_t$.
\end{enumerate}

The mixing time $\Tmix$ is the number of steps to guarantee that the
chain is within a specified (total) variation distance of the stationary distribution.
In other words, for $\eps>0$,
\[   \Tmix(\eps) = \min\{ t : \mbox{ for all } X_0, \dtv(X_t,\mu) \leq \eps\},
\]
where $\dtv()$ is the variation distance.
We use $\Tmix=\Tmix(1/4)$ to refer to the mixing time for $\eps=1/4$.

It is natural to conjecture that the Glauber dynamics has mixing time 
$O(n\log{n})$ for all $\lambda<\lambda_c(\Delta)$.
Indeed, Weitz's work implies rapid mixing for $\lambda<\lambda_c(\Delta)$ 
for amenable graphs.
On the other hand Mossel et al. in \cite{MWW}  show slow mixing when $\lambda>\lambda_c(\Delta)$
on random regular bipartite graphs.
The previously best known results for MCMC algorithms are far from
reaching the critical point.  It was known that the mixing time 
of the Glauber dynamics (and other simple, local Markov chains)
is $O(n\log{n})$ when $\lambda<2/(\Delta-2)$ for any graph with maximum degree $\Delta$
\cite{DG,LV,V}.
In addition, \cite{HV:stationarity} analyzed $\Delta$-regular graphs 
with $\Delta=\Omega(\log{n})$ and presented a polynomial-time
simulated annealing algorithm
when $\lambda<\lambda_c(\Delta)$. 

Here we prove $O(n\log{n})$ mixing time up to the critical point 
when the maximum degree is at least a sufficiently large constant $\Delta_0$,
and there are no cycles of length $\leq 6$ (i.e., girth $\geq 7$).

\begin{theorem}\label{thrm:RapidMixingMain}
For all $\delta>0$, there exists $\Delta_0=\Delta_0(\delta)$ and $C=C(\delta)$,
for all graphs $G=(V,E)$ of maximum degree $\Delta\geq\Delta_0$ and girth $\geq 7$, 
all $\lambda<(1-\delta)\lambda_c(\Delta)$, all $\eps>0$, 
the mixing time of the Glauber dynamics satisfies:
\[  \Tmix(\eps) \leq Cn\log(n/\eps).
\]
\end{theorem}

\noindent
Note that $\Delta$ and $\lambda$ can be a function of $n=|V|$. 
The above sampling result yields  (via \cite{SVV,Huber})  an $\fpras$ for estimating the  partition function 
$Z$ with running time $O^*(n^2)$  where $O^*()$  hides multiplicative $\log{n}$ factors.  
The algorithm of Weitz \cite{Weitz} is polynomial-time for  small constant $\Delta$,  in contrast our algorithm 
is polynomial-time for all $\Delta>\Delta_0$  for a sufficiently large constant $\Delta_0$.

A family of graphs of particular interest are random $\Delta$-regular graphs and
random $\Delta$-regular {\em bipartite} graphs. These graphs do not satisfy the girth requirements of 
Theorem \ref{thrm:RapidMixingMain} but they have few short cycles.  Hence, as one would expect the above result
extends to these graphs.

\begin{theorem}\label{thrm:RapidMixingRandom}
For all $\delta>0$, there exists $\Delta_0=\Delta_0(\delta)$ and $C=C(\delta)$,
for all $\Delta\geq\Delta_0$, all $\lambda<(1-\delta)\lambda_c(\Delta)$, all $\eps>0$,
with probability $1-o(1)$ over the choice of an 
$n$-vertex graph $G$ chosen uniformly at random from the set of all $\Delta$-regular 
(bipartite) graphs,  the mixing time of the Glauber dynamics on $G$ satisfies:
\[  \Tmix(\eps) \leq Cn\log(n/\eps).
\]
\end{theorem}

\noindent
Theorem \ref{thrm:RapidMixingRandom}  complements the work in  \cite{MWW} which shows slow
mixing for random $\Delta$-regular bipartite graphs when
 $\lambda>\lambda_c(\Delta)$.

The other widely used approach is BP (belief propagation) based algorithms.
BP, introduced by Pearl \cite{Pearl}, 
is a simple recursive scheme designed on trees to correctly
compute the marginal distribution for each vertex to be occupied/unoccupied.  
In particular, consider a rooted tree $T=(V,E)$ where for $v\in V$ its parent is
denoted as $p$ and its children are $N(v)$.  Let
\[ 
 q(v) = \ProbCondSub{\mu}{\mbox{$v$ is occupied}}{\mbox{$p$ is unoccupied}}
\] denote the probability in the Gibbs distribution 
that $v$ is occupied conditional on its parent $p$ being unoccupied.
It is convenient to work with ratios of the marginals, and 
hence let $R_{v\rightarrow p(v)} = q(v)/(1-q(v))$ denote the ratio of the
occupied to unoccupied marginal probabilities.
Because $T$ is a tree then it is not difficult to show that
this ratio satisfies the following recurrence:
\[  R_{v\rightarrow p(v)} = \lambda\prod_{w\in N(v)\backslash\{p(v)\}} \frac{1}{1+R_{w\rightarrow v}}.
\]
This recurrence explains the terminology of BP that $R_{w\rightarrow v}$ is a
``message'' from $w$ to its parent $v$.  Given the messages to $v$ 
from all of its children then $v$ can send its message to its parent.  
Finally the root $r$ (with a parent $p$ always fixed to be unoccupied and thus removed)
can compute the marginal probability that it is occupied by:
$q(r) = R_{r\rightarrow p}/(1+R_{r\rightarrow p})$.  

The above formulation defines (the sum-product version of) BP  a simple, natural algorithm
which works efficiently and correctly for trees.
For general graphs {\em loopy BP} implements the above approach, even though there
are now cycles and so the algorithm no longer is guaranteed to work correctly.
For a graph $G=(V,E)$, for $v\in V$ let $N(v)$ denote the set of all neighbors of $v$.
For each $p\in N(v)$ and time $t\geq 0$ we define a message 
\[
R^t_{v\rightarrow p} = \lambda\prod_{w\in N(v)\backslash\{p\}} \frac{1}{1+R^{t-1}_{w\rightarrow v}}.
\]
The corresponding estimate of the marginal can be computed from the messages by: 
\begin{equation} \label{eq:DefQVP(t)}
 q^t(v,p)=\frac{R^t_{v\to p}}{1 + R^t_{v\to p}}.
\end{equation}

Loopy BP is a popular algorithm for estimating marginal probabilities
in general graphical models (e.g., see \cite{MWJ}), 
but there are few results on when loopy BP converges to the Gibbs distribution (e.g., Weiss
\cite{Weiss} analyzed graphs with one cycle, and \cite{TJ,Heskes,IFW} presented 
various sufficient conditions, see also \cite{CCGSS,Shin} for analysis of BP variants).
We have an approach for analyzing
loopy BP and in this project we will prove that loopy BP works well in a broad range of
parameters.  Its behavior relates to phase transitions in the underlying model,
we detail our approach and expected results after formally presenting phase transitions.

We prove that, on any graph with girth $\geq 6$ and maximum degree
$\Delta\geq\Delta_0$ where $\Delta_0$ is a sufficiently large
constant, loopy BP quickly converges to the (marginals of) Gibbs distribution $\mu$.
More precisely, $O(1)$ iterations of loopy BP suffices, note each iteration
of BP takes $O(n+m)$ time where $n=|V|$ and $m=|E|$.

\begin{theorem}\label{thrm:CnvrgLoopyBP2Correct}
For all $\delta, \epsilon>0$, there exists $\Delta_0=\Delta_0(\delta,\epsilon)$ and 
$C=C(\delta,\epsilon)$,
for all graphs $G=(V,E)$ of maximum degree $\Delta\geq\Delta_0$ and girth $\geq 6$,
all $\lambda<(1-\delta)\lambda_c(\Delta)$,
the following holds:  for $t\geq C$, for all  $v\in V$,  $p\in N(v)$,
\[  
  \left| \frac{q^t(v,p)}{\mu(\textrm{$v$ is occupied} \;|\; \textrm{$p$ is unoccupied}) } -1 \right| \leq \epsilon
   \]
where $\mu(\cdot)$ is the Gibbs distribution.
\end{theorem}

\subsection*{Contributions}

Our main conceptual contribution is formally connecting 
the behavior of BP and the Glauber dynamics.
We will analyze the Glauber dynamics using path coupling \cite{BD}.  
In path coupling we need to analyze a pair of {\em neighboring configurations}, in
our setting this
is a pair of 
independent sets $X_t,Y_t$
which differ at exactly one vertex $v$.
The key is to construct a one-step coupling $(X_t,Y_t)\rightarrow (X_{t+1},Y_{t+1})$
and introduce a distance function $\Phi:\Omega\times\Omega\rightarrow \R_{\geq 0}$ which
``contracts'' meaning that the following {\em path coupling
condition} holds for some $\gamma>0$:
\[ \ExpCond{\Phi(X_{t+1},Y_{t+1})}{X_t,Y_t} \leq  (1-\gamma)\Phi(X_t,Y_t).
\]  

We use a simple maximal one-step coupling and hence in our setting
the path coupling condition simplifies to:
\[
(1-\gamma)\Phi(X_t,Y_t) \geq \sum_{z\in N(v)} \frac{\lambda}{1+\lambda}\indicator{z \mbox{ is unblocked in }X_t}\Phi(z),
\]
where {\em unblocked}  means that $N(z)\cap X_t = \emptyset$, i.e., 
all neighbors of $z$ are unoccupied, and we have assumed there are
no triangles so as to ignore the possibility that $X_t$ and $Y_t$ differ 
on the neighborhood of $z$. 

The distance function $\Phi$ must satisfy a few basic
conditions such as being a path metric, and 
if $X\neq Y$ then $\Phi(X,Y)\geq 1$ (so that by Markov's inequality $\Prob{X_t\neq Y_t} \leq \Exp{\Phi(X_t,Y_t)}$).
A standard choice for the distance function is the 
Hamming distance.  In our setting the Hamming distance does not suffice and our
primary challenge is determining a suitable distance function.

We cannot construct a suitable distance function which satisfies the path
coupling condition for arbitrary neighboring pairs $X_t,Y_t$.  But, a key
insight is that we can
show the existence of a suitable $\Phi$ when the local neighborhood of the 
disagreement $v$ behaves like the BP fixpoint. 
Our construction of this $\Phi$ is quite intriguing.  

In our proofs it is useful to consider the (unrooted) BP recurrences corresponding
to the probability that a vertex is unblocked.
This corresponds to the following function  $F:[0,1]^V\to [0,1]^V$ which is defined as follows,
for any $\omega\in [0,1]^V$ and $z\in V$:
\begin{equation}\label{eq:BP-recursion}
F(\omega)(z)=\prod_{y\in N(z)}\frac{1}{1+\lambda \omega(y)}. 
\end{equation}
Also, for some integer $i\geq 0$, let $F^i(\omega):[0,1]^V\to [0,1]^V$ be
the $i$-iterate of $F$.  This recurrence is closely related to the standard
BP operator $R()$ and hence under the hypotheses of our main results, 
we have that $F()$ has a unique fixed point $\omega^*$,
and for any $\omega$, all $z\in V$, $\lim_{i\rightarrow\infty}F^i(z) = \omega^*(z)$.

To construct the distance function $\Phi$ we 
start with the Jacobian of this BP operator $F()$.
By a suitable matrix diagonalization we obtain the path coupling condition.
Since $F()$ converges to a fixed point, and, in fact, it contracts at every level with respect to
an appropriately defined potential function, we then know that 
the Jacobian of the BP operator $F()$ evaluated at its fixed point $\omega^*$
has spectral radius $< 1$ and hence the same holds for the path coupling
condition for pairs of states that are BP fixed points.  This yields a function 
$\Phi$ that satisfies the following system of inequalities
\begin{equation}\label{eq:EigengBPCondA}
\Phi(v) >  \sum_{z\in N(v)}\frac{\lambda\omega^*(z)}{1+\lambda\omega^*(z)}\Phi(z).
\end{equation}
However  for the  path coupling condition a stronger version of the above is necessary. 
More specifically,  the sum 
on the r.h.s. should be appropriately bounded away from $\Phi(v)$,
i.e. we need to have
\[
(1-\gamma)\Phi(v) >  \sum_{z\in N(v)}\frac{\lambda\omega^*(z)}{1+\lambda\omega^*(z)}\Phi(z).
\]
Additionally, $\Phi$ should  be a distance  metric, e.g. $\Phi>0$.
It turns out that we use further properties of the distance function 
$\Phi$, hence we  need to explicitly derive a $\Phi$.

There are previous works \cite{Hayes-radius,HVV}
which utilize the spectral radius 
of the adjacency matrix of the input graph $G$ to design a suitable distance
function for path coupling.  In contrast, we use
insights from the analysis of the BP operator to derive
a suitable distance function. 
 We believe this is a richer connection that can potentially
lead to stronger results since it directly relates to convergence properties on the
tree.   Our approach has the potential to apply for a more general class of 
spin systems, we comment on this in more detail in the conclusions.

The above argument only implies that we have contraction in the path coupling condition
for pairs of configurations which are BP fixed points.  
A priori we don't even know if the BP fixed points on the tree
correspond to the Gibbs distribution on the input graph.  
We prove that the Glauber dynamics (approximately) satisfies a
recurrence that is close to the BP recurrence; this builds upon ideas
of Hayes \cite{Hayes} for colorings.   This argument requires that there are no
cycles of length $\leq 6$ for the Glauber dynamics (and no cycles of length $\leq 5$
for the direct analysis of the Gibbs distribution).  Some local sparsity
condition is necessary since if there are many short cycles
then the Gibbs distribution no longer behaves similarly to a tree and hence loopy BP may be
a poor estimator.  

As a consequence of the above relation between BP and the Glauber dynamics,
we establish that
from an arbitrary initial configuration $X_0$, 
after a short burn-in period of $T=O(n\log{\Delta})$ steps of the Glauber dynamics 
the configuration $X_T$ is a close approximation to the BP fixed point.
In particular, for any vertex $v$, the number of unblocked neighbors of $v$ in $X_T$
is $\approx\sum_{z\in N(v)} \omega^*(z)$ with high probability.
As is standard for concentration results, our proof of this result necessitates that
$\Delta$ is at least a sufficiently large constant. 
Finally we adapt ideas of \cite{DFHV} to utilize these burn-in properties and
establish rapid mixing of the Glauber dynamics.

\subsection*{Outline of Paper}


In the following section we state results about the convergence
of the BP recurrences.  We then present in Section 
\ref{sec:DistanceFunction} our theorem showing the existence
of a suitable distance function for path coupling for pairs of states at 
the BP fixed point.  Section \ref{sec:LocalUniformity-sketch}
sketches the proofs for our local uniformity results that after a burn-in period
the Glauber dynamics behaves locally similar to the BP recurrences.
Finally, in Section \ref{sec:rapid-mixing-sketch}
we outline the proof of Theorem \ref{thrm:RapidMixingMain} of
rapid mixing for the Glauber dynamics.
The extension to random regular (bipartite) graphs as stated in Theorem \ref{thrm:RapidMixingRandom}
is proven in Section \ref{sec:thrm:RapidMixingRandom} of the appendix. 
Theorem \ref{thrm:CnvrgLoopyBP2Correct} about the efficiency of loopy BP is proven in 
Section \ref{sec:thrm:CnvrgLoopyBP2Correct}
of appendix, the key technical results in the proof are sketched in Section \ref{sec:LocalUniformity-sketch}.

The full proofs of our results are quite lengthy and so we defer many to the appendix.

\section{BP Convergence}\label{sec:BPConvergence}

Here we state several useful results about the convergence of 
BP to a unique fixed point, and stepwise contraction of BP to the fixed point.
The lemmas presented in this section are proved in 
Section \ref{sec:BP-proofs}
of the  appendix.

Our first lemma (which is proved using ideas from \cite{RSTVY,LLY13,SST})
says that the recurrence for $F()$ defined in \eqref{eq:BP-recursion} has
a unique fixed point.

\begin{lemma}\label{lemma:FixPointEquations}
For all $\delta>0$, there exists $\Delta_0=\Delta_0(\delta)$,
for all $G=(V,E)$ of maximum degree $\Delta\geq\Delta_0$,
all $\lambda<(1-\delta)\lambda_c(\Delta)$,
the function $F$ has a unique fixed point $\omega^*$.
\end{lemma}

A critical result for our approach is that the 
recurrences $F()$ have stepwise contraction to the fixed point $\omega^*$.
To obtain contraction we use the following
potential function $\Psi$.
Let the function  $\Psi:[0,1]\to \mathbb{R}_{\geq 0}$ 
be as follows,
\begin{align}
\Psi(x)=(\sqrt{\lambda})^{-1}\textrm{arcsinh}\left (\sqrt{\lambda \cdot x}\right). \label{eq:potential-function}
\end{align}

Our main motivation for introducing $\Psi$ is as a normalizing
potential function that we use to define the following
distance metric, $D$, on functions $\omega \in [0,1]^V$:
\[
D(\omega_1, \omega_2) = \max_{z \in V} \left| \Psi(\omega_1(z)) -
  \Psi(\omega_2(z)) \right|.
\]
We will also need a variant, $D_{v,R}$, of this metric whose value only depends on
the restriction of the function to a ball of radius $\ell$ around vertex
$v$.  For any $v\in V$, integer $\ell\geq 0$,
let $B(v,\ell)$ be the set of vertices within distance $\leq \ell$ of $v$.
Moreover, for functions $\omega_1, \omega_2 \in [0,1]^V$, we define: 
\begin{align}\label{eq:potential-metric}
D_{v,\ell}(\omega_1, \omega_2) = \max_{z \in B(v,\ell)} \left| \Psi(\omega_1(z)) - \Psi(\omega_2(z)) \right|.
\end{align}

We can now state the following convergence result for the recurrences, which
establishes stepwise contraction.

\begin{lemma}\label{lemma:ApproxFixPointEquations}
For  all $\delta>0$,  there exists $\Delta_0=\Delta_0(\delta)$,
for all $G=(V,E)$ of maximum degree $\Delta\geq\Delta_0$, all 
$\lambda<(1-\delta)\lambda_c(\Delta)$, for any $\omega\in
[0,1]^V$, $v \in V$ and $\ell \ge 1$, we have:
\[
D_{v,\ell-1}(F(\omega), \omega^*) \le (1 - \delta/6) D_{v,\ell}(\omega,
\omega^*).
\]
where $\omega^*$ is the fixed point of $F$.
\end{lemma}

\section{Path Coupling Distance Function}\label{sec:DistanceFunction}

We now prove that there exists a suitable distance function $\Phi$ for which
the path coupling condition holds for configurations that correspond to 
the fixed points of $F()$.

\begin{theorem}\label{thrm:EigenVector}
For  all $\delta>0$,  there exists $\Delta_0=\Delta_0(\delta)$,
for all $G=(V,E)$ of maximum degree $\Delta\geq\Delta_0$, all 
$\lambda<(1-\delta)\lambda_c(\Delta)$, 
there exists $\Phi:V\to \mathbb{R}_{\geq 0}$ such that  for every $v\in V$, 
\begin{equation}\label{eq:MaxPhi}
1\leq \Phi(v)\leq 12,
\end{equation}
and 
\begin{equation}
\label{eq:FoundPhi}
(1-\delta/6)\Phi(v) \ge\sum_{u\in N(v)}\frac{\lambda\omega^*(u)}{1+\lambda\omega^*(u)}\Phi(u),
\end{equation}
where $\omega^*$ is the fixed point of $F$ defined in \eqref{eq:BP-recursion}. 
\end{theorem}

\begin{proof}
We will prove here that the convergence of BP provides the existence of a 
 distance function $\Phi$ satisfying \eqref{eq:FoundPhi}.
 We defer the technical proof of~\eqref{eq:MaxPhi} to 
 Section~\ref{sec:BP-proofs}
 of the   appendix.

The Jacobian $J$ of the BP operator $F$ is given by
\begin{displaymath}
J({v, u})\;=\; \left | \frac{\partial F(\omega)(v)}{\partial \omega(u)} \right|\; =\; 
\left \{
\begin{array}{lcl}
\frac{\lambda F(\omega)(v)}{1+\lambda\omega(u)} &\quad & \textrm{if $u\in N_v$}\\
0 && \textrm{otherwise}
\end{array}
\right .
\end{displaymath}
Let $J^*=\left.J\right|_{\omega=\omega^*}$ denote the Jacobian at the fixed point $\omega=\omega^*$. Let $D$ be the diagonal matrix with $D(v,v)=\omega^*(v)$ and let $\hat{J}=D^{-1}J^* D$.

The path coupling condition~\eqref{eq:FoundPhi} is in fact 
\begin{equation}\label{eq:Path2SpectralRad}
\hat{J}\Phi\le(1-\delta/6)\Phi.
\end{equation}
The fact that $\omega^*$ is a Jacobian attractive fixpoint implies 
the existence of a nonnegative $\Phi$ with~$\hat{J}\Phi<\Phi$.
Thus, the theorem would
follow immediately if the spectral radius of $\hat{J}$ is $\rho(\hat{J})\le 1-\delta/6$ and $\hat{J}$ has a principal eigenvector with each entry from the bounded range $[1,12]$. However, explicitly calculating this principal eigenvector can be challenging on general graphs.

The convergence of BP which is established in Lemmas \ref{lemma:FixPointEquations}, \ref{lemma:ApproxFixPointEquations}, 
with respect to the potential function $\Psi$, guides us to an explicit construction of $\Phi$ such that $\hat{J}\Phi<\Phi$. 
Indeed, let  $\Psi'(x)=\frac{1}{2\sqrt{x(1+\lambda x)}}$  denote the derivative of the potential function $\Psi$.
It will follow from the proof of Lemma~\ref{lemma:ApproxFixPointEquations} that:
\[
\sum_{u\in N(v)}J^*(v,u)\frac{\Psi'(\omega^*(v))}{\Psi'(\omega^*(u))}\le 1-\delta/6.
\]
This inequality is due to the contraction of the BP system at the fixed point with respect to the potential function $\Psi$. It is equivalent to the following:
\[
\sum_{u\in N(v)}\frac{\hat{J}(v,u)}{\omega^*(u)\Psi'(\omega^*(u))}\le \frac{1-\delta/6}{\omega^*(v)\Psi'(\omega^*(v))}.
\]
Then, \eqref{eq:Path2SpectralRad} is trivially satisfied  by choosing $\Phi$ such that 
$\Phi(v)=\frac{1}{2\omega^*(v)\Psi'(\omega^*(v))}=\sqrt{\frac{1+\lambda \omega^*(v)}{\omega^*(v)}}$.
In turn we get the path coupling condition~\eqref{eq:FoundPhi}.  
The verification of \eqref{eq:MaxPhi} is in 
Section~\ref{sec:BP-proofs}
of the  appendix.

\end{proof}

\section{Local Uniformity for the Glauber Dynamics}
\label{sec:LocalUniformity-sketch}

We will prove that the Glauber dynamics, after a sufficient burn-in,
behaves with high probability locally similar to the BP fixed points.
In this section we will formally state some of these ``local uniformity'' results
and sketch the main ideas in their proof.   
The proofs are quite technical  and deferred to 
Section \ref{sec:UniformityProofs}
of the  appendix.

For an independent set $\sigma$, for $v\in V$, and $p\in N(v)$ let 
\begin{equation}\label{eq:DefIvp}
\I_{v,p}(\sigma)=\indicator{ \sigma\cap\left(N(v)\setminus\{p\}\right) = \emptyset}
\end{equation}
be the indicator of whether the children of $v$ leave $v$ unblocked.

We now state our main local uniformity results.
We first establish that the Gibbs distribution behaves as in the BP fixpoint, when the
girth~$\geq 6$.  We will prove that for any vertex $v$,
the number of unblocked neighbors of $v$ is $\approx\sum_{z\in N(v)} \omega^*(z)$ 
with high probability.
Hence, for $v\in V$ let 
\[  \SOS_X(v) =\sum_{z\in N(v)} \I_{z,v}(X),
\]
denote the number of unblocked neighbors of $v$ in configuration $X$.

\begin{theorem}\label{thrm:ConcentrNoOfBlockedStatic}
For all $\delta, \epsilon>0$, there exists $\Delta_0=\Delta_0(\delta,\epsilon)$ and $C=C(\delta,\epsilon)$,
for all graphs $G=(V,E)$ of maximum degree $\Delta\geq\Delta_0$ and girth $\geq 6$,
 all $\lambda<(1-\delta)\lambda_c(\Delta)$,
 for all $v\in V$, it holds that: 
\[  
\ProbSub{X\sim\mu}{
\left|  \SOS_X(v)-\sum_{z\in N(v)} \omega^*(z) \right| 
\leq \epsilon \Delta }  \geq 1-\exp\left(-\Delta/C\right),
  \]
where $\omega^*$ is the fixpoint from Lemma \ref{lemma:FixPointEquations}.
\end{theorem}

Theorem \ref{thrm:ConcentrNoOfBlockedStatic} will be the key
ingredient in the proof of Theorem \ref{thrm:CnvrgLoopyBP2Correct}
(to be precise, the upcoming Lemma \ref{lemma:RApproxReccGirth6}
is the key element in the proofs of 
Theorems \ref{thrm:CnvrgLoopyBP2Correct} and \ref{thrm:ConcentrNoOfBlockedStatic}).

For our rapid mixing result (Theorem \ref{thrm:RapidMixingRandom})
we need an analogous local uniformity result for the Glauber dynamics.
This will require the slightly higher girth requirement $\geq 7$
since the grandchildren of a vertex $v$ no longer have a certain
conditionally independence and we need the additional girth requirement
to derive an approximate version of the conditional independence
(this is discussed in more detail in  Section \ref{sec:GVsG*} of the  appendix).

The path coupling proof weights the vertices according to $\Phi$.
Hence, in place of $\SOS$ we need the following weighted version $\W$.
For $v\in V$  and $\Phi:V\to \mathbb{R}_{\geq 0}$ as defined in 
Theorem \ref{thrm:EigenVector} let
\begin{equation}
\label{eq:DefWXT}
\W_{\sigma}(v)=\sum_{z\in N(v)} \I_{z,v}(\sigma)\ \Phi(z).
\end{equation}

\noindent
We then prove that the Glauber dynamics, after sufficient burn-in, also
behaves as in the BP fixpoint with a slightly higher girth requirement $\geq 7$.
(For path coupling we only need an upper bound on the number of unblocked neighbors,
hence we state and prove this simpler form.)

\begin{theorem}\label{thrm:UniformityWithBurnIn}
For  all $\delta, \epsilon>0$,  let 
$\Delta_0=\Delta_0(\delta,\epsilon), C=C(\delta,\epsilon)$,
for all graphs $G=(V,E)$ of maximum degree $\Delta\geq\Delta_0$ and girth $\geq 7$,
 all $\lambda<(1-\delta)\lambda_c(\Delta)$,
let $(X_t)$ be the Glauber dynamics on the hard-core model. 
For all $v\in V$, it holds that
\begin{eqnarray} \label{eq:uniformity-simple}
\lefteqn{
\Prob{ 
(\forall t\in {\cal I} )  \quad
  \W_{X_t}(v)<\sum_{z \in N(v)} \omega^*(z)\Phi(z)+\epsilon \Delta }
 } \hspace{2in}
\\
  &\geq& 1-\exp\left( -\Delta/ C \right),\nonumber
 \end{eqnarray}
where the time interval  ${\cal I}=[Cn\log\Delta,n\exp\left( \Delta/C \right)]$.
\end{theorem}

\subsection{Proof sketch for local uniformity results}\label{sec:SketchLocalUnifGibbs}

Here we sketch the simpler proof of Theorem \ref{thrm:ConcentrNoOfBlockedStatic}
of the local uniformity results for the Gibbs distribution. 
This will illustrate the main conceptual ideas in the proof for the Gibbs distribution,
and we will indicate the extra challenge for the analysis of the
Glauber dynamics in the proof of 
Theorem \ref{thrm:UniformityWithBurnIn}.
 The full proofs for Theorems \ref{thrm:ConcentrNoOfBlockedStatic}
and \ref{thrm:UniformityWithBurnIn} are in 
Section \ref{sec:UniformityProofs}
of the  appendix.

Consider a graph $G=(V,E)$.  
For a vertex $v$ and an independent set $\sigma$, consider the following quantity:
\begin{equation}\label{SketchEq:RwDefinition}
\R({\sigma}, v)=\prod_{z\in N(v)}\left(1-\frac{\lambda}{1+\lambda} \I_{z,v}({\sigma}) \right),
\end{equation}
where $\I_{z.v}({\sigma})$ is defined in \eqref{eq:DefIvp} (it is the indicator that the
children of $z$ leave it unblocked).
The important aspect of this quantity $\R$ is the following qualitative interpretation. 
Let $Y$ be distributed as in the Gibbs measure w.r.t. $G$.
For triangle-free $G$ we have
\begin{eqnarray}
\lefteqn{
\R({\sigma},v) 
} \hspace{.1in} \nonumber
\\
&=&\ProbCondSub{}{\textrm{$v$ is unblocked} }{v\notin Y ,\; 
Y(S_2(v))=\sigma(S_2(v))}, \nonumber
\end{eqnarray}
where $S_2(z)$ are those vertices distance $2$ from $z$ and by ``$z\notin {\sigma}$" we mean 
that $z$ is not occupied.
Moreover, conditional on the configuration at $z$ and $S_2(z)$ the neighbors of $z$
are independent in the Gibbs distribution and hence:
\begin{eqnarray} \label{eq:triangle-free} 
\lefteqn{
\R({\sigma},v) 
} \hspace{.1in}  \\
&=&\prod_{z\in N(v)} \ProbCondSub{}{ z\notin Y}{v\notin Y,\; 
Y(S_2(v))=\sigma(S_2(v))}. \nonumber
\end{eqnarray}
In the special case where the underlying graph is a {\em tree} we can
extend \eqref{eq:triangle-free} to the following recursive equations:
Let $X$ be distributed as in $\mu$. We have that
\begin{equation}\label{SketchEq:BP-RwDefinition}
\R (X, v)=\prod_{z\in N(v)}\left(1-\frac{\lambda}{1+\lambda} \R(X,z) \right) +O(1/\Delta),
\end{equation}

\noindent
For our purpose  it turns out that  $\R(X,\cdot )$ is an approximate version of $F()$ defined in 
\eqref{eq:BP-recursion}.  The error term $O(1/\Delta)$ in \eqref{SketchEq:BP-RwDefinition} 
is negligible. For understanding $\R(X,\cdot)$ qualitatively,  this error term can 
 be completely ignored.

Consider the (BP system of) equations in \eqref{SketchEq:BP-RwDefinition},
which is exact on trees.   Nothing prevents us from applying  \eqref{SketchEq:BP-RwDefinition}
on the graph $G$ and get the loopy version of the equations.
Now, 
\eqref {SketchEq:BP-RwDefinition} does not necessarily compute the probability for 
$v$ to be unblocked. 
However, we show the following interesting result regarding the
quantity $\SOS_{X}(v)$, for every $v\in V$. 
With probability  $\geq 1-\exp\left( -\Omega(\Delta)\right)$, it holds that
\begin{equation}\label{SketchEd:WtVsSumR}
\left| \SOS_{X}(v) - \sum_{z \in N(v)} \R(X,z) \right| \leq  \epsilon\Delta.
\end{equation}
That is, we can approximate $\SOS_{X}(v)$ by using quantities that arise from
the loopy BP equations.  Still, getting  a handle on $\R(X,z)$ in \eqref{SketchEd:WtVsSumR} 
is a non-trivial task.  To this end,  we show that $X\sim\mu$ 
satisfies \eqref{SketchEq:BP-RwDefinition} in the following approximate sense:
\begin{lemma}\label{lemma:RApproxReccGirth6} 
For all $\gamma,\delta>0$, there exists $\Delta_0, C>0$, for all graphs 
$G=(V,E)$ of maximum degree $\Delta\geq\Delta_0$
and girth $\geq 6$ all $\lambda<(1-\delta)\lambda_c(\Delta)$
for all $v\in V$ the following is true: 

Let $X$ be distributed as in $\mu$. Then  with probability   $\geq 1- \exp\left ( - \Delta/C \right)$ it holds that
\begin{eqnarray}\label{SketchEq:From:lemma:ApproxFixPoint}
\left | 
\R(X, v) - \prod_{z\in N(v)}\left(1 -\frac{\lambda}{1+\lambda} 
\R(X, z) 
\right)
\right |  < \gamma.
\end{eqnarray}
\end{lemma}

We will argue (via \eqref{SketchEq:From:lemma:ApproxFixPoint}) that
$\R()$ is an approximate version of $F()$ and then we can apply 
Lemma \ref{lemma:ApproxFixPointEquations} to deduce convergence
(close) to the fixpoint $\omega^*$.
Consequently, we will prove that
for every $v\in V$, with probability
at least $1-\exp\left( -\Omega(\Delta) \right)$, 
 it holds that
\begin{equation}\label{SketchEq:ApproxFix}
\left| \R(X,v) - \omega^*(v) \right|\leq \epsilon.
\end{equation}
(See 
Lemma  \ref{lemma:RConGirth6}
 in 
 Section \ref{sec:lemma:CnvrgLoopyBP2CorrectNoParent} 
 of the   appendix for a formal statement.)
Combining \eqref{SketchEq:ApproxFix} and \eqref{SketchEd:WtVsSumR} will
finish the proof of Theorem \ref{thrm:ConcentrNoOfBlockedStatic}.
For the detailed proof of Theorem \ref{thrm:ConcentrNoOfBlockedStatic}
see 
Section \ref{sec:UniformityProofs} 
in the   appendix.

\subsection{Approximate recurrence - Proof of Lemma \ref{lemma:RApproxReccGirth6}}

Here we prove Lemma \ref{lemma:RApproxReccGirth6} which shows that
$\R$ satisfies an approximate recurrence similar to loopy BP, this is the main result
in the proof of Theorem  \ref{thrm:ConcentrNoOfBlockedStatic}.
Before beginning the proof we illustrate the necessity of the girth assumption.

Recall that for triangle-free graphs we have conditional
independence in \eqref{eq:triangle-free} for the neighbors of vertex $z$.
In \eqref{SketchEd:WtVsSumR} we need to consider $\sum_{z\in N(v)} \R(X,z)$.
To get independence on the grandchildren of $v$ we need to condition on $S_3(v)$,
this will require girth $\geq 6$, see \eqref{eq:Tail4QCX^*TGirth6} below.

\begin{proof}[Proof of Lemma \ref{lemma:RApproxReccGirth6}]
Consider  $X$  distributed as in $\mu$.   Given some vertex $v\in V$, let ${\cal F}$
be the $\sigma$-algebra generated by the configuration of $v$ and the vertices at
distance $\geq 3$ from $v$.

Note that $\lambda_c(\Delta)\sim e/\Delta$. So, for $\lambda<\lambda_c(\Delta)$
and $\Delta>\Delta_0$ we have $\lambda=O(1/\Delta)$.

Note that $\SOS_{X}(v)$ is a function of the configuration at 
$S_2(v)$.  Conditional on ${\cal F}$, for any $z,z'\in N(v)$ the configurations at $N(z)\backslash\{v\}$ 
and $N(z')\backslash\{v\}$ are independent with each other.
That is, conditional on ${\cal F}$, the quantity $\SOS_{X}(v)$ is  a sum of $|N(v)|$ many independent random 
variables in $\{0,1\}$. 
Then, applying Azuma's inequality (the Lipschitz constant is $1$) we get that
\begin{equation}\label{eq:Tail4QCX^*TGirth6}
\Prob{
|\ExpCond{ \SOS_{X}(v) }{\cal F}-\SOS_{X}(v)| \leq \beta\Delta
} \geq 1-2\exp\left( -{\beta^2}\Delta/2 \right),
\end{equation}
for any $\beta>0$.

For $x\in \mathbb{R}_{\geq 0}$, let  $f(x)=\exp\left( -\frac{\lambda}{1+\lambda}x\right)$.
Since $\lambda\leq e/\Delta$ for $\Delta\geq \Delta_0$, then for  $|\gamma| \leq (3e)^{-1}$
  it holds that
$f(x+\gamma\Delta)\leq 10\gamma$.
Using these observations and \eqref{eq:Tail4QCX^*TGirth6} we get the following:
for $0<\beta<(3e)^{-1}$ it holds that
\begin{eqnarray}\label{eq:Tail4RX^*TGirth6}
\lefteqn{
\Prob{
\left | f(\SOS_X(v)) -  f(\ExpCond{ \SOS_{X}(v) }{\cal F})
\right | \leq 10\beta 
} 
} \hspace{1.5in} 
\\
&\geq &1-2\exp\left( -\beta^2 \Delta/2 \right). \nonumber
\end{eqnarray}

\noindent
Recalling the definition of
$\R(X,v)$, we have that 
\begin{eqnarray}
\R(X,v)&= & \prod_{z\in N(v)}\left(1-\frac{\lambda}{1+\lambda} \I_{z,v}(X) \right)
\nonumber \\
&=& \exp\left( - \frac{\lambda}{1+\lambda}\sum_{z\in N(v)}  \I_{z,v}(X)   +
O\left(1/{\Delta}\right) \right)  \nonumber\\
 &=&  f(\SOS_{X}(v)) +O\left(1/{\Delta}\right), \label{eq:RXtVsExpUsGirth6} 
\end{eqnarray}
where the second equality we use the fact that $\lambda = O(1/\Delta)$
and that for $|x|<1$ we have $1+x=\exp(x+O(x^2))$;
the last equality follows by noting that $f(\SOS_{X}(v))\leq 1$.

We are now   going to show  that for every $z\in N(v)$ it holds that 
\begin{equation}\label{eq:UzxVsRzStatic}
\left | \ExpCond{\I_{z,v}(X)}{\cal F }-\R(X,z) \right| \leq 2\lambda.
\end{equation}
Before showing that \eqref{eq:UzxVsRzStatic} is indeed correct, let us show
how  we use it to get the lemma.  

We have that 
\begin{eqnarray}
\lefteqn{
f(\ExpCond{ \SOS_{X}(v) }{\cal F})} 
\nonumber  \hspace{.2in}
\\
&=&\exp\left(-\frac{\lambda}{1+\lambda}\sum_{z\in N(v)} \ExpCond{\I_{z,v}(X_t)}{\cal F} \right)
\nonumber \\
&=&\exp\left(-\frac{\lambda}{1+\lambda}\sum_{z\in N(v)} \R(X,z) \right)+O(1/\Delta),   
\qquad 
\label{eq:RzChildrenStatic}
\end{eqnarray}
where in the first derivation we use  linearity of expectation and in the second derivation we
use \eqref{eq:UzxVsRzStatic} and the fact that $\lambda=O(1/\Delta)$.

The lemma follows by plugging \eqref{eq:RzChildrenStatic} and  \eqref{eq:RXtVsExpUsGirth6} into \eqref{eq:Tail4RX^*TGirth6}
and taking sufficiently large $\Delta$.

It remains to show \eqref{eq:UzxVsRzStatic}.  We first get an appropriate upper bound for 
$\ExpCond{\I_{z,v}(X)}{\cal F }$. Using the fact that $\I_{z,w}(X)\leq 1$ and
$\Prob{\textrm{$z\in X$} | {\cal F}} \leq \lambda$ we have that
\begin{eqnarray}
\lefteqn{
\ExpCond{\I_{z,v}(X)}{\cal F }
} \hspace{.4cm}\nonumber \\
&=&\ExpCond{\I_{z,v}(X)}{\cal F,  \textrm{$z\in X$ } }\cdot\Prob{\textrm{$z\in X$} | {\cal F}}
\nonumber \\ &&+\ExpCond{\I_{z,v}(X)}{\cal F,  \textrm{$z\notin X$} }\cdot\Prob{\textrm{$z\notin X$ } | {\cal F}} \nonumber \\
&\leq&\Prob{\textrm{$z\in X$} | {\cal F}}+\ExpCond{\I_{z,v}(X)}{\cal F,  \textrm{$z\notin X$} } 
\nonumber \\
&\leq&\lambda+ \ExpCond{\I_{z,v}(X)}{\cal F,  \textrm{$z\notin X$} } 
\nonumber \\
&=&\lambda+\prod_{u\in N(z)\setminus\{v\}} \left(1-\frac{\lambda}{1+\lambda}\I_{u,z}(X) \right) 
\label{eq:RandomNotSoRandom}\\
&\leq& 2\lambda+\prod_{u\in N(z)} \left(1-\frac{\lambda}{1+\lambda}\I_{u,z}(X) \right)\nonumber\\
&=&2\lambda+\R(X,z),  \label{eq:ExUVsRUpperGirth6}
\end{eqnarray}
where \eqref{eq:RandomNotSoRandom} uses the fact that 
 given ${\cal F}$ the values of $\I_{u,z}(X)$, for $u\in N(z)\setminus\{v\}$ are fully determined.
Similarly, we get the  lower bound:
\begin{eqnarray}
\ExpCond{\I_{z,v}(X)}{\cal F }&=&\ExpCond{\I_{z,v}(X)}{\cal F,  \textrm{$z\in X$} }\cdot\Prob{\textrm{$z\in X$} | {\cal F}}
\nonumber \\ &&+\ExpCond{\I_{z,v}(X)}{\cal F,  \textrm{$z\notin X$} }\cdot\Prob{\textrm{$z\notin X$} | {\cal F}} \nonumber \\
&\geq&(1-2\lambda)\ExpCond{\I_{z,v}(X)}{\cal F,  \textrm{$z\notin X$} }  \nonumber \\
&\geq& (1-2\lambda) \prod_{u\in N(z)\setminus\{w\}} \left(1-\frac{\lambda}{1+\lambda}\I_{u,z}(X) \right)\nonumber\\
&\geq&(1-2\lambda)\prod_{u\in N(z)} \left(1-\frac{\lambda}{1+\lambda}\I_{u,z}(X) \right)\nonumber \\
&=&(1-2\lambda)\R(X,z)  \nonumber \\
&\geq & \R(X,z)-2\lambda,
\label{eq:ExUVsRLowerGirth6}
\end{eqnarray}
where in the last inequality we use the fact that $\R(X,z)\leq 1$.

From \eqref{eq:ExUVsRUpperGirth6} and \eqref{eq:ExUVsRLowerGirth6} we have
proven \eqref{eq:UzxVsRzStatic}, which completes the proof of the lemma.
\end{proof}

\section{Sketch of Rapid Mixing Proof}
\label{sec:rapid-mixing-sketch}

Theorem \ref{thrm:UniformityWithBurnIn} 
tells us that after a burn-in period the Glauber dynamics
locally behaves like the BP fixpoints $\omega^*$ with high probability (whp).
(In this discussion, we use the term whp to refer to events that occur with 
probability $\geq 1 - \exp(-\Omega(\Delta))$. ) 
Meanwhile Theorem \ref{thrm:EigenVector} says
that there is an appropriate distance function $\Phi$ for which path coupling 
has contraction for pairs of states that behave as in $\omega^*$.
The snag in simply combining this pair of results and deducing rapid mixing is
that when $\Delta$ is constant then 
 there is still a constant
fraction of the graph that does not behave like $\omega^*$, and our disagreements
in our coupling proof may be biased towards this set.  
We follow the approach in \cite{DFHV} to overcome this obstacle and complete
the proof of Theorem \ref{thrm:RapidMixingMain}.
We give a brief sketch of the approach, the details are contained in 
Section  \ref{sec:thrm:RapidMixingMain} of the appendix.

The burn-in period for Theorem \ref{thrm:UniformityWithBurnIn} to apply is $O(n\log{\Delta})$
steps from the worst-case initial configuration $X_0$.  In fact, for a ``typical'' initial configuration
only $O(n)$ steps are required as we only need to update $\geq 1-\eps$ fraction of the neighbors 
of every vertex in the local neighborhood of the specified vertex $v$.  The ``bad'' initial configurations
are ones where almost all of the neighbors of $v$ (or many of its grandchildren) are occupied.  
We call such configurations ``heavy''  (see  Section~\ref{sec:heavy} of the  appendix). 
We first prove that after $O(n\log{\Delta})$ steps a chain is not-heavy in the local neighborhood of $v$, 
and this property persists whp (see 
Lemma \ref{lemma:BurnIn} in the appendix).   Then,  only $O(n)$ steps are required for the burn-in period (see  
Theorem \ref{thrm:Uniformity}  in  Section \ref{sec:UniformityProofs} of the appendix).

Our argument has two stages.  We start with a pair of chains $X_0,Y_0$
that differ at a single vertex~$v$.  In the first stage we burn-in for $T_b=O(n\log{\Delta})$ steps.
After this burn-in period, we have the following properties whp: every vertex in the local neighborhood of $v$ is not-heavy,
the number of disagreements is $\leq\poly(\Delta)$, and the disagreements are all in the
local neighborhood of $v$  (see  Lemma \ref{lemma:ArbitraryPair},  parts 2 and 4,  in 
Section~\ref{sec:thrm:RapidMixingMain} of the   appendix).

In the second stage we have sets of epochs of length $T=O(n)$ steps.  For the 
pair of chains $X_{T_b},Y_{T_b}$ we apply path coupling again.  Now we consider
a pair of chains that differ at one vertex $z$ which is not heavy.   We look again at the local neighborhood of $z$
(in this case, that means all vertices within distance $\leq\sqrt{\Delta}$ of $z$).
After $T$ steps, whp every vertex in the local neighborhood has the local uniformity properties
and the disagreements are contained in this local neighborhood.   Then we have 
contraction in the path coupling condition (by applying Theorem \ref{thrm:EigenVector}),
and hence after $O(n)$ further steps the expected Hamming distance is small
 (see  Lemma \ref{lemma:DistanceFromLight} in the appendix).  Combining a sequence of these 
 $O(n)$ length epochs we get that the original pair has is likely to have coupled and we can deduce rapid
mixing.

\section{Conclusions}

The work of Weitz \cite{Weitz} was a notable accomplishment in the field
of approximate counting/sampling.  However a limitation of his approach
is that the running time depends
exponentially on $\log{\Delta}$.
It is widely believed that the Glauber dynamics has mixing time $O(n\log{n})$
for all $G$ of maximum degree $\Delta$ when 
$\lambda<\lambda_c(\Delta)$.
However, until now there was little theoretical work to support this conjecture.
We give the first such results which analyze the widely used algorithmic
approaches of MCMC and loopy BP.

One appealing feature of our work is that it directly ties together with Weitz's
approach:   Weitz uses decay of correlations on trees to truncate his
self-avoiding walk tree, whereas we use decay of correlations to deduce
a contracting metric for the path coupling analysis, at least when the chains
are at the BP fixed point.  We believe this technique of utilizing the 
principal eigenvector for the BP operator for the path coupling metric will
apply to a general class of spin systems, such as 2-spin antiferromagnetic
spin systems (Weitz's algorithm was extended to this class \cite{LLY13}).

We hope that in the future more refined analysis of the local uniformity
properties will lead to relaxed girth assumptions.  
However dealing with very short cycles, such as triangles, will require
a new approach since loopy BP no longer seems to be a good estimator
of the Gibbs distribution for certain examples.

\newpage

\appendix

\section{BP convergence: Missing proofs in Sections~\ref{sec:BPConvergence}~and~\ref{sec:DistanceFunction}}
\label{sec:BP-proofs}

In this section we prove Lemma~\ref{lemma:FixPointEquations} about the convergence of recurrence $F$ defined in~\eqref{eq:BP-recursion} to a unique fixed point $\omega^*$, Lemma~\ref{lemma:ApproxFixPointEquations} about the contraction of the error at every step with respect to a potential function, and Theorem~\ref{thrm:EigenVector} about the existence of a suitable distance function $\Phi$ for path coupling.

The next theorem unifies these key results regarding the convergence of $F$ defined in~\eqref{eq:BP-recursion}.
\begin{theorem}\label{theorem:HardcoreConvergence}
For all $\delta>0$, there exists $\Delta_0=\Delta_0(\delta)$, for all $G=(V,E)$ of maximum degree $\Delta\geq\Delta_0$, all $\lambda<(1-\delta)\lambda_c(\Delta)$, the following hold:
\begin{enumerate}
\item
For any $x_1,x_2\in[(1+\lambda)^{-\Delta},1]$,
\begin{equation}
\label{eq:PotentialVsRealDistance}
\frac{1}{3}|x_1-x_2|\le |\Psi(x_1)-\Psi(x_2)|\le 3|x_1-x_2|.
\end{equation}
\item\label{item:FixPointEquations}
(Lemma~\ref{lemma:FixPointEquations}) The function $F$ defined in~\eqref{eq:BP-recursion} has a unique fixed point $\omega^*$. Moreover, for any initial value $\omega^{0}\in[0,1]^V$, denoting by $\omega^{i}=F^i(\omega)$ the vector after the $i$-th iterate of $F$, it holds that
\[
\|\omega^{i}-\omega^*\|_\infty\le3(1-\delta/6)^i.
\]
\item\label{item:ApproxFixPointEquations}
(Lemma~\ref{lemma:ApproxFixPointEquations})
for any $\omega\in[0,1]^V$, $v \in V$ and $R \ge 1$, we have:
\[
D_{v,R-1}(F(\omega), \omega^*) \le (1 - \delta/6) D_{v,R}(\omega,
\omega^*),
\]
where $D_{v,R}$ is as defined in~\eqref{eq:potential-metric}.
\item\label{item:EigenVector}
(Theorem~\ref{thrm:EigenVector})
There exist $\Phi:V\to \mathbb{R}_{\geq 0}$ such that  for every $v\in V$, $1\leq \Phi(v)\leq 12$,
and 
\[
(1-\delta/6)\Phi(v) \ge\sum_{u\in N(v)}\frac{\lambda\omega^*(u)}{1+\lambda\omega^*(u)}\Phi(u).
\]
\end{enumerate}
\end{theorem}

In part \ref{item:EigenVector} the astute reader may notice that 
we are considering BP without a parent, and hence each vertex 
depends on $\Delta$ neighbors.  Consequently 
parts of our analysis will consider the tree with branching factor $\Delta$.  
This is not essential in our proof, but it allows us to 
consider slightly simpler recurrences.  In our setting we have
$\Delta$ sufficiently large and since $\lambda_c(\Delta)=O(1/\Delta)$
and hence this simplification has no effect on the final result that we prove.

We first analyze the uniqueness regime described in the above Theorem \ref{theorem:HardcoreConvergence}.

Let $f_{\lambda,d}(x)=(1+\lambda x)^{-d}$ be the symmetric version of the BP recurrence~\eqref{eq:BP-recursion}. Let $\hat{x}=\hat{x}(\lambda,d)$ be the unique fixed point of $f_{\lambda,d}(x)$, satisfying $\hat{x}(\lambda,d)=(1+\lambda \hat{x}(\lambda,d))^{-d}$.
We define 
\begin{align}
\alpha(\lambda,d)=\sqrt{\frac{d\cdot\lambda \hat{x}(\lambda,d)}{1+\lambda\hat{x}(\lambda,d)}}.\label{eq:alpha-lambda-d}
\end{align}

\begin{proposition}\label{prop:uniqueness-condition}
For all $\delta>0$, there exists $\Delta_0=\Delta_0(\delta)$, for all $\Delta\ge\Delta_0$, all $\lambda<(1-\delta)\lambda_c(\Delta)$ where $\lambda_c(\Delta)=\frac{(\Delta-1)^{\Delta-1}}{(\Delta-2)^\Delta}$, it holds that $\alpha(\lambda,\Delta)\le 1-\delta/6$.
\end{proposition}
\begin{proof}
Let $x_0=\frac{1-\delta/3}{\lambda(\Delta-1+\delta/3)}$. It is easy to verify that
\[
\sqrt{\frac{\Delta\cdot\lambda x_0}{1+\lambda x_0}}\le1-\delta/6.
\]
Note that the function $\sqrt{\frac{\Delta\lambda x}{1+\lambda x}}$ is increasing in $x$. Since $f(x)$ is increasing in $\lambda$, it is easy to verify that $\hat{x}(\lambda,d)$ is increasing in $\lambda$. 
We then show that for all $\Delta\ge\Delta_0$, it holds that $\hat{x}(\lambda_0,\Delta)\le x_0$ where $\lambda_0=(1-\delta)\lambda_c(\Delta)=\frac{(1-\delta)(\Delta-1)^{\Delta-1}}{(\Delta-2)^\Delta}$, which will prove our proposition.

Since $f_{\lambda_0,\Delta}(x)$ is decreasing in $x$ and $f_{\lambda_0,\Delta}(\hat{x}(\lambda_0,\Delta))=\hat{x}(\lambda_0,\Delta)$, it is sufficient to show that
\[
f_{\lambda_0,\Delta}(x_0)=(1+\lambda_0 x_0)^{-\Delta}\le x_0.
\]

Note that it holds that 
\begin{align*}
\frac{f_{\lambda_0,\Delta}(x_0)}{x_0}
&=
\frac{\lambda_0(\Delta-1+\delta/3)}{(1-\delta/3)(1+\frac{1-\delta/3}{(\Delta-1+\delta/3)})^{\Delta}}
=
\frac{1-\delta}{1-\delta/3}\cdot
\frac{(\Delta-1)^{\Delta}(\Delta-1+\delta/3)^{\Delta}}{(\Delta-2)^\Delta \Delta^{\Delta}}\cdot\frac{\Delta-1+\delta/3}{\Delta-1}.
\end{align*}
Therefore, there is a suitable $\Delta_0=O(\frac{1}{\delta})$ such that for all $\Delta\ge\Delta_0$, 
\begin{align*}
\frac{f_{\lambda_0,\Delta}(x_0)}{x_0}
&\le\frac{1-\delta}{1-\delta/3}\left(1+O\left(\frac{\eta}{\Delta}\right)\right)\mathrm{e}^{\delta/2.99}<1,
\end{align*}
which proves the proposition.
\end{proof}

Recall recurrence $F$ as defined in~\eqref{eq:BP-recursion}.
The following proposition was proved implicitly in~\cite{LLY13}. 
\begin{proposition}[\cite{LLY13}]\label{prop:multi-variate-opt}
Let $G=(V,E)$ be a graph with maximum degree at most $\Delta$.
Assume that $\alpha(\lambda,\Delta)\le1$. For any $\omega\in[0,1]^V$, and $v\in V$,
\begin{align*}
\sqrt{\frac{\lambda F({\omega})(v)}{1+\lambda F({\omega})(v)}}\sum_{u\in N(v)} \sqrt{\frac{\lambda {\omega}(u)}{1+\lambda{\omega}(u)}}\le \alpha(\lambda,\Delta),%
\end{align*}
where $\alpha(\lambda,\Delta)$ is defined in~\eqref{eq:alpha-lambda-d}.
\end{proposition}
\begin{proof}
Let $\bar{\omega}\in[0,1]$ be that satisfies $1+\lambda\bar{\omega}=\left(\prod_{u\in N(v)}(1+\lambda {\omega}(u))\right)^{\frac{1}{|N(v)|}}$.  Denote that $\bar{\nu}=\ln (1+\lambda\bar{\omega})$ and $ {\nu}(u)=\ln(1+\lambda {\omega}(u))$. It then holds that $\bar{\nu}=\frac{1}{|N(v)|}\sum_{u\in N(v)} {\nu}(u)$. Due to the concavity of $\sqrt{\frac{\mathrm{e}^{ {\nu}}-1}{\mathrm{e}^{ {\nu}}}}$ in $ {\nu}$, by Jensen's inequality:
\[
\frac{1}{|N(v)|}\sum_{u\in N(v)} \sqrt{\frac{\lambda  {\omega}(u)}{1+\lambda {\omega}(u)}}
=\frac{1}{|N(v)|}\sum_{u\in N(v)} \sqrt{\frac{\mathrm{e}^{ {\nu}(u)}-1}{\mathrm{e}^{ {\nu}(u)}}}
\le \sqrt{\frac{\mathrm{e}^{\bar{\nu}}-1}{\mathrm{e}^{\bar{\nu}}}}
=  \sqrt{\frac{\lambda \bar{\omega}}{1+\lambda\bar{\omega}}}.
\]
Therefore, 
\[
\sqrt{\frac{\lambda F( {\omega})(v)}{1+\lambda F( {\omega})(v)}}\sum_{u\in N(v)} \sqrt{\frac{\lambda  {\omega}(u)}{1+\lambda {\omega}(u)}}
\le
\sqrt{\frac{\lambda d f(\bar{\omega})}{1+\lambda f(\bar{\omega})}\cdot \frac{\lambda d\bar{\omega}}{1+\lambda{\bar{\omega}}}},
\]
where $d=|N(v)|$ is the degree of vertex $v$ in $G$ and $f(\bar{\omega})=(1+\lambda\bar{\omega})^{-d}$ is the symmetric version of the recursion~\eqref{eq:BP-recursion}.

Define $\alpha_{\lambda,d}(x)=\sqrt{\frac{\lambda d f(x)}{1+\lambda f(x)}\cdot \frac{\lambda dx}{1+\lambda x}}$ where as before $f(x)=(1+\lambda x)^{-d}$.  The above convexity argument shows that
\begin{align}
\sqrt{\frac{\lambda F( {\omega})(v)}{1+\lambda F( {\omega})(v)}}\sum_{u\in N(v)} \sqrt{\frac{\lambda  {\omega}(u)}{1+\lambda {\omega}(u)}}
\le \alpha_{\lambda,d}(x), \text{ for some }x\in[0,1].\label{eq:bounded-by-alpha-d-x}
\end{align}

Fixed any $\lambda$ and $d$, the critical point of $\alpha_{\lambda,d}(x)$ is achieved at the unique positive $x(\lambda,d)$ satisfying
\begin{align}
\lambda d x(\lambda,d)=1+\lambda f(x(\lambda,d)).\label{eq:decay-critical}
\end{align}
It is also easy to verify by checking the derivative $\frac{\mathrm{d} \alpha_{\lambda,d}(x)}{\mathrm{d} x}$ that the maximum of $\alpha_{\lambda,d}(x)$ is achieved at this critical point $x(\lambda,d)$. 

Recall that $\hat{x}(\lambda,d)$ is the fixed point satisfying $\hat{x}(\lambda,d)=f(\hat{x}(\lambda,d))=(1+\lambda\hat{x}(\lambda,d))^{-d}$, and $\alpha(\lambda,d)=\sqrt{\frac{\lambda d \hat{x}(\lambda,d)}{1+\lambda\hat{x}(\lambda,d)}}$.
Under the assumption that $\alpha(\lambda,d)\le1$, we must have $\hat{x}(\lambda,d)\le {x}(\lambda,d)$. If otherwise $\hat{x}(\lambda,d)> {x}(\lambda,d)$, then we would have $\lambda d\hat{x}(\lambda,d)>\lambda d{x}(\lambda,d)=1+\lambda f({x}(\lambda,d))>1+\lambda f(\hat{x}(\lambda,d))=1+\lambda \hat{x}(\lambda,d)$, contradicting that $\frac{\lambda d \hat{x}(\lambda,d)}{1+\lambda\hat{x}(\lambda,d)}=\alpha(\lambda,d)^2\le 1$. Therefore, for any $x\in[0,1]$, it holds that
\begin{align*}
\alpha_{\lambda,d}(x) 
&\le \alpha(d,x(\lambda,d))\\
&=
\sqrt{\frac{\lambda d f(x(\lambda,d))}{1+\lambda f(x(\lambda,d))}\cdot \frac{\lambda d x(\lambda,d)}{1+\lambda x(\lambda,d)}}\\
&=
\sqrt{\frac{\lambda d f(x(\lambda,d))}{1+\lambda x(\lambda,d)}}
&\text{(due to~\eqref{eq:decay-critical})}\\
&\le 
\sqrt{\frac{\lambda d f(\hat{x}(\lambda,d))}{1+\lambda\hat{x}(\lambda,d)}}
&\text{(}\hat{x}(\lambda,d)\le x(\lambda,d)\text{)}\\
&=
\sqrt{\frac{\lambda d \hat{x}(\lambda,d)}{1+\lambda \hat{x}(\lambda,d)}}\\
&=
\alpha(\lambda,d).
\end{align*}
Finally, it is easy to observe that $\alpha(\lambda,d)$ is increasing in $d$ since $\alpha(\lambda,d)$ is increasing in $\hat{x}(\lambda,d)$ and $\hat{x}(\lambda,d)$ is increasing in $d$. Therefore, $\alpha(\lambda,d)\le\alpha(\lambda,\Delta)$ because $d=|N(v)|\le\Delta$. Combined this with~\eqref{eq:bounded-by-alpha-d-x}, the proposition is proved.
\end{proof}

We are now ready to prove Theorem~\ref{theorem:HardcoreConvergence}

\begin{proof}[Proof of Theorem~\ref{theorem:HardcoreConvergence}]
By Proposition~\ref{prop:uniqueness-condition}, for the regime of $\lambda$ described in the theorem, it holds that $\alpha(\lambda,\Delta)<1-\delta/6$ where $\Delta$ is the maximum degree of the graph $G=(V,E)$.

Recall that in Section~\ref{sec:BPConvergence}, we introduce the following potential function:
\[
\Psi(x)=(\sqrt{\lambda})^{-1}\textrm{arcsinh}\left (\sqrt{\lambda \cdot x}\right).
\]

We then show that for any $\omega_1,\omega_2\in[0,1]^V$, and $v\in V$,
\begin{align}
|\Psi(\omega_1(v))-\Psi(\omega_2(v))|
&\le1, \label{eq:potential-unconditional-bound}
\end{align}
and
\begin{align}
|\Psi(F(\omega_1)(v))-\Psi(F(\omega_2)(v))|
&\le (1-\delta/6) \max_{u\in N(v)}|\Psi(\omega_1(v))-\Psi(\omega_2(v))|
\label{eq:potential-contraction}.
\end{align}

We first prove~\eqref{eq:potential-unconditional-bound}. It is easy to see that $\Psi(x)$ is monotonically increasing for $x\in[0,1]$, thus $|\Psi(\omega_1(v))-\Psi(\omega_2(v))|\le \Psi(1)-\Psi(0)=\mathrm{arcsinh}(\sqrt{\lambda})/\sqrt{\lambda}$. Observe that $\mathrm{arcsinh}(x)\le x$ for any $x\ge 0$ and hence $\mathrm{arcsinh}(\sqrt{\lambda})/\sqrt{\lambda}\le 1$. This proves~\eqref{eq:potential-unconditional-bound}.

We then prove~\eqref{eq:potential-contraction}. Note that the derivative of the potential function $\Psi$ is $\Psi'(x)=\frac{\mathrm{d}\, \Psi(x)}{\mathrm{d}\, x}=\frac{1}{2\sqrt{x(1+\lambda x)}}$. Due to the mean value theorem, there exists an $\tilde{\omega}\in[0,1]^{N(v)}$ such that 
\begin{align*}
|\Psi(F(\omega_1)(v))-\Psi(F(\omega_2)(v))|
&=
\sum_{u\in N(v)}\left.\left|\frac{\partial F({\omega})(v)}{\partial {\omega}(u)}\right|_{\omega=\tilde{\omega}}\frac{\Psi'(F(\tilde{\omega})(v))}{\Psi'(\tilde{\omega}(u))} \right| |\Psi(\omega_1(u))-\Psi(\omega_2(u))|\\
&=
\sqrt{\frac{\lambda F(\tilde{\omega})(v)}{1+\lambda F(\tilde{\omega})(v)}}\sum_{u\in N(v)} \sqrt{\frac{\lambda \tilde{\omega}(u)}{1+\lambda\tilde{\omega}(u)}} |\Psi(\omega_1(u))-\Psi(\omega_2(u))|\\
&\le
\left(\sqrt{\frac{\lambda F(\tilde{\omega})(v)}{1+\lambda F(\tilde{\omega})(v)}}\sum_{u\in N(v)} \sqrt{\frac{\lambda \tilde{\omega}(u)}{1+\lambda\tilde{\omega}(u)}} \right)
\cdot\max_{u\in N(v)}|\Psi(\omega_1(u))-\Psi(\omega_2(u))|.
\end{align*}
Then \eqref{eq:potential-contraction} is implied by Proposition~\ref{prop:uniqueness-condition} and Proposition~\ref{prop:multi-variate-opt}.

Next, we prove the statements in the theorem.
\begin{enumerate}
\item
(Proof of Equation~\eqref{eq:PotentialVsRealDistance})
By the mean value theorem, for any $x_1,x_2\in[(1+\lambda)^{-\Delta},1]$, there exists a mean value $\xi\in[(1+\lambda)^{-\Delta},1]$ such that 
\[
|\Psi(x_1)-\Psi(x_2)|=\Psi'(\xi)|x_1-x_2|=\frac{1}{2\sqrt{\xi(1+\lambda\xi)}}|x_1-x_2|.
\]
For all sufficiently large $\Delta$, it holds that $(1+\lambda)^{-\Delta}>1/36$ and $\lambda<\lambda_c(\Delta)\le 0.25$, thus $\xi\in[1/36,1]$. Therefore, $\frac{1}{2\sqrt{\xi(1+\lambda\xi)}}\ge \frac{1}{2\sqrt{1+\lambda}}>\frac{1}{3}$ and $\frac{1}{2\sqrt{\xi(1+\lambda\xi)}}< \frac{1}{2\sqrt{\xi}}<3$.

\item (Proof of Lemma~\ref{lemma:FixPointEquations})
Consider the dynamical system defined by $\omega^{(i)}=F(\omega^{(i-1)})$ with arbitrary two initial values $\omega_1^{(0)},\omega_2^{(0)}\in[0,1]^V$.
The derivative of the potential function satisfies that $\Psi'(x)\ge\frac{1}{2\sqrt{1+\lambda}}$ for any $x\in[0,1]$.
Due to the mean value theorem, for any $v\in V$, there exists a mean value $\xi\in[0,1]$ such that
\[
\left|\omega_1^{(i)}(v)-\omega_2^{(i)}(v)\right|
=
\frac{1}{\Psi'(\xi)}\left|\Psi\left(\omega_1^{(i)}(v)\right)-\Psi\left(\omega_2^{(i)}(v)\right)\right|
\le 
2\sqrt{1+\lambda}\left|\Psi\left(\omega_1^{(i)}(v)\right)-\Psi\left(\omega_2^{(i)}(v)\right)\right|.
\]
Combined with~\eqref{eq:potential-unconditional-bound} and~\eqref{eq:potential-contraction}, we have
\begin{align*}
\left\|\omega_1^{(i)}-\omega_2^{(i)}\right\|_\infty
&\le2\sqrt{1+\lambda}
\left\|\Psi\left(\omega_1^{(i)}\right)-\Psi\left(\omega_2^{(i)}\right)\right\|_\infty\\
&\le 2(1-\delta/6)^{i}\sqrt{1+\lambda}\max_{z\in V}\left|\Psi\left(\omega_1^{(0)}(z)\right)-\Psi\left(\omega_2^{(0)}(z)\right)\right|\\
&\le 2(1-\delta/6)^{i}\sqrt{1+\lambda},
\end{align*}
which is at most $3(1-\delta/6)^{i}$ for $\lambda<\lambda_c(\Delta)$ for all sufficiently large $\Delta$.

Therefore, $\left\|\omega_1^{(i)}-\omega_2^{(i)}\right\|_\infty\to 0$ as $i\to\infty$ for arbitrary initial values $\omega_1^{(0)},\omega_2^{(0)}\in[0,1]^V$. This shows that the $F$ defined in~\eqref{eq:BP-recursion} has a unique fixed point $\omega^*$.

\item
(Proof of Lemma~\ref{lemma:ApproxFixPointEquations}) 
According to the definition of $D_{v,R}$ in~\eqref{eq:potential-metric},
\begin{align*}
D_{v,R-1}(F(\omega), \omega^*) 
&= 
\max_{u \in B(v,R-1)} \left| \Psi(F(\omega)(u)) - \Psi(\omega^*(u)) \right|\\
&= 
\max_{u \in B(v,R-1)} \left| \Psi(F(\omega)(u)) - \Psi(F(\omega^*)(u)) \right| &\text{($\omega^*$ is fixed point)}\\
&\le
\max_{u \in B(v,R-1)}(1-\delta/6)\max_{z\in N(u)} \left| \Psi(\omega(z)) - \Psi(\omega^*(z)) \right| &\text{(due to \eqref{eq:potential-contraction})}\\
&=
(1-\delta/6)\max_{u \in B(v,R)} \left| \Psi(\omega(u)) - \Psi(\omega^*(u)) \right|\\
&=
(1-\delta/6) \cdot D_{v,R}(\omega,\omega^*).
\end{align*}

\item(Proof of Theorem~\ref{thrm:EigenVector})
Due to Propositions~~\ref{prop:uniqueness-condition} and \ref{prop:multi-variate-opt}, for any $\omega\in[0,1]^V$, and $v\in V$, 
\begin{align*}
\sqrt{\frac{\lambda F({\omega})(v)}{1+\lambda F({\omega})(v)}}\sum_{u\in N(v)} \sqrt{\frac{\lambda {\omega}(u)}{1+\lambda{\omega}(u)}}\le 1-\delta/6.
\end{align*}

In particular, this inequality holds for the fixed point $\omega^*$ where $F(\omega^*)(v)=\omega^*(v)$. Therefore, 
\begin{align*}
\sqrt{\frac{\lambda {\omega}^*(v)}{1+\lambda {\omega}^*(v)}}\sum_{u\in N(v)} \sqrt{\frac{\lambda {\omega}^*(u)}{1+\lambda{\omega}^*(u)}}\le 1-\delta/6.
\end{align*}
We construct $\Phi:V\to\mathbb{R}_{\ge0}$ as $\Phi(v)=\sqrt{\frac{1+\lambda\omega^*(v)}{\omega^*(v)}}$ for every $v\in V$. Then
\[
\sum_{u\in N(v)}\frac{\lambda\omega^*(u)}{1+\lambda\omega^*(u)}\Phi(u)\le(1-\delta/6)\Phi(v).
\]
We now show that $1\le \Phi(v)\le 12$.
Since $\omega^*\in[0,1]^V$, we have $\Phi(v)=\sqrt{\frac{1+\lambda\omega^*(v)}{\omega^*(v)}}\ge 1$. 
Meanwhile, it holds that $\omega^*(v)=\prod_{u\in N_v}\frac{1}{1+\lambda\omega^*(u)}\ge(1+\lambda)^{-\Delta}$.
By our assumption, $\lambda\le(1-\delta)\lambda_c(\Delta)\le\frac{4}{\Delta-2}$ for all $\Delta\ge 3$. Therefore, $\omega^*(v)\ge (1+\frac{4}{\Delta-2})^{-\Delta}\ge 5^{-3}$ and  $\Phi(v)=\sqrt{\frac{1+\lambda\omega^*(v)}{\omega^*(v)}}\le \sqrt{5^3+4}\le 12$.

\end{enumerate}
\end{proof}

We consider a recurrence which corresponds to the rooted belief propagation. 
For an undirected graph $G=(V,E)$, let $\bar{E}$ be the set of all orientations of edges in $E$. The function $H:[0,1]^{\bar{E}}\to [0,1]^{\bar{E}}$ is defined as follows:
For any $\omega\in[0,1]^{\bar{E}}$ and $(v,p)\in\bar{E}$,
\begin{align}
H(\omega)(v,p)=\prod_{u\in N(v)\setminus\{p\}}\frac{1}{1+\lambda\omega(u,v)}\label{eq:BP-recursion-parent}
\end{align}
With the approach used in the proof of Theorem~\ref{theorem:HardcoreConvergence}, analyzing the convergence of $H$ is the same as analyzing $F$ on a graph $G$ with maximum degree $\Delta-1$. Recall that $\alpha(\lambda,\Delta)$ is increasing in $\Delta$. The same proof as of Theorem~\ref{theorem:HardcoreConvergence} gives us the following corollary.
\begin{corollary}\label{theorem:HardcoreConvergenceParent}
For $G=(V,E)$ and $\lambda$ assumed by Theorem~\ref{theorem:HardcoreConvergence}, the function $H$ defined in~\eqref{eq:BP-recursion-parent} has a unique fixed point $\omega^*$. Moreover, for any initial value $\omega^{0}\in[0,1]^{\bar{E}}$, denoting by $\omega^{i}=F^i(\omega)$ the vector after the $i$-th iterate of $F$, it holds that
\[
\|\omega^{i}-\omega^*\|_\infty\le3(1-\delta/6)^i.
\]

\end{corollary}

\section{Loopy BP: Proof of Theorem \ref{thrm:CnvrgLoopyBP2Correct}}\label{sec:thrm:CnvrgLoopyBP2Correct}

Consider the  version of  Loopy BP defined with the following  sequence  of messages:
For all $t\geq 1$, for $v\in V$:
\begin{equation}\label{eq:BPRecurrnceNoParent}
\tilde{R}^t_{v}=\lambda \prod_{w\in N(v)} \frac{1}{1+\tilde{R}^{t-1}_{w}}.
\end{equation}

\noindent
The system of equations specified by \eqref{eq:BPRecurrnceNoParent} is equivalent to
the one in \eqref{eq:BP-recursion} in the following sense:
Given any set of initial messages $(\tilde{R}^0_{v})_{v\in V}\in \mathbb{R}_{\geq 0}$,  it holds that
$\tilde{R}^t_{v}=\lambda F^t(\bar{\omega})(v)$, for appropriate $\bar{\omega}$ which depends
on the initial messages, i.e. $(\tilde{R}^0_{v})_{v\in V}$.   $F^t$ is the $t$-th iteration of the function $F$.

Of interest is in the quantity $q^t(v)$, $v\in V$, defined as follows:
\[   \tilde{q}^t(v)=\frac{\tilde{R}^t_{v}}{1 + \tilde{R}^t_{v}}.
\]
From Lemma \ref{lemma:FixPointEquations},  there exists $\tilde{q}^*\in [0,1]^V$
such that $\tilde{q}^t$ converges to $\tilde{q}^*$ as $t\to\infty$, in the sense that 
$\tilde{q}^t/\tilde{q}^*\to 1$.
It is elementary to show that the following holds for any $t>0$, any $p\in V$ and $v\in N(p)$:
\begin{equation}\nonumber
 \frac{q^{t}(v,p)}{\mu(\textrm{$v$ is occupied} \;|\; \textrm{$p$ is unoccupied})}=
 \frac{q^{t}(v,p)}{q^{*}(v,p)} \frac{q^{*}(v,p)}{\tilde{q}^{*}(v)} \cdot  
 \frac{\tilde{q}^{*}(v)}{\mu(\textrm{$v$ is occupied})} \cdot 
 \frac{\mu(\textrm{$v$ is occupied})}{\mu(\textrm{$v$ is occupied} \;|\; \textrm{$p$ is unoccupied})}.
\end{equation}
The theorem  follows by showing  that each of the four ratios on the r.h.s. 
are sufficiently close to 1. For
the first two ratios we use  
Theorem \ref{thrm:ConvergenceParenBP}, and for the third one we use the
 Lemma 
\ref{lemma:CnvrgLoopyBP2CorrectNoParent}.

\begin{theorem}\label{thrm:ConvergenceParenBP}
For all $\delta, \epsilon>0$, there exists $\Delta_0=\Delta_0(\delta,\epsilon)$ and $C=C(\delta,\epsilon)$,
such that for all $\Delta\geq\Delta_0$, all $\lambda<(1-\delta)\lambda_c(\Delta)$,
all graphs $G$ of maximum degree $\Delta$ and girth $\geq 6$,
all $\eps>0$ the following holds: 

There exists $q^*\in [0,1]^{E}$ such that  for $t\geq C$, for all  $p\in V$,  $v\in N(p)$
we have that
\begin{eqnarray}
\left|\frac{q^t(v,p)}{q^*(v,p)}-1 \right|\leq \epsilon \quad \textrm{and} \quad
\left|\frac{q^*(v,p)}{\tilde{q}^*(v)}-1 \right|\leq \epsilon,
\end{eqnarray}
where $q^t(v,p)$ is defined in \eqref{eq:DefQVP(t)}.
\end{theorem}

\noindent
The proof of Theorem \ref{thrm:ConvergenceParenBP} appears in Section \ref{sec:thrm:ConvergenceParenBP}.

\begin{lemma}\label{lemma:CnvrgLoopyBP2CorrectNoParent}
For all $\delta, \epsilon>0$, there exists $\Delta_0=\Delta_0(\delta,\epsilon)$ and $C=C(\delta,\epsilon)$,
such that for all $\Delta\geq\Delta_0$, all $\lambda<(1-\delta)\lambda_c(\Delta)$,
all graphs $G$ of maximum degree $\Delta$ and girth $\geq 6$,
the following holds:  Let $\mu(\cdot)$ be the Gibbs distribution,  for all  $v\in V$
we have 
\[  
  \left| \frac{\tilde{q}^{*}(v)}{\mu(\textrm{$v$ is occupied} )  } -1 \right| \leq \epsilon.
   \]
\end{lemma}

\noindent
The proof of Lemma \ref{lemma:CnvrgLoopyBP2CorrectNoParent} appears in Section \ref{sec:lemma:CnvrgLoopyBP2CorrectNoParent}.

The theorem follows by showing that 
\begin{equation}\nonumber 
\left|  \frac{\mu(\textrm{$v$ is occupied})}{\mu(\textrm{$v$ is occupied} \;|\; \textrm{$p$ is unoccupied})}-1\right| \leq 10/\Delta.
\end{equation}
From Bayes' rule we get that 
$
\mu(\textrm{$v$ is occupied} \;|\; \textrm{$p$ is unoccupied})=
\frac{\mu(\textrm{$v$ is occupied})}{\mu(\textrm{$p$ is unoccupied})}.
$ Using this observation we get that
\begin{eqnarray}\nonumber
\left|  \frac{\mu(\textrm{$v$ is occupied})}{\mu(\textrm{$v$ is occupied} \;|\; \textrm{$p$ is unoccupied})}-1\right|
&=&\left|  {\mu(\textrm{$p$ is unoccupied})}-1\right| \leq 10/\Delta.
\end{eqnarray}
In the last inequality we use the fact that  $0\leq \mu(\textrm{$p$ is occupied})\leq \lambda$.

\subsection{Proof of Theorem \ref{thrm:ConvergenceParenBP}}\label{sec:thrm:ConvergenceParenBP}
\begin{proof}
Note that by denoting $\omega^t(v,p)=\frac{R^t_{p\to v}}{\lambda}$, we have
\[
\omega^{t+1}(v,p)=H(\omega^t)(v,p),
\]
where $H$ is as defined in~\eqref{eq:BP-recursion-parent}. Then the convergence of $q^t(v,p)=\frac{R^t_{p\to v}}{1+R^t_{p\to v}}$ to a unique fixed point $q^*$ follows from Corollary~\ref{theorem:HardcoreConvergenceParent}. More precisely, there is  $\Delta_0=\Delta_0(\delta)$ and $C=C(\epsilon_0,\delta)$ such that for all $\Delta>\Delta_0$ all $\lambda<(1-\delta)\lambda_c(T_\Delta)$ and all $t>C$,
\[
\left|\omega^t(v,p)-\omega^*(v,p)\right|\leq \epsilon_0,
\]
Note that for all $t>1$, we have $\omega^t(v,p),\omega^*(v,p)\in[(1+\lambda)^{-\Delta},1]$ where $(1+\lambda)^{-\Delta}>1/36$ for $\lambda<\lambda_c(T_\Delta)$ for all sufficiently large $\Delta$. Then
\[
\left|\frac{q^t(v,p)}{q^*(v,p)}-1 \right|=\left|\frac{\omega^t(v,p)}{\omega^*(v,p)}\cdot \frac{1+\omega^*(v,p)}{1+\omega^t(v,p)}-1\right|=\frac{|\omega^t(v,p)-\omega^*(v,p)|}{\omega^*(v,p)(1+\omega^t(v,p))}\leq 36\epsilon_0.
\]
By choosing $\epsilon_0=\frac{\epsilon}{36}$, we have $\left|\frac{q^t(v,p)}{q^*(v,p)}-1 \right|\le\epsilon$.

We then show that there is a $\Delta_0=O(\frac{1}{\delta\epsilon})$ such that for all $\Delta>\Delta_0$ and all $\lambda<(1-\delta)\lambda_c(T_\Delta)$, the fixed points of the two BPs have $\left|\frac{q^*(v,p)}{\tilde{q}^*(v)}-1 \right|\leq \epsilon$

Let $\omega^t(v,p)=\frac{q^t(v,p)}{\lambda(1-q^t(v,p))}$ and $\tilde{\omega}^t(v)=\frac{\tilde{q}^t(v)}{\lambda(1-\tilde{q}^t(v))}$. It follows that
\begin{align*}
\omega^{t+1}(v,p)
&=\prod_{u\in N(v)\setminus\{p\}}\frac{1}{1+\lambda\omega^t(u,v)}
=(1+\lambda\omega^t(p,v))\prod_{u\in N(v)}\frac{1}{1+\lambda\omega^t(u,v)},\\
\tilde{\omega}^{t+1}(v)
&=\prod_{u\in N(v)}\frac{1}{1+\lambda\tilde{\omega}^t(u)}.
\end{align*}
We also define 
\[
\omega^{t+1}(v)=\prod_{u\in N(v)}\frac{1}{1+\lambda\omega^t(u,v)},
\]
therefore $\omega^{t+1}(v,p)=(1+\lambda\omega^t(p,v))\omega^{t+1}(v)$. Note that $\omega^t(p,v)\in(0,1]$, thus $|\omega^{t+1}(v,p)-\omega^{t+1}(v)|\le \lambda$. Also recall that $\lambda<\lambda_c(T_\Delta)\le 3/(\Delta-2)$ for all sufficiently large $\Delta$, therefore 
\[
|\omega^{t+1}(v,p)-\omega^{t+1}(v)|\le {3}/({\Delta-2}).
\] 
Let $\Psi(\cdot)$  be as defined in~\eqref{eq:potential-function}. Note for $t>1$ both $\omega^{t+1}(v,p)$ and $\omega^{t+1}(v)$ are from the range $[(1+\lambda)^{-\Delta},1]$. By~\eqref{eq:PotentialVsRealDistance}, for $\lambda<\lambda_c(T_\Delta)$ for all sufficiently large $\Delta$, we have
\begin{align}
|\Psi(\omega^{t+1}(v,p))-\Psi(\omega^{t+1}(v))|\le {9}/({\Delta-2}).\label{eq:potential-omega-parent-no-parent}
\end{align}
We assume that $|\Psi(\omega^t(v,p))-\Psi(\tilde{\omega}^t(v))|\le\epsilon_0$ for all $(v,p)\in E$. Then due to~\eqref{eq:potential-contraction},
\begin{align*}
|\Psi(\omega^{t+1}(v))-\Psi(\tilde{\omega}^{t+1}(v))|
&\le (1-\delta/6)\cdot \max_{u\in N(v)}|\Psi(\omega^t(u,v))-\Psi(\tilde{\omega}^t(u))|
\le(1-\delta/6)\epsilon_0.
\end{align*}
Combined with~\eqref{eq:potential-omega-parent-no-parent}, by triangle inequality, we have
\[
|\Psi(\omega^{t+1}(v,p))-\Psi(\tilde{\omega}^{t+1}(v))|\le(1-\delta/6)\epsilon_0+{9}/({\Delta-2}),
\]
which is at most $\epsilon_0$ as long as $\Delta\ge\Delta_0\ge \frac{54}{\delta\epsilon_0}+2$. It means that if $|\Psi(\omega^t(v))-\Psi(\omega^t(v,p))|\le\epsilon_0\le \frac{54}{\delta(\Delta_0-2)}$, then $|\Psi(\omega^{t+1}(v,p))-\Psi(\tilde{\omega}^{t+1}(v))|\le \frac{54}{\delta(\Delta_0-2)}$. Knowing the convergences of $\omega^{t}(v,p)$ to $\omega^*(v,p)$ and $\tilde{\omega}^{t}(v)$ to $\omega^*(v)$ as $t\to\infty$, this gives us that
\[
|\Psi(\omega^{*}(v,p))-\Psi(\tilde{\omega}^{*}(v))|\le \frac{54}{\delta(\Delta_0-2)}.
\]
By~\eqref{eq:PotentialVsRealDistance}, it implies $|\omega^{*}(v,p)-\tilde{\omega}^{*}(v)|\le \frac{162}{\delta(\Delta_0-2)}$. Again since  $\omega^{*}(v,p), \tilde{\omega}^{*}(v)\in [1/36,1]$ when $\lambda<\lambda_c(T_\Delta)$ for sufficiently large $\Delta$. It holds that
\[
\left|\frac{q^*(v,p)}{\tilde{q}^*(v)}-1 \right|
=
\left|\frac{\omega^*(v,p)}{\tilde{\omega}^*(v)}\cdot\frac{1+\lambda\tilde{\omega}^*(v)}{1+\lambda\omega^*(v,p)}-1 \right|\leq
 \frac{6000}{\delta(\Delta_0-2)}.
\]
By choosing a suitable $\Delta_0=O(\frac{1}{\delta\epsilon})$, we can make this error bounded by $\epsilon$.
\end{proof}

\subsection{Proof of Lemma \ref{lemma:CnvrgLoopyBP2CorrectNoParent}}\label{sec:lemma:CnvrgLoopyBP2CorrectNoParent}

\begin{proof}
It holds that
\begin{eqnarray}\label{eq:basis:CnvrgLoopyBP2CorrectNoParent}
  \left| \frac{\tilde{q}^*(v)}{\mu(\textrm{$v$  occupied}) } -1\right| &=&
    \left| \frac{q^*(v)}{{\frac{\lambda}{1+\lambda}\Exp{\R(X,v)}}}\frac{{\frac{\lambda}{1+\lambda}\Exp{\R(X,v)}}}{\mu(\textrm{$v$ occupied} ) } -1\right|,
\end{eqnarray}
where the expectation in the nominator is w.r.t. the random variable $X$ which is distributed as in $\mu$.
For showing the  lemma we need to bound appropriately the two ratios on the r.h.s. of 
\eqref{eq:basis:CnvrgLoopyBP2CorrectNoParent}.
For this we use the following two results.  The first one is that
\begin{equation}\label{eq:MuVsExpR}
\left| \frac{\frac{\lambda}{1+\lambda}\Exp{\R(X,v)}}{\mu(\textrm{$v$ is occupied} ) }-1 \right|\leq {200e^e\lambda}.
\end{equation}

\noindent
The second result is   Lemma \ref{lemma:RConGirth6}.
\begin{lemma}\label{lemma:RConGirth6}
For every  $\delta, \theta>0$,  there exists $\Delta_0=\Delta_0(\delta,\theta)$ and $C>0$
all $\lambda<(1-\delta)\lambda_c(\Delta)$, and   $G$ of maximum degree $\Delta$ and girth $\geq 6$,
the following is true:

Let $X$  be distributed as the Gibbs distribution. For  any $z\in V$,
it holds that 
\[
\Prob{\left| \R(X, z) - \omega^*(z) \right|\leq \theta }
  \geq 1- \exp(-\Delta/C),
\]
where $\omega^*$ is defined in Lemma \ref{lemma:FixPointEquations}.
\end{lemma}

\noindent
The proof of Lemma \ref{lemma:RConGirth6} appears in Section \ref{sec:lemma:RConGirth6}.

Before proving \eqref{eq:MuVsExpR},  let us show how   it  
implies the lemma, together with Lemma \ref{lemma:RConGirth6}. 
For any independent set $\sigma$ and any $v$, 
it holds that   $e^{-e}\leq \omega^*(v),\R(\sigma,v)\leq 1$. Then, Lemma \ref{lemma:RConGirth6} implies that 
 \begin{equation}\label{eq:ExpRVsOmegaStar}
\left| \frac{\omega^*(v)}{\Exp{\R(X,v)}}-1\right|\leq \epsilon/20.
\end{equation}
Noting that by definition it holds that   $\tilde{q}^*(v)=\frac{\lambda\omega^* }{1+\lambda\omega^*}$,
we have  that
\begin{eqnarray}
\left| \frac{\tilde{q}^*(v)}{\frac{\lambda}{1+\lambda}\Exp{\R(X,v)}} -1\right| &=&
\left| \frac{1+\lambda}{1+\lambda\omega^*(v)} \frac{\omega^*(v)}{ \Exp{\R(X,v)}}-1\right| \nonumber \\
&\leq& \frac{10\lambda}{(1+\lambda\omega^*(v))} \frac{\omega^*(v)}{ \Exp{\R(X,v)}}+
\left|  \frac{\omega^*(v)}{ \Exp{\R(X,v)}}-1\right| \leq \epsilon/15. \label{eq:ExpRVsQStar}
\end{eqnarray}
In the last inequality we use \eqref{eq:ExpRVsOmegaStar}, the fact that $\lambda<2e/\Delta$
and $\Delta$ is sufficiently large.
The lemma follows by plugging  \eqref{eq:MuVsExpR} and \eqref{eq:ExpRVsQStar} into \eqref{eq:basis:CnvrgLoopyBP2CorrectNoParent}.
We proceed by showing \eqref{eq:MuVsExpR}.
It holds that
\begin{equation}\label{eq:ERXVsMuUnblockedGirth6}
\mu(\textrm{$v$ is occupied})=\frac{\lambda}{1+\lambda}\mu(\textrm{$v$ is unblocked})
\end{equation}

\noindent
We are going to express $\mu(\textrm{$v$ is unblocked})$ it terms of the 
quantity $\R(\cdot,\cdot)$.
For $X$ distributed as in $\mu$  it is elementary to verify that 
\begin{equation}\label{eq:ExRVsMuUnblockedGirth6}
\ExpCond{\R(X,v)}{\textrm{$v$ is unoccupied}}=\mu(\textrm{$v$ is unblocked}|\textrm{$v$ is unoccupied})
\end{equation}

\noindent
Furthermore, it holds that
\begin{eqnarray}
\Exp{\R(X,v)} &=& \mu( \textrm{$v$ occupied})\cdot
\ExpCond{\R(X,v)}{\textrm{$v$ occupied}} +
\mu( \textrm{$v$ unoccupied}) \cdot
\ExpCond{\R(X,v)}{\textrm{$w$ unoccupied}}\nonumber \\
&\leq & \mu( \textrm{$v$ occupied}) +
\ExpCond{\R(X,v)}{\textrm{$v$ unoccupied}} 
\qquad\qquad\qquad \mbox{[since $0<R(X,v) \leq 1$]}\nonumber\\
&\leq & 2\lambda +
\ExpCond{\R(X,v)}{\textrm{$v$ unoccupied}} 
\;\;\qquad\qquad \qquad\qquad\qquad \mbox{[since $\mu(\textrm{$v$ occupied})\leq 2\lambda$]}\nonumber
\end{eqnarray}
Since 
$e^{-e}\leq \R(X,v)\leq 1$,  the inequality above yields 
\[
\ExpCond{\R(X,v)}{\textrm{$v$ unoccupied}}\geq \left(1-2e^e\lambda \right) \Exp{\R(X,v)}.
\]
Also, using the fact that $\R(X,v)>0$, we get 
\[
\ExpCond{\R(X,v)}{\textrm{$v$ unoccupied}}\leq \frac{\Exp{\R(X,v)}}{\mu(\textrm{$v$ is unoccupied})}\leq (1+5\lambda)\Exp{\R(X,v)}.
\]
In the last inequality we use the fact that $\mu(\textrm{$w$ is occupied})\leq 2\lambda$.
From the above two inequalities we get that
\begin{equation}\label{eq:CondVsUnCondExpRGirth6}
\left| \ExpCond{\R(X,w)}{\textrm{$w$ unoccupied}}-\Exp{\R(X,w)} \right|\leq 10e^e\lambda.
\end{equation}
In a very similar manner as above, we also get that
\begin{equation}\label{eq:CondVsUnCondMuGirth6}
\left| \mu(\textrm{$v$ is unblocked}|\textrm{$v$ is unoccupied}) 
- \mu(\textrm{$v$ is unblocked}) \right|\leq 10e^e\lambda
\end{equation}
Combining \eqref{eq:ExRVsMuUnblockedGirth6}, \eqref{eq:CondVsUnCondExpRGirth6}, 
\eqref{eq:CondVsUnCondMuGirth6} ,  \eqref{eq:ERXVsMuUnblockedGirth6} and using the fact
that $e^{-e}\leq \mu(\textrm{$v$ is unblocked}), \Exp{\R(X,w)}$
we get the following
\begin{equation}\label{eq:MuOccVsExpR}
\mu(\textrm{$v$ is occupied})=\frac{\lambda}{1+\lambda}\Exp{\R(X,w)}\left(1+50e^{e}\lambda\right).
\end{equation}
Then \eqref{eq:MuVsExpR} follows from  \eqref{eq:MuOccVsExpR}.
\end{proof}

\subsection{Proof of Lemma \ref{lemma:RConGirth6}}\label{sec:lemma:RConGirth6}

The proof of the lemma is similar to the proof of Lemma \ref{lem:ApproxVsExactFixPoint}.

Let some fixed integer $R>0$ whose value is going to be specified later. $R$ is independent of
$\Delta$, the maximum degree of $G$. For every  integer $i\leq R$, we define
$$
\beta_i:=\max\left| \Psi (\R(X, x))-\Psi(\omega^*(x) )\right|,
$$
where $\Psi$ is defined in \eqref{eq:potential-function}. The maximum is taken over all  vertices $x\in B_i(w)$.

An elementary observation   is that $\beta_i\leq C_0=3$ for every $i\leq R$.
To see why this holds,  note  that for any  $z\in V$  and any independent sets 
$\sigma$, it  holds that $e^{-e}\leq \R(\sigma,z),\omega^*(z)\leq 1$.
Then  we get $\beta_i\leq 3$ from \eqref{eq:PotentialVsRealDistance}.

We start by using  the fact that $\beta_R\leq C_0$. Then we show that 
 with sufficiently large probability, if $\beta_{i+1}\geq \theta/5$, then 
$\beta_i \le (1- \gamma) \beta_{i+1}$ where $0<\gamma<1$. 
Then the lemma follows by taking  large $R$.

For any $i\leq R$,   there exists    $C_d>0$  
such that with probability at least $1- \exp\left ( -\Delta/C_d \right)$ 
the following is true: For  every  vertex $x\in B_i(w)$ it holds that
\begin{eqnarray}\label{eq:From:lemma:ApproxFixPointGirth6}
\left | 
\R(X, x) - \exp\left( -\frac{\lambda}{1+\lambda} \sum_{z\in N(x)}
\R(X, z)
\right)
\right |  < \frac{ \theta  \delta }{40} 
\end{eqnarray}

\noindent
Note that \eqref{eq:From:lemma:ApproxFixPointGirth6}  (that 
follows from Lemma \ref{lemma:RApproxReccGirth6}) implies the following.

Fix some $i\leq R$,   $z \in B_i(w)$. 
From the definition of the quantity $\beta_{i+1}$ we get the following:
For any $x\in B_{i+1}(w)$  consider the quantity $\tilde{\omega}(x)=\R(X, x)$.
We have that 
\begin{eqnarray}\label{eq:OmegaTildeConditionGirth6}
D_{v,i+1}(\tilde{\omega}_s,\omega^*)  \leq \beta_{i+1}.
\end{eqnarray}

\noindent
We will show that if \eqref{eq:From:lemma:ApproxFixPointGirth6} holds  for
$\R(X,z)$,  where $z\in B_i(w)$,  and $\beta_{i+1}\geq \theta/5$, then 
we have that
\[
\left|  \Psi \left( \R(X,z) \right) -\Psi \left( \omega^*(z)\right)  \right| \leq (1-\delta/24)\beta_{i+1}.
\]

\noindent
For proving the above inequality, first note that if   $\R(X, z)$  satisfies  \eqref{eq:From:lemma:ApproxFixPointGirth6},
then  \eqref{eq:PotentialVsRealDistance}  implies that

\begin{eqnarray}\label{eq:UniformityBound4PhiSGirth6}
\left| \Psi\left(  \R(X, z)  \right) - \Psi\left( \exp\left( -\frac{\lambda}{1+\lambda} \sum_{r\in N(z)}
\R(X, r)
\right)\right) \right| \; \leq \; \frac{\delta \theta}{12}.
\end{eqnarray}

\noindent
Furthermore, we have that  
\begin{eqnarray}
\lefteqn{
\left|  \Psi \left( \R(X,z) \right) -\Psi \left( \omega^*(z)\right)  \right| 
} \hspace{.2in}
\nonumber
\\
&\leq &  \frac{\delta \theta}{12} + 
 \left|  \Psi \left( \exp\left( -\frac{\lambda}{1+\lambda} \sum_{r\in N_z }
\R(X, r) 
\right)\right)  -\Psi \left( \omega^*(z)\right)  \right|
 \quad\qquad\mbox{[from (\ref{eq:UniformityBound4PhiSGirth6})]} \nonumber \\
&\leq &  \frac{\delta \theta}{12} + 
 \left|  
\Psi \left( \prod_{r\in N(z)}\left(1-
\frac{\lambda
\R(X, r) }{1+\lambda} \right) \right)
-\Psi \left( \omega^*(z)\right)  \right| + \nonumber \\
&&+ 
\left | \Psi \left( \prod_{r\in N(z)}\left(1-\frac{\lambda
\R(X, r) }{1+\lambda} \right) \right)-  
\Psi  \left( \exp\left( -\frac{\lambda}{1+\lambda} \sum_{r\in N(z)}
\R(X, r) 
\right) \right) \right|, \qquad 
 \label{eq:TowardsApproxFixPointAGirth6}
\end{eqnarray}
where the last derivation follows from the triangle inequality.

From  our assumption
about  $\lambda$   and the fact that  $\R(X, r)\in [e^{-e},1]$, for  $r\in N(z)$,  we have that 
\[
\left | \prod_{r\in N(z)}\left(1-\frac{\lambda
\R(X, r) }{1+\lambda} \right)-  
\exp\left( -\lambda \sum_{r\in N(z) }
 {\frac{\R(X, r)}{1+\lambda}} 
\right) \right| \leq \frac{10}{\Delta}. \nonumber
\]
The above inequality and (\ref{eq:PotentialVsRealDistance}) 
 imply that
\[
\left | \Psi \left( \prod_{r\in N(z)}\left(1-\frac{\lambda
\R(X,r) }{1+\lambda} \right) \right)-  
\Psi  \left( \exp\left( -\frac{\lambda}{1+\lambda} \sum_{r\in N(z)}
 {\R(X,r)} 
\right) \right) \right| \leq \frac{30}{\Delta}.   
\]

\noindent
Plugging the inequality above into (\ref{eq:TowardsApproxFixPointAGirth6}) we get that
\begin{eqnarray}
\left|  \Psi \left( \R(X,z) \right) -\Psi \left( \omega^*(z)\right )  \right|  
&\leq &
\frac{\delta \theta}{12} + \frac{30}{\Delta}+
\left|  \Psi \left( \prod_{r\in N(z) }\left(1-\frac{\lambda
\R(X,r) }{1+\lambda } \right) \right) - \Psi \left( \omega^*(z)\right) \right| \nonumber \\
&\leq &
{\delta \theta}/{12} + {60}/{\Delta}+
D_{v,i}(F(\tilde{\omega}), \omega^*),
\label{eq:ProductVsExpontialWorstCaseGirth6}
\end{eqnarray}
where $\tilde{\omega}\in[0,1]^V$ is such that $\tilde{\omega}(z)=\R(X,z)$ for $z\in V$.  The function $F$ is defined in \eqref{eq:BP-recursion}.  
Since   $\tilde{\omega}$  satisfies  \eqref{eq:OmegaTildeConditionGirth6},  Lemma \ref{lemma:ApproxFixPointEquations}  implies that
\begin{equation}
 D_{v,i}(F(\tilde{\omega}), \omega^*) \leq (1-\delta/6 )\beta_{i+1}. \label{eq:ProductWorstCaseVsFixPointGirth6}
\end{equation}

\noindent
Plugging (\ref{eq:ProductWorstCaseVsFixPointGirth6}) into (\ref{eq:ProductVsExpontialWorstCaseGirth6}) we get that
\begin{eqnarray}
\left|  \Psi \left( \R(X_s,z) \right) -\Psi \left( \omega^*(z)\right)  \right|  \leq 
{\delta \theta}/{12} + {60}/{\Delta}+ (1-\delta/6)\beta_{i+1} 
 \leq    (1-\delta/24)\beta_{i+1}, \label{eq:FixTimeAlphaIGirth6}
\end{eqnarray}
\noindent
where the last inequality holds if we have $\beta_{i+1}\geq \theta/5$.
Note that \eqref{eq:FixTimeAlphaIGirth6} holds provided that $\R(X,z)$ satisfies \eqref{eq:From:lemma:ApproxFixPointGirth6}.
The lemma follows by taking sufficiently large $R=R(\theta)$.

\section{Basic Properties of Glauber dynamics}

\subsection{Continuous versus discrete time chains}\label{sec:ContVsDisc}

For many of our results we have a simpler proof when instead of a  discrete 
time Markov chain we consider  a continuous time version of the chain.
That is, consider the Glauber dynamics where the spin of each vertex is updated according
to an independent Poisson clock with rate $1/n$. 

We use the following observation, Corollary 5.9 in \cite{MitzenUpfal}, as a generic tool to argue that typical properties of
continuous time chains are typical properties of the discrete time chains too. 

\begin{observation}\label{fact:DiscVsCont}
Let $(X_t)$ by any discrete time Markov chain on state space $\Omega$, and let $(Y_t)$ be the corresponding
continuous-time chain. Then for any property $P\subset \Omega$ and positive integer $t$, we have that
\[
\Prob{X_{t}\notin P}\leq e\sqrt{t}\Prob{Y_t\notin P}.
\]
\end{observation}

Observation \ref{fact:DiscVsCont} would suffice for our purposes when   $\Delta=\Omega(\log n)$, but not for Glauber dynamics on graphs of e.g. constant degree. For the latter case, instead of focusing on specific times $t$ in discrete time, our goal will be to show how events which are rare at a single instant in continuous time must also be rare over a time interval of length $O(n)$ in discrete time, without taking a union bound over all the times in the time interval.

Let the set $\Omega$ contain all the independent sets of $G$. 
We say that a function $f:\Omega\to \mathbb{R}$ has ``total influence" $J$, if for every independent set $X\in \Omega$
we have 
\[
\Exp{\left| f(X')-f(X) \right|} \leq J/n,
\]
where $X'$ is the result of one Glauber dynamics update, starting from $X$.

The  next result, Lemma 13 in \cite{Hayes}, shows that, for  functions $f$  which  have Lipschitz constant 
$O(1/\Delta)$ and total influence $J=O(1)$, in order to prove high-probability bounds for the discrete-time chain that apply for all times in an interval of length $O(n)$, it suffices to be able to prove a similar bound at a single instant in continuous time.

\begin{lemma}[Hayes \cite{Hayes}]\label{lemma:ContVsDiscSmallDelta}
Suppose $f:\Omega\to \mathbb{R}$ is a function of independent sets of $G$ and $f$ has Lipschitz constant 
$\alpha<O(1/\Delta)$ and total influence $J=O(1)$. Let $X_0=Y_0$ be given and let $(X_t)_{t\geq 0}$ be 
continuous-time single site dynamics on the hard-core model of $G$ and let $(Y_i)_{i=0,1,2\ldots}$ be the
corresponding discrete-time dynamics. 

Suppose that $t_0$ is a positive integer and $S$ is a measurable set 
of real numbers, such that for all $t\geq t_0$, $\Prob{f(X_t)\in S}\geq 1-\exp(-\Omega(\Delta ))$. Then, for all
$\epsilon\in \Omega(1)$ and all integers $t_1\geq t_0$, there $t_1-t_0=O(n)$ we have that
\[
\Prob{(\forall i\in\{t_0,t_0+1,\ldots, t_1\})\; f(Y_i)\in S \pm\epsilon }\geq 1-\exp\left( -\Omega(\Delta )\right),
\]
where the hidden constant in $\Omega$ notation depends only on the hidden constant in the assumption.
\end{lemma}

\subsection{Basic burn-in properties}

\label{sec:heavy}

Consider a graph $G=(V,E)$.  Given  some integer $r\geq 0$ and $v\in V$, let  $B_r(v)$ be the the ball of radius $r$, centered 
at $v$. Also, let $S_r(v)$ be the sphere of radius $r$, centered at $v$. Finally, let $N(v)$ denote the set
of vertices which are adjacent to $v$.

\begin{definition}
Let $G=(V,E)$ be a graph of maximum degree $\Delta$ and let $\sigma$ be an independent
set of $G$. For some $\rho>0$,  we say that $\sigma$  is $\rho$-heavy for the vertex $v\in V$ if 
$|B_2(v)\cap \sigma|\geq \rho\Delta$ or $|B_1(v)\cap \sigma|\geq \rho\Delta/\log \Delta$.
\end{definition} 

\begin{definition}
Let $G=(V,E)$ be a graph of maximum degree $\Delta$. Let  $\sigma,\tau$ be  independent sets
 of $G$. Consider integer $r>0$ and  $v\in V$.
If  there is a vertex $w\in B_r(v)$ such $w$ is  $\rho$-heavy, then  $\sigma$ is called $\rho$-suspect for radius 
$r$ at $v$. Otherwise,  we say that $\sigma$ is $\rho$-above suspicion for radius $r$ at $v$.

Similarly, for $\sigma,\tau$ such that $\sigma(v)\neq \tau(v)$, we say that $v$ is  a $\rho$-suspect disagreement
for radius $r$ if there exists a vertex $w\in B_r(v)$ such that either $\sigma$ or $\tau$ is $\rho$-heavy at $w$. Otherwise, 
we say that $v$ is a $\rho$-above suspicion disagreement for radius $r$. 
\end{definition}

\noindent
For the purposes of path coupling for every  pair of independent sets $X,Y$ we consider 
shortest paths between $X$ and $Y$ along neighboring independent sets. That is,
$X=Z_0\sim Z_1\sim \cdots \sim Z_{\ell}=Y$.  This sequence 
$Z_1, \ldots, Z_{\ell}$  we call  interpolated independent sets for $X$ and $Y$. A key
aspect of the above definitions is that the ``niceness" is inherited by interpolated
independent sets.

\begin{observation}
If $X,Y$ are independent sets, neither of which is $\rho$-heavy at vertex $v$, then
no interpolated independent set is $2\rho$-heavy at $v$.  Likewise, if $v$ is $\rho$-above
suspicion disagreement for radius $r$, then in every interpolated independent sets $v$
is $2\rho$-above suspicion for radius $r$.
\end{observation}

The following lemma states that   $(X_t)$ requires  $O(n\log\Delta)$ to burn-in, regardless of
$X_0$.

\begin{lemma}\label{lemma:BurnIn}
For $\delta>0$ let $\Delta\geq \Delta_0(\delta)$ and $C_b=C_b(\delta)$.
Consider a graph $G=(V,E)$ of maximum degree $\Delta$. Also, let $\lambda\leq (1-\delta)\lambda_c({\Delta})$. 

Let $(X_t)$  be the continuous (or discrete) time Glauber dynamics on the hard-core model with fugacity $\lambda$ and underlying graph
$G$.
Consider $v\in V$ and   let   ${\cal C}_t$ be the event that  $X_t$, is 50-above suspicion  for radius 
$r=\Delta^{9/10}$ for $v$ at time $t$.
Then,  for ${\cal I}=[10n\log \Delta, n\exp(\Delta/C_b)]$ it holds that
\[
 \Prob{\cap_{t\in {\cal I}}{\cal C}_t}\geq 1-\exp\left (-\Delta/ C_b  \right).
\]
\end{lemma}

\begin{proof} 
For now, consider the continuous time version of $(X_t)$. 
Recall that  for $X_t$,  the vertex $u$ is not 
$\rho$-heavy if both of the following  conditions hold
\begin{enumerate}
\item $|X_t\cap B_2(u)|\leq \rho\Delta $
\item $|X_t\cap N(u)|\leq \rho\Delta/\log \Delta.$
\end{enumerate}

\noindent
First we consider a fixed time $t\in {\cal I}$. Let $c=t/n$. Note that $c=c(\Delta)\geq 10\log \Delta$.
We are going to show that   there exists $C'>0$ such that
\begin{equation}\label{eq:FixTAlphaBoundWorstInit}
\Prob{{\cal C}_t}\geq 1-\exp\left(-\Delta/C' \right).
\end{equation}

\noindent
Fix some vertex $u\in B_r(v)$.  Let $N_0$ be the set  of vertices in $B_2(u)\cap X_0$ 
which are not updated during the time period $(0,t]$. That is, for $z \in N_0$ it holds that
$X_0(z)=X_t(z)$.  Each vertex $z\in B_2(u )\cap X_0$ belongs to $N_0$ with
probability $\exp\left(-t/n \right)=e^{-c}$, independently of the other vertices.  
Since    $|B_2(u)\cap X_0|\leq \Delta^{2}$,  it is elementary that the  distribution of $|N_0|$ 
is dominated by ${\cal B}(\Delta^2, e^{-c})$,  i.e. the binomial with parameters $\Delta^2$ and $e^{-c}$.

Using Chernoff's bounds we get the following:
 for $c>10\log \Delta$ it holds that
\begin{equation}\label{eq:N_0Tail4LightnessWorstInit}
\Prob{N_0>\Delta/10}\leq\exp\left( -\Delta/10\right).
\end{equation}

\noindent
Additionally, let $N_1\subseteq B_2(u)$ contain every vertex $u$ which is updated
at least once during the period $(0,t]$. Each vertex  $z\in N_1$, which is last updated 
 prior to $t$ at time $s\leq t$, becomes occupied during the update at time $s$ with 
probability at most $\frac{\lambda}{1+\lambda}$,  regardless of $X_{s}(N(z))$.
Then, it is direct that $|X_t\cap N_1|$ is dominated by ${\cal B}(N_1,\frac{\lambda}{1+\lambda})$.

Noting that $|N_1|\leq |B_2(u)|\leq \Delta^2$ and  $\frac{\lambda}{1+\lambda}<2e/\Delta$, for
$\Delta>\Delta_0$  Chernoff's bound imply that
\begin{equation}\label{eq:N_1Tail4LightnessWorstInit}
\Prob{ |N_1\cap X_t|\geq 15e\Delta }\leq \exp\left( -15e\Delta \right).
\end{equation}
From \eqref{eq:N_0Tail4LightnessWorstInit}, \eqref{eq:N_1Tail4LightnessWorstInit} and a simple union 
bound, we get that
\begin{equation}\label{eq:B2LightnessTailWorstInit}
\Prob{|X_t\cap B_2(u)|> 42\Delta } \leq \exp\left(- \Delta/20 \right).
\end{equation}
Using exactly the same arguments, we also get that
\begin{equation}\label{eq:NLightnessTailWorstInit}
\Prob{|X_t\cap N(u)|> 42\Delta/\log \Delta } \leq \exp\left( - \Delta /20 \right).
\end{equation}
\noindent
Note that $X_0$ could be such that $N(u)\cap X_0=\alpha \Delta$, for some fixed $\alpha>0$.
So as to get $|X_t\cap N(u)|\leq 42\Delta/\log\Delta$ with large probability, we have
to ensure that with large probability all the vertices in $N(v)$ are updated at least once.
For this reason the burn-in  requires at least $10n\log\Delta$ steps.

From \eqref{eq:B2LightnessTailWorstInit} and \eqref{eq:NLightnessTailWorstInit} we get the following: 
For any $\rho>50$ it holds that
\begin{equation}\label{eq:X^*(v)NotHeavyWorstInit}
\Prob{X_t(u)\textrm{ is not $\rho$-heavy}}\leq \exp\left (- \Delta /25 \right).
\end{equation}
Then  \eqref{eq:FixTAlphaBoundWorstInit}  follows by taking a union bound over all
the, at most $\Delta^r$ vertices in $B_r(v)$. In particular, for  $r=\Delta^{9/10}$ 
and sufficiently large $\Delta$, there exists $C>0$ such that 
\[
\Prob{{\cal C}_t}\leq \Delta^r \exp\left (- \Delta /25\right) \leq 
 \exp\left(- \Delta /30\right).
\]
The above  implies that \eqref{eq:FixTAlphaBoundWorstInit} is indeed true
but only for  a specific time step $t\in {\cal I}$. Now we use a covering argument 
to deduce the above  for the whole interval $\cal I$.

For sufficiently small $\gamma>0$, independent of $\Delta$, 
 consider a partition of the time interval ${\cal I}$ into subintervals each of length $\frac{\gamma^2}{\Delta}n$,
(where the last part can be shorter). 
We let $T(j)$ be the $j$-th part in the partition.

Each $z \in  B_2(w)$ is updated at least once during the time period $T(j)$ with probability  less than 
$2\frac{\gamma^2}{\Delta}$, independently of the other vertices. Note that $|B_2(w)|\leq \Delta^2$.
Clearly, the number of vertices in $B_2(v)$ which are updated during $T(j)$ is dominated by
${\cal B}(\Delta^2, 2\gamma^2/\Delta)$.
Chernoff bounds imply that  with probability at least $1-\exp\left( -20\Delta \gamma^2 \right)$, the number of
vertices in $B_2(w)$ which are updated  during the interval $T(j)$ is at most $20\gamma^2\Delta$. 
Furthermore, changing any $20\Delta\gamma^2$ variables in $B_2(w)$ can only make the independent set 
heavier by at most $20\Delta\gamma^2$.

Similarly, we get that with probability at least $1-\exp\left( -\gamma\Delta  \right)$, the number of vertices in $N(v)$
which are updated during the interval $T(j)$ is at most $\gamma\Delta/\log \Delta$. The change of
at most $\gamma\Delta/\log \Delta$ neighbors of $v$ does not change the weight of its neighborhood by
more than $\gamma\Delta/\log \Delta$.

From the above arguments we get that the following: We can choose sufficiently large $C_b>0$ such that
for  $j\in \{1,2, \ldots, \lceil \Delta/(\gamma^2 ) \exp\left( \Delta/ C_b\right) \rceil\}$
it holds that
\[
\Prob{\cap_{t\in T(j)} {\cal C}_t }\geq 1-\exp\left( -100\Delta/C_b \right).
\]
The result for continuous time  follows by taking a union bound over all the $\left\lceil \Delta/(\gamma^2 )
\exp\left ( \Delta/C_b \right) \right\rceil$ many subintervals of ${\cal I}$. 

For the discrete time case the arguments are very similar.  The only extra ingredient we need
is that, now,  the updates of the vertices are negatively dependent and  use 
\cite{NegativeDep}. The lemma follows. 
\end{proof}

The following lemma states that if $(X_t)$ start from a not so heavy state it only requires
$O(n)$ steps to burn in.

\begin{lemma}\label{lemma:LightnessOfX}
For $\delta>0$,   let $\Delta\geq \Delta_0(\delta)$ and $C_b=C_b(\delta)$. 
Consider a graph $G=(V,E)$ of maximum degree $\Delta$. Also, let $\lambda\leq (1-\delta)\lambda_c({\Delta})$. 

Let $(X_t)$ be the continuous (or discrete) time Glauber dynamics on the hard-core model with fugacity 
$\lambda$ and underlying graph $G$.
Also, let  ${\cal C}_t$ be the event  that  $X_t$ is,  are 50-above suspicion  
for radius $R\leq \Delta^{9/10}$ for $v$ at time $t$.
Assume that $X_0$  is 400-above suspicion  for radius $R$ for $v$.  
Then,    for ${\cal I}=[C_bn, n \exp (\Delta/C_b )]$
we have that
\[
\Prob{\cap_{t\in {\cal I}}{\cal C}_t}\geq 1-\exp\left[-\Delta/ C_b  \right].
\]
\end{lemma}
\noindent
The proof of Lemma \ref{lemma:LightnessOfX} is almost identical to the proof of
Lemma \ref{lemma:BurnIn}, for this reason we omit it.

\subsection{ $G$ versus $G^*$ and comparison}\label{sec:GVsG*}

\noindent
Consider $G$ with girth $7$. For such graph and some  vertex $w$ in $G$, the radius 3 ball around $w$ is a tree.
We let $G^*_w$ be graph that is derived from $G$ by orienting  towards $w$ every edge  that is  
within distance 2 from $w$ \footnote{
An edge $\{w_1,w_2\}\in E$ is at distance $\ell$ from $w$ if the minimum distance between
 $w$ and any  of $w_1,w_2$ is $\ell$.}. 
For a vertex $x\in G^*_w$, we let $N^*(x)\subseteq N(x)$ contain  every $z$ in the neighborhood of
$x$ such that either  the edge between  $x,z$ is unoriented, or the orientation is towards $x$.

We let the Glauber dynamics $(X^*_t)$  on the hard-core model with underlying graph $G^*_w$ 
and fugacity  $\lambda$,  be a Markov chain whose   transition $X_t\rightarrow X_{t+1}$ is 
defined by the following:
\begin{enumerate}
\item Choose $u$ uniformly at random from $V$.
\item If $N^*(u)\cap X^*_t = \emptyset$, then let 
\[ 
X^*_{t+1} = \begin{cases}  
X^*_t\cup\{u\} & \mbox{ with probability } \lambda/(1+\lambda) \\
X^*_t\setminus\{u\} & \mbox{ with probability } 1/(1+\lambda) 
\end{cases}
\]
\item If $N^*(w)\cap X_t \neq \emptyset$, then let $X^*_{t+1} = X^*_t$.
\end{enumerate}
The state space of $(X^*_t)$ that is implied by the above is a superset of the independent
sets of $G$, since there are pairs of vertices  which are adjacent in $G$ while they can both
be occupied  in $X^*_t$.

The motivation for using $G^*_w$ and $(X^*_t)$ is better illustrated by considering
Lemma \ref{lemma:RApproxReccGirth6}. 
In Lemma \ref{lemma:RApproxReccGirth6} we establish a recursive relation for $\R()$ 
for $G$ of girth $\geq 6$,  in the setting of the Gibbs distribution. 
An important  ingredient in the proof  there 
is that   for every vertex $x$  conditioned on the configuration at $x$ and the 
vertices at distance  $\geq 3$ from $x$,
the children of $x$ are  mutually independent of each other
under  the Gibbs distribution.

For establishing the  uniformity property for Glauber dynamics we need to
establish a similar ``conditional independence" relation but in the dynamic setting of  Markov chains.  
To obtain this, we will need that  $G$ has girth at least 7.
Clearly, the conditional independence of Gibbs distribution no longer holds for the Glauber dynamics.
To this end we employ the following: Instead of considering $G$  and the standard
Glauber dynamics  $(X_t)$, we consider $G^*_w$ and the corresponding dynamics $(X^*_t)$.

Using $G^*_w$ and $(X^*_t)$  we   get (in the dynamics setting) 
an effect which is similar to the conditional independence.
During the evolution of $(X^*_t)$ the neighbors of $w$ can only exchange information through
paths of $G^*_w$ which travel outside the ball of radius 3  around $w$, i.e. $B_3(w)$. This holds 
due to the girth assumption for $G^*_w$
and the definition of $(X^*_t)$. In turn this implies that conditional on the configuration of $X^*_t$
outside $B_3(w)$, the (grand)children of $w$ are  mutually independent. 

The above trick allows to get a recursive relation for $R(X^*_t,w)$ similar to that
for the Gibbs distribution. So
as to argue that a somehow similar relation holds for the standard dynamics $(X_t)$, we use the
following result which states that if $(X^*_t)$ and $(X_t)$ start from the same configuration,
then after $O(n)$ the number of disagreements between the two chains is not too large.

\begin{lemma}\label{lem:ComparisonXVsXStar}
For  $ \gamma>0$,  $C_1>0$,  there exists $\Delta_0$,  $C_2 > 0$ such that the following is true:
For $w\in V$ consider $G^*_w$ of maximum degree $\Delta>\Delta_0$ and girth at least 7.
 Also,  let $(X_t)$ and $(X^*_t)$ be the continuous time Glauber dynamics  on the hard-core model with fugacity 
 $\lambda <(1-\delta)\lambda_c(\Delta)$, underlying graphs
$G$ and $G^*_u$, respectively.

Assume that $(X^*_t)$ and $(X_t)$ are maximally coupled.
Then, if  $X_0=X^*_0$ for  $X_0$ which is 400-above suspicion for radius $R\leq \Delta^{9/10}$,
we have that
\[
\Pr[\forall s\le C_1n, \; \forall u \in V\; |(X_s\oplus X^*_s) \cap B_2(u)|\leq \gamma
\Delta ]\geq 1- \exp\left( -\Delta/C_2\right).
\]
\end{lemma}

\noindent
Before proving Lemma \ref{lem:ComparisonXVsXStar} we need to introduce certain notions.

Let us call $Z$ a ``generalized Poisson random variable with jumps $\alpha$ and instantaneous
rate $r(t)$" if $Z$ is the result of a continuous-time adapted process, which begins at $0$ and in each 
subsequent infinitesimal time interval, samples an increment $\partial Z$ from some distribution over 
$[0,\alpha]$, having mean $\leq r(t)dt$.
 $Z$, the sum of the increments over all times $0<t<1$, is a random variable, as is the 
 maximum observed rate, $r^*=\max_{t\in [0,1]}r(t)$.
 
 \begin{remark}
 In the special case where $\alpha\geq 1$ and the distribution is supported in $\{0,1\}$
 with constant rate $\mu \cdot dt$, $Z$ is a Poisson random variable with mean $\mu$.
 \end{remark}

\noindent
We are going to use the following result, Lemma 12 in \cite{Hayes}.

\begin{lemma}[Hayes]\label{lem:GenPoisson} Suppose $Z$ is a generalized Poisson random variable with
maximum jumps $\alpha$ and maximum observed rate $r^*$. Then, for every
$\mu>0$, $C>1$ it holds that
\[
\Prob{Z\geq C\mu \textrm{ and } r^* \leq \mu} \leq \exp\left[-\frac{\mu}{\alpha}(C\ln (C) -C+1)\right]<\left(\frac{e}{C} \right)^{\mu C \alpha}.
\]
\end{lemma}

\begin{proof}[Proof of Lemma \ref{lem:ComparisonXVsXStar}]

In this proof assume that $\gamma C_1$ is sufficiently small constant.
Also,  let $D=\cup_{t\leq C_1n} (X_t\oplus X^*_t)$, i.e. $D$ denotes the set of all vertices 
which are  disagreeing at least once during the time interval from $0$ to $C_1n$.
Given some vertex $u\in V$ let $A_u=\cup_{t\leq C_1n}X_t\cap N(u)$ and 
$A^*_u=\cap_{t\leq C_1n}X^*_t\cap N(u)$. That is $A_u$ contains every 
$z\in N(u)$ for which   there exists  at least one $s<C_1n$  such that $z$ is 
occupied in $X_s$. Similarly for $A^*_u$.
Finally, let  the integer $r=\left\lfloor \gamma^5 \frac{\Delta}{\log \Delta}\right \rfloor$. 

Let ${\cal A}$ denote the event that 
$\exists  s\le C_1n, \; \exists u \in V\; |(X_s\oplus X^*_s) \cap S_2(u)|\geq \gamma\Delta/2$.
Consider  the  events ${\cal B}_1, {\cal B}_2$, ${\cal B}_3$, ${\cal B}_4$ and ${\cal B}_5$ be defined as follows:
${\cal B}_1$ denotes the event that $D\nsubseteq B_r(w)$. 
${\cal B}_2$ denotes the event that $|D|\geq \gamma^3 \Delta^{2}$.
${\cal B}_3$ denotes the event that   the total number of disagreements
that appear in $N(u)$, for every $u\in V$, is at most $\gamma^3 \Delta$.
Finally, ${\cal B}_4$ denotes the event that there exists  $u\in B_{100}(w)$ such that 
either $|A(u)|\geq \gamma^3\Delta$ or  $|A^*(u)|\geq \gamma^3\Delta$.

Then, the lemma follows by noting the following:
\begin{equation}\label{eq:GVsGstartBasicEq}
\Prob{\exists s\le C_1n, \; \exists u \in V\; |(X_s\oplus X^*_s) \cap B_2(u)|\geq \gamma
\Delta} \leq \Prob{\cal A}+\Prob{{\cal B}_3}.
\end{equation}
The lemma follows by bounding appropriately the probability terms on the r.h.s. of 
\eqref{eq:GVsGstartBasicEq}.

First consider $\Prob{\cal A}$. Let ${\cal B}={\cal B}_1\cup {\cal B}_2\cup {\cal B}_3\cup {\cal B}_4$.
We bound $\Prob{\cal A}$ by using ${\cal B}$ as follows: 
\begin{eqnarray}
\Prob{{\cal A}} &=& \Prob{ {\cal B},  {\cal A} }+
\Prob{ \bar{{\cal B}}, {\cal A}} \nonumber\\
&\leq &  \Prob{ {\cal B}}+\Prob{ \bar{{\cal B}}, {\cal A}}\nonumber \\
&\leq &  \sum^4_{i=1}\Prob{ {\cal B}_i} +\Prob{ \bar{{\cal B}}, {\cal A}},\label{eq:Basis4ProbA}
\end{eqnarray}
where the last inequality follows by applying a simple union bound.

Consider some vertex $u\in V$ and let
$Z$ be the total number of disagreements that ever occur in $S_2(u)$
up ot the first time that either ${\cal B}$ occurs or  up to time $C_1n$,
whichever happens first.
If $u\notin B_{r}(w)$,  then $Z$ is always zero since we stop the clock when $D\nsubseteq B_{r-1}(w)$.
So our focus is on the case where $u\in B_{r-1}(w)$. For such $u$ the random variable $Z$ follows
a  generalized Poisson distribution, with jumps of size 1 and maximum observed rate
at most $30\gamma^3\Delta dt/n$,  over at most $C_1n$ time units.
To see this consider the following.

 Given that ${\cal B}$ does not occur,  disagreements in $S_2(u)$ may be caused
due to  the following categories of disagreeing edges. Each disagreement in $N(u)$ has
 at most $\Delta-1$ disagreeing edges  in $S_2(u)$. Since the number of disagreements that
appear in $N(u)$ during the time period up to $C_1n$ is  at most $\gamma^3\Delta$,
there are at most $\gamma^3\Delta^2$ disagreeing edges incident to $S_2(u).$
On the whole there are at most $\gamma^3\Delta^{2}$ disagreements from vertices different
than those in $N(u)$.
 Each one of them has at most one neighbor in $S_2(u)$, since the girth is at least 7. That
is there are additional $\gamma^3\Delta^{2}$ many disagreeing edges. Finally, 
disagreements on $S_2(u)$ may be caused by  edges which belong to  $G\oplus G^*_w$.
There are at most $\Delta^3$ many such edges. 
Each one of these edges generates disagreements only on the vertex on its tail.
Since the out-degree in $G^*_w$ is at most 1, there are $\Delta^2$ disagreeing edges
from $G\oplus G^*_w$ which are incident to $S_2(v)$.
Additionally,  each one of these edges   should point to an occupied vertex so as to be disagreeing.
Since ${\cal B}_4$ does not occur,  there at most $2\gamma^3\Delta^2$ edges in $G\oplus G^*_w$ 
which point to an occupied vertex and have the tail in $S_2$.

From the above observations, we have that  there are at most  $10\gamma^3\Delta^2$ disagreeing edges incident to $S_2$. 
For the new disagreement to occur in $S_2$ due to a given such edge, a specific vertex  
must chosen and should become occupied, which occurs with rate at most $e\cdot dt/(n\Delta)$.

Using  Lemma \ref{lem:GenPoisson}, applied with $\mu=30C_1\gamma^3\Delta$, $\alpha=1$ and
$C=\gamma\Delta/\mu$, we have that
\[
\Prob{Z\geq \gamma \Delta} \leq \left (30e \gamma^2 C_1\right)^{\gamma\Delta}.
\]
Taking a union bound over the, at most, $\Delta^{r}$ vertices in $B_r(v)$, we get that
\begin{eqnarray}
\Prob{\bar{{\cal B}}, {\cal A} }
&\leq & \Delta^{r}  \left (30 e \gamma^2 C_1\right)^{\gamma\Delta} 
\;=\; \exp\left( -\Delta/C_3 \right),\label{eq:ProbBcA}
\end{eqnarray}
where $C_3=C_3(\gamma)>0$ is a sufficiently large number. In the last derivation we
used the fact that $r \leq \frac{\gamma^5\Delta}{\log \Delta}$.

We proceed by bounding the probability of the events ${\cal B}_1$, ${\cal B}_2$, ${\cal B}_3$ and ${\cal B}_4$.
The approach is very similar to the proof of Theorem 27 in \cite{Hayes}. We repeat
it for the sake of completeness.

Recall that ${\cal B}_1$ denotes the event that $D\nsubseteq B_r(w)$. 
The bound for $\Prob{{\cal B}_1}$ uses standard arguments of disagreement percolation. 
First we observe that every disagreement outside $B_{r}(w)$ must arise via some path of 
disagreement which starts within $B_2(w)$. That is we need at least one path of disagreement 
of length $r-4$. We fix a particular path of length $r-4$ with $B_r(w)$. Let us call it ${\cal P}$.
 We are going to bound the probability that disagreements percolate along ${\cal P}$ within $C_1n$ 
 time units. Let us call this probability $\rho$.

The number of steps along this path that a disagreement actually percolates is a generalized
Poisson random variable with jumps 1 and maximum overall rate at most $C_1e/\Delta$.
This follows by noting that the maximum instantaneous rate is at most $e \cdot dt/(n\Delta)$
integrated over $C_1n$ time units.  We use Lemma \ref{lem:GenPoisson},
to  bound the probability for the disagreement  to percolate along ${\cal P}$, i.e. $\rho$.
Setting $\mu=eC_1/\Delta$, $\alpha=1$ and $C=(r-4)/\mu$ in Lemma \ref{lem:GenPoisson} yields
the following bound for $\rho$
\[
\rho\leq \left( \frac{e^2C_1}{\Delta(r-4)}\right)^{r-4}.
\]
The above bound holds for any path of length $r-4$ in $B_r(w)$. 
Taking a union bound over the at most $\Delta^3$ starting point in $B_2(w)$ and the at most
$\Delta^{r-4}$ paths of length $r-4$ from a given starting point we get that
\begin{equation}\label{eq:ProbB1}
\Prob{{\cal B}_1}\leq \Delta^3\left( \frac{e^2C_1}{r-4} \right)^{r-4}\leq \exp\left( -\Delta/C_4 \right),
\end{equation}
where $C_4=C_4(\gamma)>0$ is a sufficiently large number. 
 
Recall that ${\cal B}_2$ denotes the event that $|D|\geq \gamma^3 \Delta^{2}$. 
For $\Pr[{\cal B}_2]$ we consider the waiting time $\tau_i$ for the $i$'th disagreement,
counting from when the $(i-1)$'st disagreement is formed. The event ${\cal B}_2$ is equivalent
to $\sum^{(\gamma^3\Delta^{2})}_{i=1}\tau_i\leq C_1n$.
 
Each new disagreement can be attributed to either an edge joining it to an existing disagreement,
or to one of the edges in $G\oplus G^*_w$. It follows easily that the total number of such edges is at most 
$|G\oplus G^*_w|+|(i-1)\Delta|=\Delta^3+(i-1)\Delta$. Furthermore, for the new disagreement to occur
due to a given such edge, a specific vertex  must chosen, which occurs with rate at most $e\cdot dt/(n\Delta)$.

The above observations suggest that the waiting time $\tau_i$ is stochastically dominated by an exponential
distribution  with mean $n/[e(\Delta^2+i-1)]$,  even conditioning on an arbitrary previous histry
$\tau_1,\tau_2,\ldots, \tau_{i-1}$. Therefore, $\sum_{i}\tau_i$ is stochastically dominated by the sum
of independent exponential distributions with mean $n/[e(\Delta^2+i-1)]$.

Applying Corollary 26, from \cite{Hayes} to $\tau_1+\cdots+\tau_{(\gamma^3\Delta^{2})}$, with
\[
\mu=\sum^{(\gamma^3\Delta^{2})}_{i=1}\frac{n}{e(\Delta^2+i-1)}\geq \int^{(\gamma^3\Delta^{2})}_{0} \frac{n}{e(\Delta^2+x)}dx=\frac{n}{e}\log(1+\gamma^3)
\]
and
\[
V=\sum^{(\gamma^3\Delta^{2})}_{i=1}\frac{n^2}{e^2(\Delta^2+i-1)^2}\leq \int^{\infty}_0\frac{n^2}{e^2(\Delta^2+x-1)^2}dx=\frac{n^2}{e^2(\Delta^2-1)}.
\]
All the above yield 
\begin{equation}\label{eq:ProbB2}
\Pr[{\cal B}_2]\leq \exp\left( -(\mu-C_1n)^2/(2V) \right)\leq \exp\left( -\Delta^2/C_5\right),
\end{equation}
where $C_5=C_5(\gamma)>0$ is sufficiently large number.

Let  $Y$ be the total number of disagreements that ever occur in $N(u)$ up to the first time that
either $D\nsubseteq B_{r-1}(w)$ or $|D|>\gamma^3\Delta^{2}$ occur or time $C_1n$ whichever happens
first. The variable $Y$ follows a generalized Poisson distribution with jumps of size 1. It is direct to
check that the maximum observed rate is at most $(\gamma^3\Delta^{2}+2\Delta)e\cdot dt/(\Delta n)\leq
10\gamma^3\Delta dt/n$, integrated over at most $C_1n$ time units. This is because the clock stops
when $|D|\geq \gamma^3 \Delta^{2}$ and since $G$ has girth at least 7 
it is only vertex $u$ that is adjacent to more than one element  of $N(u)$. Hence there are at most 
$\gamma^3\Delta^{2}+\Delta-1$ edges joining joining a disagreement with some vertex in $N(u)$ before the
clock stops.
Furthermore, disagreements on $N(u)$ may also be caused by incident edges which belong to  $G\oplus G^*_w$.
Each vertex in $v\in N(u)$ is incident to at most one edge which belongs to $G\oplus G^*_w$ and 
could cause disagreement in $v$. That is, $N(u)$ has at most  at most $\Delta$ such edges.

Applying Lemma \ref{lem:GenPoisson}, once more, for $Y$ with $\mu=10C_1\gamma^3\Delta$, $\alpha=1$
and $C=\gamma^{3/2}\Delta/\mu$ we get that
\[
\Prob{Y\geq \gamma^2\Delta}\leq \left( \frac{10eC_1 \gamma^3  \Delta }{\gamma^{3/2}\Delta}  \right)^{\gamma^{3/2}\Delta}
\leq \left( 10eC_1 \gamma^{3/2}  \right)^{\gamma^{3/2}\Delta} .
\]
Taking a union bound over the at most $\Delta^{r}$ vertices in $B_r(w)$ gives an upper bound for
the probability the event ${\cal B}_3$ happens and at the same time neither ${\cal B}_1$ nor ${\cal B}_2$
occur. That is
\begin{equation}
\Prob{\bar{\cal B}_1 \textrm{ and } \bar{\cal B}_2\textrm{ and } {\cal B}_3} \leq 
 \Delta^r  \left( 10eC_1 \gamma^{3/2}  \right)^{\gamma^{3/2}\Delta} \label{eq:almostProbB3}
%
\end{equation}
Letting ${\cal C}={\cal B}_1\cup{\cal B}_2$, we have that 
\begin{eqnarray}
\Prob{{\cal B}_3} &=& \Prob{ {\cal C},  {\cal B}_3}+
\Prob{ \bar{{\cal C}}, {\cal B}_3} \nonumber\\
&\leq &  \Prob{ {\cal C}}+\Prob{ \bar{{\cal C}}, {\cal B}_3}\nonumber \\
&\leq &  \Prob{ {\cal B}_1}+\Prob{ {\cal B}_2}+\Prob{ \bar{\cal B}_1 \textrm{ and } \bar{\cal B}_2\textrm{ and } {\cal B}_3 } 
\qquad\mbox{[union bound for $\Prob{\cal C}$ ]}\nonumber \\
&\leq &  \exp\left( -\Delta/C_6 \right),\label{eq:ProbB3}
\end{eqnarray}
where $C_6=C_6(\gamma)>0$. In the last inequality we used \eqref{eq:almostProbB3}, \eqref{eq:ProbB2} and \eqref{eq:ProbB1}.

As far as $\Prob{{\cal B}_4}$ is regarded, first recall that
${\cal B}_4$ denotes the event that there exists  $z\in B_{100}(w)$ such that 
either $|A(u)|\geq \gamma^3\Delta$ or  $|A^*(u)|\geq \gamma^3\Delta$.
Fix some vertex $z\in B_{100}(w)$. W.l.o.g. we consider the chain $X_t$.
There are two cases for $z$. The first one is that $z$ is occupied in $X_0$. 
The second one is $z$ is not occupied in $X_0$. Then  probability
that the vertex $z$ is updated  becomes occupied at least once 
up to time $C_1n$ is at most $2C_1e/\Delta$, regardless of the
rest of the vertices.

Fix some vertex $u\in B_{100}(w)$.  Let $J_u$ be the number of  vertices $z\in N(u)$ 
which are {\em unoccupied} in $X_0$ but they get into $A_u$. 
$J_u$ is dominated by the binomial distribution  with parameters $\Delta$ and 
$2C_1e/\Delta$, i.e.  ${\cal B}(\Delta,2C_1e/\Delta)$.
Using Chernoff's bounds we get that
\[
\Prob{J_u \geq \gamma^3\Delta/10}\leq \exp\left( -\gamma^3\Delta/10\right).
\]
Let $L_u$ be the number of vertices in $z\in N(u)$ which are occupied in
$X_0$. Since we have that $X_0$ is $400$-above suspicious for radius $R\gg 100$  
around $w$ and $u\in {\cal B}_{100}(w)$, it holds that $L_u\leq 400\Delta/\log \Delta$. 
Since $|A_u|=J_u+L_u$ we get that
$
\Prob{|A_u| \geq \gamma ^3\Delta }\leq \exp\left( -\gamma^3\Delta/10\right).
$
Taking a union bound over the at most $\Delta^{100}$  vertices in $ B_{100}(w)$ we
get that
\[
\Prob{\exists u\in B_{100}(w)\; s.t.\; |A_u| \geq \gamma^3\Delta }\leq \Delta^{100}\exp\left( -\gamma^3\Delta/10\right)
\leq \exp\left( -\gamma^3\Delta/20\right),
\]
where the last inequality holds for sufficiently large $\Delta$.
Working in the same way we get that
\[
\Prob{\exists u\in B_{100}(w)\; s.t.\; |A^*_u| \geq \gamma^3\Delta } 
\leq \exp\left( -\gamma^3\Delta/20\right),
\]
Combining the two inequalities above, 
there exists $C_7=C_7(\gamma)>0$ such that
\begin{equation}\label{eq:ProbB_4}
\Prob{{\cal B}_4}\leq \exp\left( -\Delta/C_7 \right).
\end{equation}
Plugging \eqref{eq:ProbB_4}, \eqref{eq:ProbB3}, \eqref{eq:ProbB2}, \eqref{eq:ProbB1}
and \eqref{eq:ProbBcA} into \eqref{eq:Basis4ProbA}, we get that
\begin{equation}\label{eq:A-Bound}
\Prob{\cal A}\leq \exp\left( -\Delta/C_8\right),
\end{equation}
for appropriate $C_8>0$. The lemma follows by plugging  \eqref{eq:A-Bound} and \eqref{eq:ProbB3} 
into \eqref{eq:GVsGstartBasicEq}.
\end{proof}

\section{Proof of Local Uniformity - Proof of Theorems \ref{thrm:ConcentrNoOfBlockedStatic}, \ref{thrm:UniformityWithBurnIn}}\label{sec:UniformityProofs}

In this section we prove the uniformity results (Theorems \ref{thrm:ConcentrNoOfBlockedStatic} and \ref{thrm:UniformityWithBurnIn})
that are presented in Section \ref{sec:LocalUniformity-sketch}.

\subsection{Proof of Theorem \ref{thrm:UniformityWithBurnIn}}\label{sec:thrm:UniformityWithBurnIn}

\noindent
In light of Lemma \ref{lemma:BurnIn}, Theorem \ref{thrm:UniformityWithBurnIn} follows as a corollary
from the following result which considers initial state for $(X_t)$ which is not heavy around $v$.

\begin{theorem}\label{thrm:Uniformity}
For  all $\delta, \epsilon>0$,  let 
$\Delta_0=\Delta_0(\delta,\epsilon), C=C(\delta,\epsilon)$.
For graph $G=(V,E)$ of maximum degree $\Delta\geq\Delta_0$ and girth $\geq 7$,
 all $\lambda<(1-\delta)\lambda_c(\Delta)$,
let $(X_t)$ be the  continuous (or discrete) time  Glauber dynamics on the hard-core model. 
 If  $X_0$ is  400-above suspicion for radius $R=R(\delta,\epsilon)>1$ for $v\in V$,
it holds that
\begin{equation}
\label{eq:uniformity-simple2}
\Prob{ 
(\forall t\in {\cal I} )  \quad
  \W_{X_t}(v)<\sum_{z \in N(v)} \omega^*(z)\Phi(z)+\epsilon \Delta }
  \geq 1-\exp\left( -\Delta/C \right),
\end{equation}
where the time interval   ${\cal I}=[Cn,n\exp\left( \Delta/C \right)]$.
\end{theorem}

We will use Lemmas  \ref{lemma:BurnIn}, \ref{lemma:LightnessOfX}
and  \ref{lem:ComparisonXVsXStar}  to complete the proof of 
Theorem \ref{thrm:Uniformity}.
For  an independent set $\sigma$ of $G$ and  $w\in V$, recall that
$
\R (\sigma,w)=\prod_{z \in N(w)}
\left(1-\frac{\lambda}{1+\lambda}\I_{z,v}(\sigma) \right),
$
where $\I_{z,w}(\sigma)=\indicator{ \sigma\cap\left(N(z)\setminus\{w\}\right) = \emptyset}$.

The following result which is the Glauber dynamics' 
 version of \eqref{SketchEq:ApproxFix}, in Section \ref{sec:SketchLocalUnifGibbs}.

\begin{lemma}\label{lem:ApproxVsExactFixPoint}
Let $\epsilon>0$, $R$, $C$ and $\lambda$ be as in Theorem \ref{thrm:Uniformity}.
Let $(X_t)$  be the continuous  time  Glauber dynamics on the hard-core model with fugacity 
$\lambda$ and underlying graphs $G$.
If  $X_0$ is  400-above suspicion for radius $R$ for $w\in V$,
then  we have that
\begin{eqnarray}\label{eq:ApproxVsExactFixPoint}
\Prob{  
(\forall t\in {\cal I} )  \quad \left| \R(X_t,w) - \omega^*(w) \right|\leq \epsilon/10 }
  \geq 1- \exp\left( -20\Delta/C \right),
\end{eqnarray}
where ${\cal I}=[Cn, n\exp\left( \Delta/C\right)]$
\end{lemma}

\noindent
The proof of Lemma \ref{lem:ApproxVsExactFixPoint}  makes use of the following result,
which is the Glauber dynamics' version of Lemma \ref{lemma:RApproxReccGirth6}, in Section 
\ref{sec:SketchLocalUnifGibbs}.

\begin{lemma}\label{lem:Approx-Glauber-BP}
For  $\delta, \gamma>0$,  let 
$\Delta_0=\Delta_0(\delta,\gamma), C=C(\delta,\gamma)$, $\hat{C}=\hat{C}(\delta,\gamma)$.
For all graphs $G=(V,E)$ of maximum degree $\Delta\geq\Delta_0$ and girth $\geq 7$,
 all $\lambda<(1-\delta)\lambda_c(\Delta)$,
let $(X_t)$ be the continuous time Glauber dynamics on the hard-core model.

Let $X_0$ be  400-above suspicion for radius $R\leq \Delta^{9/10}$ for $w\in V$.
Then,  for  $x\in B_{R/2}(w)$ and  $I=[t_0, t_1]$, where $t_0=Cn$, 
it holds that 
\begin{multline*}
\Prob{(\forall t\in I)
 \quad \left | 
\R(X_t, x) - \exp\left( -\frac{\lambda}{1+\lambda} \sum_{z \in N(x)}
\Expectation_{t_z}\left[ \R(X_{t_z}, z) \right] 
\right)
\right | 
\leq \gamma
} 
\\
\geq  1- \left(1+\frac{t_1-t_0}{n}\right) \exp\left( -{\Delta}/{\hat{C} }\right),
\end{multline*}
where  $\Expectation_{t_z}\left[ \R(X_{t_z}, z) \right]$ is the expectation  w.r.t. random time $t_z$, the last time
that vertex $z$ is updated prior to time $t$.
\end{lemma}

\noindent
Note that
$
\Expectation_{t_z}\left[ \R(X_{t_z}, z) \right]=\exp(-t/n) \  \R(X_{0},z)+\int^t_{0}\R(X_s,z)n\exp\left( (s-t)/n\right)ds.
$
The proof of Lemma \ref{lem:Approx-Glauber-BP}  appears in Section \ref{sec:lem:Approx-Glauber-BP}.

\begin{proof}[Proof of Lemma \ref{lem:ApproxVsExactFixPoint}]

Recall that ${\cal I}=[Cn, n\exp(\Delta/C]$.  
Let 
$R=\left \lfloor 30\delta^{-1}\log(6\epsilon^{-1}) \right \rfloor$
Assume that $C$ is sufficiently large such that  $C=(R+1)C_1$, where   $C_1$ is  specified later.
Let $T_0=(R+1)C_1 n$ and  $T_1= \exp(\Delta/C)$. 
Finally,  for $i\leq R$ let ${\cal I}_{i}:=[T_0-iC_1n,T_1] $.

Consider the continuous time Glauber dynamics $(X_t)$.
Also, consider times $t\geq T_0-RC_1 n$ . For each such time $t$ and positive integer $i\leq R$, we define
$$
\alpha_i:=\max\left| \Psi (\R(X_t, x))-\Psi(\omega^*(x) )\right|,
$$
where $\Psi$ is defined in \eqref{eq:potential-function}. The maximum is taken over all $t\in {\cal I}_{i}$ and over all vertices $x\in B_i(w)$.

An elementary observation about  $\alpha_i$  is that $\alpha_i\leq 3$ for every $i\leq R$.
To see why this holds,  note the following: For any  $z\in V$  and any independent sets 
$\sigma$, it  holds that
\[
\R(\sigma,z) \; = \; \prod_{r \in N(z)}
\left(1-\frac{\lambda \cdot \I_{r,z}(\sigma) }{1+\lambda } \right)
\; \ge \; ( 1 + \lambda )^{-\Delta} 
\; \ge \; e^{-\lambda \Delta}  \; \ge \; \eee^{-\eee},
\]
where in the last inequality we use the fact that $\Delta$ is sufficiently large, i.e. $\Delta>\Delta_0(\epsilon,\delta)$
and $\lambda<e/\Delta$.
Furthermore,  using the same arguments as above  we get that $\omega^*(z)\geq e^{-e}$, as well.
Since for any  $x\in V$  and any independent sets $\sigma$, we have $\R(\sigma,x), \omega^*(x)\in[ e^{-e}, 1]$,
\eqref{eq:PotentialVsRealDistance} implies  $\alpha_i\leq C_0=3$, for every $i\leq R$.

We prove our result by showing that typically  $\alpha_0$ is very  small. Then, the lemma follows by using
standard arguments. 
We  use  an inductive argument to show that $\alpha_0$ very small. 
We start by using  the fact that $\alpha_R\leq C_0$. Then we show that 
 with sufficiently large probability, if $\alpha_{i+1}\geq \epsilon/20$, then 
$\alpha_i \le (1- \gamma) \alpha_{i+1}$ where $0<\gamma<1$.

For any $i\leq R$, we use the fact that  there exists    $\hat{C} >0$  
such that with probability at least $1- \exp\left ( -\Delta/\hat{C} \right)$ 
the following is true: For  every  vertex $z\in B_i(w)$ it holds that
\begin{eqnarray}\label{eq:From:lemma:ApproxFixPoint}
(\forall t \in \mathcal{I}_i) \quad \left | 
\R(X_t, z) - \exp\left( -\frac{\lambda}{1+\lambda} \sum_{r\in N(z)}
\tilde{\omega}(r)
\right)
\right |  < \frac{ \epsilon  \delta }{40}
\end{eqnarray}

\noindent
where 
\begin{equation}\label{def:omega(r)}
\tilde{\omega}_t(r)=\exp(-C_1) \cdot \R(X_{t-C_1n},r)+\int^t_{t-C_1n}\R(X_s,r)n\exp\left[ (s-C_1n)/n\right]ds.
\end{equation}
\noindent
Eq.   \eqref{eq:From:lemma:ApproxFixPoint} is implied by Lemmas  \ref{lemma:LightnessOfX}, \ref{lem:Approx-Glauber-BP}.

Fix some $i\leq R$,   $z \in B_i(w)$ and  time $s\in \mathcal{I}_{i}$. 
We consider $X_s$ by conditioning on  $X_{s-C_1n}$. 
From the definition of the quantity $\alpha_{i+1}$ we get the following:
For any $x\in B_{i+1}(w)$  consider the quantity $\tilde{\omega}_s(x)$. 
We have that 
\begin{eqnarray}\label{eq:OmegaTildeCondition}
D_{v,i+1}(\tilde{\omega}_s,\omega^*)  \leq \alpha_{i+1}.
\end{eqnarray}

\noindent
We will show that if \eqref{eq:From:lemma:ApproxFixPoint} holds  for
$\R(X_s,z)$,  where $z\in B_i(w)$,  and $\alpha_{i+1}\geq \epsilon/20$, then 
we have that
\[
\left|  \Psi \left( \R(X_s,z) \right) -\Psi \left( \omega^*(z)\right)  \right| \leq (1-\delta/24)\alpha_{i+1},
\]

\noindent
For proving the above inequality, first note that if   $\R(X_s, z)$  satisfies  \eqref{eq:From:lemma:ApproxFixPoint},
then  \eqref{eq:PotentialVsRealDistance}  implies that

\begin{eqnarray}\label{eq:UniformityBound4PhiS}
\left| \Psi\left(  \R(X_s, z)  \right) - \Psi\left( \exp\left( -\frac{\lambda}{1+\lambda} \sum_{r\in N(z)}
 {\tilde{\omega}_s(r)}
\right)\right) \right| \; \leq \; \frac{\delta \epsilon}{12}.
\end{eqnarray}

\noindent
Furthermore, we have that  
\begin{eqnarray}
\lefteqn{
\left|  \Psi \left( \mathbf{R}(X_s,z) \right) -\Psi \left( \omega^*(z)\right)  \right| 
}
\hspace{.2in}
\nonumber
\\
&\leq & \frac{\delta \epsilon}{12} + 
 \left|  \Psi \left( \exp\left( -\frac{\lambda}{1+\lambda} \sum_{r\in N_z }
 {  \tilde{\omega}_s(r) } 
\right)\right)  -\Psi \left( \omega^*(z)\right)  \right|
 \quad\qquad\mbox{[from (\ref{eq:UniformityBound4PhiS})]} \nonumber \\
 &\leq &
 \frac{\delta \epsilon}{12} + 
 \left|  
 \Psi \left( \prod_{r\in N(z)}\left(1-\frac{\lambda
\tilde{\omega}_s(r) }{1+\lambda} \right) \right)
-\Psi \left( \omega^*(z)\right)  \right| \nonumber \\
&& + 
\left | \Psi \left( \prod_{r\in N(z)}\left(1-\frac{\lambda
\tilde{\omega}_s(r) }{1+\lambda} \right) \right)-  
\Psi  \left( \exp\left( -\frac{\lambda}{1+\lambda} \sum_{r\in N(z)}
 {\tilde{\omega}(r)} 
\right) \right) \right|, \qquad 
 \label{eq:TowardsApproxFixPointA}
\end{eqnarray}
where the last derivation follows from the triangle inequality.

From our assumptions
about  $\lambda, \Delta$   and the fact that  $\tilde{\omega}_s(r)\in [e^{-e},1]$, for  $r\in N(z)$,  we have that 
\begin{eqnarray}
\left | \prod_{r\in N(z)}\left(1-\frac{\lambda
\tilde{\omega}_s(r) }{1+\lambda} \right)-  
\exp\left( -\lambda \sum_{r\in N(z) }
 {\frac{\tilde{\omega}_s(r)}{1+\lambda}} 
\right) \right| \leq \frac{10}{\Delta}. \nonumber
\end{eqnarray}
The above inequality and (\ref{eq:PotentialVsRealDistance}) 
 imply that
\begin{eqnarray}
\left | \Psi \left( \prod_{r\in N(z)}\left(1-\frac{\lambda
\tilde{\omega}_s(r) }{1+\lambda} \right) \right)-  
\Psi  \left( \exp\left( -\frac{\lambda}{1+\lambda} \sum_{r\in N(z)}
 {\tilde{\omega}_s(r)} 
\right) \right) \right| \leq \frac{30}{\Delta}. \nonumber 
\end{eqnarray}

\noindent
Plugging the inequality above into (\ref{eq:TowardsApproxFixPointA}) we get that
\begin{eqnarray}
\left|  \Psi \left( \R(X_s,z) \right) -\Psi \left( \omega^*(z)\right )  \right|  
&\leq & 
\frac{\delta \epsilon}{12} + \frac{30}{\Delta} 
+\left|  \Psi \left( \prod_{r\in N(z) }\left(1-\frac{\lambda
\tilde{\omega}_s(r) }{1+\lambda } \right) \right) - \Psi \left( \omega^*(z)\right) \right| \nonumber \\
&\leq &
\frac{\delta \epsilon}{12} + \frac{60}{\Delta}+
D_{v,i}(F(\tilde{\omega}), \omega^*),
\label{eq:ProductVsExpontialWorstCase}
\end{eqnarray}
where the function $F$ is defined in \eqref{eq:BP-recursion}.  
Since   $\tilde{\omega_s}$  satisfies  \eqref{eq:OmegaTildeCondition},  Lemma \ref{lemma:ApproxFixPointEquations}  implies that
\begin{equation}
 D_{v,i}(F(\tilde{\omega}), \omega^*) \leq (1-\delta/6 )\alpha_{i+1}. \label{eq:ProductWorstCaseVsFixPoint}
\end{equation}

\noindent
Plugging (\ref{eq:ProductWorstCaseVsFixPoint}) into (\ref{eq:ProductVsExpontialWorstCase}) we get that
\begin{eqnarray}
\left|  \Psi \left( \R(X_s,z) \right) -\Psi \left( \omega^*(z)\right)  \right|  \leq 
\frac{\delta \epsilon}{12} + \frac{60}{\Delta}+ (1-\delta/6 )\alpha_{i+1} 
 \leq    (1-\delta/24)\alpha_{i+1}, \label{eq:FixTimeAlphaI}
\end{eqnarray}
\noindent
where the last inequality follows if we have $\alpha_{i+1}\geq \epsilon/20$.
Note that \eqref{eq:FixTimeAlphaI} holds provided that $\R(X_s,z)$ satisfies \eqref{eq:From:lemma:ApproxFixPoint}.

So as to  bound $\alpha_i$ we have to take the maximum over all times $t \in {\cal I}_i$ and vertices
$z\in B_i(w)$. So far, i.e. in \eqref{eq:FixTimeAlphaI}, we only considered  a fixed time $s\in {\cal I}_i$ and a fixed 
vertex $z$. 

Consider, now, a partition of  ${\cal I}_i$ into subintervals each of length $\frac{\epsilon^4\eta}{200\Delta}n$,
where the last part can be of smaller length. 
Let $T(j)$ be the $j$-th part, for $j\in \{1,\ldots, \lceil 200C_1\Delta/(\epsilon^4\eta )\rceil\}$.  
For some some vertex $x\in V$,  each $r \in  N(x)$ is updated  during the time period
$T(j)$ with probability  less than $\frac{\epsilon^4 \eta}{100\Delta}$, independently of the other vertices.

 Chernoff's bounds imply that  with probability at least $1-\exp\left( -\Delta \epsilon^3/3 \right)$, the number of
vertices in $S_2(x)$ which are updated  during the interval $T(j)$ is at most $\Delta\epsilon^3/3$. 
Furthermore, changing any $\Delta\epsilon^2/3$ variables in $S_2(x)$ can only change 
$\R(X_s,x)$ by at most $\epsilon^2/3$.
Consequently, $\Psi(\R(X_s,x))$ can change by only $\epsilon^{2}$ within $T(j)$.
From a union bound  over all subintervals $T(j)$ 
and all vertices  $x\in B_i(w)$,
there exists sufficiently large $C>0$
such that: 
\[
\Prob{\alpha_i=\max\{3\epsilon^{2}+(1-\delta/24)\alpha_{i+1}, \epsilon/20 \}}\geq 1-\exp\left( -52\Delta/C  \right).
\]
The fact that $\alpha_R\leq C_0$ and 
$R=\left \lfloor 20\delta^{-1}\log(6\epsilon^{-1}) \right \rfloor$, implies the following: 
With probability at least $1-\exp\left( -50\Delta/C\right)$  
for every $t\in {\cal I}$ it holds  that $\alpha_0\leq \epsilon/30$.
In turn, (\ref{eq:PotentialVsRealDistance}) implies that 
\begin{equation}\label{eq:ResultContTime}
|\R(X_t,v)-\omega^*(v)|\leq \epsilon/11.
\end{equation}
The lemma follows.
\end{proof}

We conclude  the technical results for Theorem \ref{thrm:Uniformity} by proving the following lemma.

\begin{lemma}\label{lem:ExpWVsR}
Let $\epsilon>0$, $R$, ${\cal I}$ and $\lambda$ be as in Theorem \ref{thrm:Uniformity}.
Let $(X_t)$  be the continuous time Glauber dynamics on the hard-core model with fugacity 
$\lambda$ and underlying graphs $G$.
Assume that  $X_0$  is $400$ above suspicion for $v$. Then, 
for  any $t\in {\cal I}$, any $\gamma>0$, there is $\hat{C}=\hat{C}(\gamma)>0$ such that
\[
\Prob{ \left| \W(X_t,v)  - \sum_{z \in N(v)}
    \Phi(z)\cdot  \Expectation_{t_z}\left[ \R(X_{t_z}, z) \right] 
    \right| >  \gamma\Delta} < \exp\left(-\Delta/\hat{C}  \right).
\]
\end{lemma}

\noindent
Recall that   $\Expectation_{t_z}\left[ \R(X_{t_z}, z) \right] $ is the expectation w.r.t. $t_z$ the time when
$z$ was last updated prior to time $t$, i.e.
$
\Expectation_{t_r}\left[ \R(X_{t_z}, z) \right] =\exp(-t/n)\R(X_0,z)+\int^t_0\R(X_s,z)n\exp\left[ (s-t)/n\right]ds.
$

\begin{proof} 
Consider, first, the graph $G^*_v$ and the dynamics $(X^*_t)$ such that $X^*_0=X_0$.
Condition on $X^*_0$ and on $X^*_t$ restricted to $V\setminus B_2(x)$ for all
$t\in {\cal I}$. 
Denote this conditional information by ${\mathcal F}$.

First we are going to show that $\ExpCond{\W (X^*_t,v)}{\cal F}$ and
$\sum_{z \in N(v)} \Phi(z) \cdot \ExpCond{\R(X^*_t,z)}{\cal F} $ are  very close.
From the definition of $\W(X^*_t,v)$ we have that
\[
\ExpCond{W (X^*_t,v)}{\cal F} = \sum_{z\in N(v)}\Phi (z) \cdot \ExpCond{\I_{z,v}(X^*_t)}{\cal F}.
\]
Let $c>0$ be such that $t/n=c$.  For $\zeta>0$ whose value is going to be specified later,  let $H(v)\subseteq N(v)$ be such that 
$z\in H(v)$ is $|N(z)\cap X^*_0|\geq \zeta^{-1}$. 
In \eqref{eq:ExpUXUpper} and \eqref{eq:ExpUXLower} we have shown that for
 $z\notin H(v)$ it holds that
\begin{equation}\label{eq:UzvVsRzExpectation}
\left | \ExpCond{\I_{z,v}(X^*_t)}{\cal F} - \ExpCondSub{t_z}{\R(X^*_{t_z},z)}{\FFF} \right|\leq \theta,
\end{equation}
where $0<\theta=\theta(c,\zeta)< 20(\zeta e^{c})^{-1}$ while (as in we previously defined)
\[
\ExpCondSub{t_z}{\R(X^*_{t_z},z)}{\FFF}=\exp(-t/n)\R(X^*_0,z)+\int^t_0\R(X^*_s,z)n\exp\left[ (s-t)/n\right]ds.
\]

\noindent
Since  $X^*_0$  is  400-above suspicion for radius $R$ around $v$,
it holds that $|H(v)|\leq 400\zeta\Delta$. We have that, 
 
\begin{eqnarray} 
\lefteqn{
\left | \ExpCond{\W(X^*_t,v)}{\cal F}-
\sum_{z \in N(v)} \Phi(z) \cdot \ExpCondSub{t_z}{\R(X^*_{t_z},z)}{\cal F} \right| 
}
\nonumber\\
&\leq &\left | \ExpCond{\W(X^*_t,v)}{\cal F}-
\sum_{z \notin H(v)} \Phi(z) \cdot \ExpCondSub{{t_z}}{\R(X^*_{t_z},z)}{\cal F} \right| +
\sum_{z \in H(v)} \Phi(z)\cdot \ExpCondSub{{t_z}}{\R(X^*_{t_z},z)}{\cal F} 
\nonumber \\
&\leq &\left | \ExpCond{\W(X^*_t,v)}{\cal F}-
\sum_{z \notin H(v)} \Phi(z) \cdot \ExpCondSub{{t_z}}{\R(X^*_{t_z},z)}{\cal F} \right| +
5000\zeta\Delta
\hspace{.25in}
 \mbox{[since $\max_{z}\Phi(z)\leq 12$]}\nonumber \\
&\leq &\left(12\theta+5000\zeta\right)\Delta.
\hspace*{8.375cm} \mbox{[from  \eqref{eq:UzvVsRzExpectation}]} \label{eq:ExpWVsSumExpR}
\end{eqnarray}
The fact that $\max_{z}\Phi(z)\leq 12$ is from  Theorem \ref{thrm:EigenVector}.

We proceed by showing  that $W(X^*_t,v)$ is sufficiently well concentrated about its expectation.
Conditioning on ${\cal F}$ the random variables $\I_{z,v}(X^*_t)$, for $z\in N(v)$,
are fully independent.
From Chernoff's bounds, there exists appropriate $C_1>0$ such that 
\begin{equation}\label{eq:TailWVsExpW}
\Prob{\left |\W(X^*_t,v) -\ExpCond{\W(X^*_t,v)}{\cal F} \right | >    \gamma \Delta/100}
\leq
\exp\left(-  \Delta/C_1 \right).
\end{equation}
From  \eqref{eq:TailWVsExpW} and  \eqref{eq:ExpWVsSumExpR} %
there exists $C_2>0$ such that 
that
\begin{equation}\label{eq:ExpWVsSumExpRWithTail}
\Prob{\left | W(X^*_t,v)-
\sum_{z \in N(v)} \Phi(z) \cdot \ExpCondSub{t_z}{\R(X^*_{t_z},z)}{\FFF}
 \right|\geq \gamma\Delta/50 }
\leq \exp\left( -\Delta/C_2\right).
\end{equation}

\noindent
Furthermore, using Lemma \ref{lem:ComparisonXVsXStar} with error parameter ${\gamma^2}$,  i.e. 
$|(X^*_t\oplus X_t) \cap B_2(v)| \le \gamma^2 \Delta$,  we get the following:
There exists appropriate  $C_3=C_3(\gamma)>0$ such that 
\begin{equation}\label{eq:WisLipschitzNewVer}
\Prob{|\W(X^*_t, v)-\W(X_t,v)|\le \gamma \Delta/40 }\geq 1- \exp\left( -\Delta/C_3\right).
\end{equation}
Also, (from Lemma \ref{lem:ComparisonXVsXStar} again)  with probability at least  $1-\exp\left( -\Delta/C_3 \right)$ it holds that
\begin{equation}\label{eq:ExpR^*VsExpR}
\left| \int^t_0\R(X_s,z)n\exp\left[ (s-t)/n\right]ds -\int^t_0\R(X^*_s,z)n\exp\left[ (s-t)/n\right]ds \right|\leq \gamma/600,
\end{equation}
 for every $z\in N(v)$.  The above follows by using the fact that changing the spin of any $\gamma^2\Delta$ vertices
in $X^*_t(B_2(z))$ changes  $\R(X^*_t,z)$ by at most $\gamma/1000$. 

Noting that $\Phi(z)\leq 12$, for any $z$, the lemma follows by 
combining  \eqref{eq:ExpR^*VsExpR}, \eqref{eq:WisLipschitzNewVer} and \eqref{eq:ExpWVsSumExpRWithTail}.
\end{proof}

\subsection{Local Uniformity for the Glauber Dynamics: Proof of Theorem \ref{thrm:Uniformity}}

Using Lemmas \ref{lem:ApproxVsExactFixPoint} and \ref{lem:ExpWVsR}, in this section we prove
Theorem \ref{thrm:Uniformity}. Recall that Theorem  \ref{thrm:UniformityWithBurnIn} follows
as a corollary of Theorem \ref{thrm:Uniformity} and Lemma \ref{lemma:BurnIn}.

\begin{proof}[Proof of Theorem \ref{thrm:Uniformity}]

For a vertex $u\in N(v)$ consider $G^*_u$. Consider also the continuous time dynamics
 $(X^*_t)$ such that $X^*_0=X_0$.  

We condition on the restriction of $(X^*_t)$ to $V\setminus B_2(u)$,  for every $t\in {\cal I}$. 
We denote this by ${\cal F}$.  
Fix some $t\in {\cal I}$.   Since $u\in B_R(v)$ and $X_0$ is 400-above suspicion for radius $R$ around $v$
we get that 
\begin{eqnarray}
\lefteqn{
\ExpCondSub{s}{\R(X^*_s,u)}{\cal F} 
}
\nonumber
\\
&=& \exp(-t/n)\R(X^*_0,u)+\int^t_0\R(X^*_t,u)n\exp\left((s-t)/n\right) \nonumber \\
&=&  \ExpCondSub{s}{\exp\left( - \frac{\lambda}{1+\lambda}\sum_{z\in N(u)} \I_{z,u}(X^*_s)   +O\left(1/{\Delta}\right) \right)}{\cal F}
\nonumber \\
&=&
\exp\left( - \frac{\lambda}{1+\lambda}\sum_{z\in N(u)}  \ExpCondSub{s}{ \I_{z,u}(X^*_s)}{\cal F}   +O\left(1/{\Delta}\right) \right)
\qquad\mbox{[due to conditioning on ${\cal F}$]}\nonumber \\
&\leq &
\exp\left( - \frac{\lambda}{1+\lambda}\sum_{z\in N(u)}  \ExpCondSub{s}{ \R(X^*_s,z)}{\cal F}  + \theta\lambda\Delta+O\left({1}/{\Delta}\right) \right),
\label{eq:ExpRRecurrence}
\end{eqnarray}
where in the last inequality we use \eqref{eq:UzvVsRzExpectation}. Note that so as apply
\eqref{eq:UzvVsRzExpectation} $X^*_0(u)$ should be sufficiently ``light". This is guaranteed from
our assumption that $u\in B_R(v)$ and $X_0$ is 400-above suspicious for  radius $R$ from $v$.

Furthermore,   \eqref{eq:From:lemma:ApproxFixPoint} and Lemma \ref{lem:ComparisonXVsXStar}
imply  the following: There exists $C_1>0$  such that with probability   at least $1- \exp\left ( -{\Delta}/{ C_1} \right )$,
 we have that
\begin{equation}\label{eq:Again19}
(\forall t \in \mathcal{I}) \quad \left | 
\R(X^*_t, u) - \exp\left( -\frac{\lambda}{1+\lambda} \sum_{r\in N^*(u)}
\hat{\omega}(r)
\right)
\right |  < \gamma,
\end{equation}
where 
\[
\hat{\omega}(r)=\exp(-t/n)\R(X^*_0,r)+\int^t_0\R(X^*_s,r)n\exp\left[ (s-t)/n\right]ds.
\]
Note   that  for every $r\in N^*(u)$ we have  $\hat{\omega}(r) = \ExpCondSub{t_r}{ \R(X^*_{t_r},r)}{\cal F}$.
Using  this observation, we plug  \eqref{eq:ExpRRecurrence}  into  \eqref{eq:Again19} and get
\begin{equation}\label{eq:RX*TailBound}
\Prob{ 
\left| \R(X^*_t, u)-   \hat{\omega}(u) \right| \geq 10e \theta +\gamma}
\leq \exp\left( - {\Delta }/{C_1} \right).
\end{equation}
In the above inequality  we used the fact that $\lambda\Delta<2e$.

Consider the continuous time Glauber dynamics $(X_t)$.
From  Lemma \ref{lem:ComparisonXVsXStar}  and \eqref{eq:RX*TailBound} 
there  exists $C_3>0$ such that for $X_t$ the following holds
\begin{equation}\label{eq:RTailBound}
\Prob{
\left| \R(X_t, u)-   \tilde{\omega}(u) \right| \geq 20e \theta +2\gamma}
\leq \exp\left( -\Delta/C_3 \right),
\end{equation}
where
\[
\tilde{\omega}(z)=\exp(-t/n)\R(X_0,z)+\int^t_0\R(X_s,z)n\exp\left[ (s-t)/n\right]ds.
\]

\noindent
Furthermore a simple union bound  over $u\in N(v)$ and   \eqref{eq:RTailBound}  gives that
\begin{equation}\label{eq:UnionRTailBound}
\Prob{ \forall u\in N(v)\quad
\left| \R(X_t, u)-   \tilde{\omega}(u) \right| \geq 20e \theta +2\gamma}
\leq \Delta\exp\left( -\Delta/C_3  \right).
\end{equation}
Taking  sufficiently small   $\theta,\gamma$ in \eqref{eq:UnionRTailBound}  
and using Lemma \ref{lem:ExpWVsR} we get that
\begin{equation}\label{eq:WVsRNotStar}
\Prob{\left|\W(X_t,v)-\sum_{w\in N(v)}\Phi(z)\cdot \R(X_t,w) \right|> \epsilon \Delta/15 }\leq 
\exp\left(-\Delta/C_4 \right),
\end{equation}
for appropriate $C_4>0$.
Furthermore, applying Lemma \ref{lem:ApproxVsExactFixPoint}, 
for each $w \in N(v)$ and using   \eqref{eq:WVsRNotStar} 
yields
\begin{equation}\label{eq:WVsOmega*} 
\Prob{\left| \W(X_t,v)-\sum_{w\in N(v)}\Phi(w)\cdot \omega^*(w) \right|>\epsilon \Delta/2 }\leq 
\exp\left( -\Delta/C_5  \right),
\end{equation}
for appropriate $C_5>0$.
The above inequality establishes the desired result for a fixed $t \in \mathcal{I}$.

Now we will prove that \eqref{eq:WVsOmega*} holds for all $t\in{\cal I}$. 
Consider a partition of the time interval ${\cal I}$ into subintervals each of length $\frac{\psi^2}{\Delta}n$,
where the last part can be of smaller length. The quantity $\psi>0$ is going to be specified later.  
Also, let $T(j)$ be the $j$-th part.  

Each $z \in  B_2(v)$ is updated at least once during the time period $T(j)$ with probability  less than 
$2\frac{\psi^2}{\Delta}$, independently of the other vertices. Note that $|B_2(v)|\leq \Delta^2$.
Clearly, the number of vertices in $B_2(v)$ which are updated during $T_i(j)$ is dominated by
${\cal B}(\Delta^2, 2\psi^2/\Delta)$.
Chernoff's bounds imply that  with probability at least $1-\exp\left( -20\Delta \psi^2 \right)$, the number of
vertices in $B_2(v)$ which are updated  during the interval $T(j)$ is at most $20\psi^2\Delta$. 
Furthermore, changing any $2\Delta\psi^2$ variables in $B_2(v)$ can only change
the weighted sum of unblocked vertices in $N_v$ by at most   $20C_0\psi^2\Delta$.
Taking sufficiently small  $\psi>0$ we get the following:
\begin{equation}
\Prob{\left| \W(X_t,v)-\sum_{w\in N(v)}\Phi(w)\cdot \omega^*(w) \right|>\epsilon \Delta }\leq 
\exp\left( -2\Delta/C_b  \right).
\end{equation}
The above completes the proof of Theorem \ref{thrm:Uniformity} 
for the case where $(X_t)$ is the continuous time process.

The discrete time result follows by working as follows: instead of $\W(X_t,v)$ we consider the ``normalized" 
variable $\Lambda(X_t,v)=\frac{\W(X_t,v)}{\Delta}$.  
Rephrasing \eqref{eq:WVsOmega*} in terms of $\Lambda (X_t,v)$ we have, for a specific $t\in\mathcal{I}$:
\begin{equation}\label{eq:NormWVsOmega*} 
\Prob{\left| \Lambda(X_t,v)-\Delta^{-1}\sum_{w\in N(v)}\Phi(w)\cdot \omega^*(w) \right|>\epsilon/2 }\leq 
\exp\left( -\Delta/C_5  \right).
\end{equation}
Note that  $\Lambda(X_t,v)$ satisfies the  Lipschitz and  total influence conditions of Lemma \ref{lemma:ContVsDiscSmallDelta}.
Hence by Lemma \ref{lemma:ContVsDiscSmallDelta} the result for
the discrete time process holds.
\end{proof}

\subsection{Approximate recurrence for Glauber dynamics - Proof of Lemma \ref{lem:Approx-Glauber-BP}}\label{sec:lem:Approx-Glauber-BP}

\noindent
Consider $G^*_x$ and let $(X^*_t)$ be the Glauber dynamics on $G^*_x$ with fugacity
$\lambda>0$ and let $X^*_0=X_0$. Also assume that $(X^*_t)$ and $(X_t)$ are maximally
coupled.

Condition on $X^*_0$, let ${\cal F}$ be the $\sigma$-algebra generated by  $X^*_t$ restricted to 
$V\setminus B_2(x)$ for all $t\in {I}$. 
Fix $t \in {I}$.  Let $c>0$ be such that $t/n=c$, i.e.  $c$ is a large constant.  Recalling the definition of
$\R(X^*_t,x)$, we have that 
\begin{eqnarray}
\R(X^*_t,x)&= &\prod_{z\in N(x)}\left(1-\frac{\lambda}{1+\lambda} \I_{z,x}(X^*_t) \right)  
\nonumber \\
&=& \exp\left( - \frac{\lambda}{1+\lambda}\sum_{z\in N(x)} \I_{z,x}(X^*_t)   +
O\left(1/{\Delta}\right) \right).\label{eq:RXtVsExpUs} 
\end{eqnarray}
%
%
Let ${\bf Q}(X^*_t)=\sum_{z\in N(x)} \I_{z,x}(X^*_t)$. Conditional on ${\cal F}$, the quantity ${\bf Q}(X^*_t)$ is 
a sum of $|N(x)|$ many independent random variables in $[0,1]$.  Applying Azuma's inequality,
 for $0\leq \gamma\leq (3e)^{-1}$,  we have
\begin{equation}\label{eq:Tail4QCX^*T}
\Prob{
|\ExpCond{ {\bf Q}(X^*_t) }{\cal F}-{\bf Q}(X^*_t)| \geq \gamma\Delta
} \leq 2\exp\left( -{\gamma^2}\Delta/2 \right).
\end{equation}

\noindent
Combining  the fact that  $\ExpCond{ {\bf Q}(X^*_t) }{\cal F}=\sum_{z\in N(x)} \ExpCond{\I_{z,x}(X^*_t)}{\cal F}$
with    \eqref{eq:Tail4QCX^*T} and \eqref{eq:RXtVsExpUs} we get that
\begin{equation}\label{eq:Tail4RX^*T}
\Prob{
\left | \R(X^*_t,x) - \exp\left(- \frac{\lambda}{1+\lambda} \sum_{z\in N(x)}\ExpCond{\I_{z,x}(X^*_t)}{\cal F } \right) \right | \geq 3\gamma{\lambda}\Delta
} \leq 2\exp\left( -\gamma^2 \Delta/2 \right).
\end{equation}

\noindent
 For every $z\in N^*(x)$, it holds that
\begin{eqnarray}
\lefteqn{
\ExpCond{\I_{z,x}(X^*_t)}{\cal F} 
} \hspace{.2in}
\nonumber 
\\
&=&\prod_{y\sim N(z)\setminus\{x \}} \ExpCond{\mathbf{1}\{ y\notin X^*_t\}}{\cal F} \nonumber \\
&=&\prod_{y\sim N(z)\setminus\{x \}} 
\left(\Pr[t_y=0] \cdot \mathbf{1}\{y\notin X^*_0 \}+
 \ExpCond{\mathbf{1}\{ y\notin X^*_t\}\cdot \mathbf{1}\{t_y>0\} }{\cal F} \right), \label{eq:1stGroup4EUHalf}
\end{eqnarray}
where $t_y$ is the time that vertex $y$ is last updated prior to time $t$ and it is defined to be equal to zero if $y$
is not updated prior to $t$. Note, for any $0\leq s \leq t$,
it holds that $\Pr[t_y \leq s]=e^{-(t-s)/n}$. Also, we have that
\begin{eqnarray}
 \ExpCond{\mathbf{1}\{ y\notin X^*_t\}\cdot \mathbf{1}\{t_y>0\} }{\cal F} &=& 
 \ExpCond{ \ExpCond{\mathbf{1}\{ y\notin X^*_t\}\cdot \mathbf{1}\{t_y>0\} }{{\cal F}, t_y}}{\cal F} \nonumber\\
&=&\int^t_0 \left( 1-\frac{\lambda}{1+\lambda} \I_{y,z}(X^*_s) \right)n\exp\left[ (s-t)/n\right]ds, 
\label{eq:1stGroup4EUHalfB}
\end{eqnarray}
where the last equality follows because we are using $G^*_x$ and $(X^*_t)$.
The use of $G^*$ and $(X^*_t)$ 
ensures that the configuration in $V\setminus B_2(x)$ is never
affected by that in $B_2(x)$.  For this reason,  if  $y$ is updated
at time $s\in I$, then the probability for it to be occupied, given ${\cal F}$, is exactly 
$\frac{\lambda}{(1+\lambda)} \I_{y,z}(X^*_s)$. That is, the configuration outside
$B_2(x)$ does not provide any information for $y$ but the value of $\I_{y,z}(X^*_s)$.

Plugging \eqref{eq:1stGroup4EUHalfB} into \eqref{eq:1stGroup4EUHalf} we get that
\begin{eqnarray}
\lefteqn{
\ExpCond{\I_{z,x}(X^*_t)}{\cal F} 
} \hspace{.2in}
\nonumber 
\\
&=&\prod_{y\sim N(z)\setminus\{x\}}\left[ \exp\left(-t/n \right)\mathbf{1}\{y\notin X^*_0\} -
\int^t_0\left(1- \frac{\lambda}{1+\lambda} \I_{y,z}(X^*_s)\right)n\exp\left[ (s-t)/n\right]ds \right] \nonumber \\
&=& \prod_{y\sim N(z)\setminus\{x\}}\left[ 1-\exp\left(-t/n \right)\mathbf{1}\{y\in X^*_0\} -
\int^t_0 \frac{\lambda}{1+\lambda} \I_{y,z}(X^*_s)n\exp\left[ (s-t)/n\right]ds\right]. \qquad \label{eq:1stGroup4EU}
\end{eqnarray}

\noindent
For appropriate $\zeta\in (0,1)$, which we define later, let $H(x)\subseteq N^*(x)$ be such that 
$z\in H(x)$ if $|N^*(z)\cap X^*_0|\geq 1/\zeta$. 

Noting that  each  integral in \eqref{eq:1stGroup4EU} is less than $\lambda$, for every $z\notin H(x)$, 
we get that 
\begin{equation}\label{eq:EUzvVsIntegral}
\ExpCond{\I_{z,x}(X^*_t)}{\cal F} =  (1+\delta) 
\prod_{y\in  N(z)\setminus\{x\}}\left( 1-\int^t_0 \frac{\lambda}{1+\lambda} \I_{y,z}(X^*_s)n\exp\left[ (s-t)/n\right]ds\right), 
\end{equation}
where $|\delta|\leq 4(\zeta e^{c})^{-1}$.

Recall that for some vertex $y$ in $G^*_x$ we let $\ExpCondSub{t_y}{\cdot}{\cal F}$, denote the expectation w.r.t. $t_y$, the random
time that $y$ is updated prior to time $t$. It holds that
\[
\mathbb{E}_{t_y}[\I_{y,z}(X^*_{t_y})|{\cal F}]=\exp(-t/n)\I_{y,z}(X^*_{0})+\int^t_0
\I_{y,z}(X^*_{s})
n\exp\left[ (s-t)/n\right]ds.
\]

\noindent
For every $y\in N(z)\backslash\{x\}$, where $z\notin H(x)$ it holds that
\begin{eqnarray} \label{eq:EUyzVsIntegral}
\ExpCondSub{t_y}{ \I_{y,z}(X^*_{t_y})}{\cal F} -
\int^t_0  \I_{y,z}(X^*_s)n\exp\left[ (s-t)/n\right]ds  &=& \exp\left( -t/n\right)\I_{y,z}(X^*_0) \nonumber \\
&\le &  \exp\left( -t/n\right) \;\le \; \exp(-c). \quad
\end{eqnarray}

\noindent
Since $\lambda<e/\Delta$,  \eqref{eq:EUzvVsIntegral} implies that
there is a quantity $\theta$, with   $0<\theta \leq 20(\zeta e^{c})^{-1}$, such that 
\begin{eqnarray}
\ExpCond{\I_{z,x}(X^*_t)}{\cal F}&\leq & \prod_{y\in  N(z)\setminus\{x\}}\left( 1-\int^t_0\frac{\lambda}{1+\lambda}\I_{y,z}(X^*_s)n\exp\left[ (s-t)/n\right]ds  \right) +\theta/2
\nonumber\\
&\leq & \prod_{y\in  N(z)\setminus\{x\}}\left( 1-\ExpCondSub{t_y}{\frac{\lambda}{1+\lambda}\I_{y,z}(X^*_{t_y})}{\FFF}\right) +\theta
\qquad \qquad \mbox{[from \eqref{eq:EUyzVsIntegral}]}
\nonumber\\
&=&\prod_{y\in  N(z)\setminus\{x\}}\left( 1-\ExpCondSub{s}{\frac{\lambda}{1+\lambda}\I_{y,z}(X^*_{s})}{\FFF}\right) +\theta,
\nonumber
\end{eqnarray}
where in the last derivation, we substituted the variables $t_y$, for $y\in  N(z)\setminus\{x\}$, with a new random
variable $s$ which follows the same distribution as $t_y$.  Note that the variables $t_y$ are identically distributed.

Given  the $\sigma$-algebra $\cal F$, the variables $\I_{y,z}(X^*_{s})$, for $y\in  N(z)\setminus\{x\}$, are independent
with each other, this yields 
\begin{eqnarray}
\ExpCond{\I_{z,x}(X^*_t)}{\cal F}
&=&\ExpCondSub{s}{\prod_{y}\left(1-\frac{\lambda}{1+\lambda} \I_{y,z}(X^*_{s})\right) }{\FFF}+\theta\nonumber\\
&=& \ExpCondSub{s}{\R(X^*_s,z)}{\FFF}+\theta, \label{eq:ExpUXUpper}
\end{eqnarray}
where the last derivation follows from the definition of $\R(X^*_s,z)$.
In  the same manner, we get that
\begin{equation}
\ExpCond{\I_{z,x}(X^*_t)}{\cal F} \geq \ExpCondSub{s}{\R(X^*_s,z)}{\FFF}- \theta, \label{eq:ExpUXLower}
\end{equation}
for every $z\notin H(x)$.

Since  $X^*_0$  is 400 above suspicion for radius $R$, around $w$ and  $x\in B_R(w)$,
we have that  $|H(x)|\leq 400\zeta\Delta$. 
This observation
and \eqref{eq:ExpUXUpper}, 
\eqref{eq:ExpUXLower}  \eqref{eq:Tail4RX^*T},  yield that
there exists $C'>0$ such that
\begin{equation}\label{eq:FixTimeXStarRecurence}
\Prob{
\left |  \R(X^*_t,x) -\exp\left( - \frac{\lambda}{1+\lambda} \sum_{z\in N(x)}\ExpCondSub{s}{\R(X^*_{s},z)}{\cal F }\right) \right | \geq 
7(\theta+400\zeta+3\gamma)
} \leq \exp\left( -C'\Delta  \right),
\end{equation}
where we use the fact   $\frac{\lambda}{1+\lambda}\Delta<e$ and $\theta, \zeta,\gamma$ are sufficiently small.

So as to get from $(X^*_t)$ to $(X_t)$  we use Lemma \ref{lem:ComparisonXVsXStar}, 
with parameter $\gamma^3$. That is, we have that
\[
\Prob{\exists s \in {I} \; |(X_s\oplus X^*_s) \cap S_2(x)|\geq \gamma^3 \Delta}\leq \exp\left( -\Delta/C''\right),
\]
for some sufficiently large constant  $C''>0$. This implies that
\begin{equation}\label{eq:RVsExpERTail-RealChain}
\Prob{ \exists t \in {I} \; | \R(X^*_t,x)-\R(X_t,x)| \geq \gamma^2 } \leq \exp\left( - \Delta/C''\right),
\end{equation}
since changing any $\Delta\gamma^3$ variables in $S_2(x)$ can only change $\R(X^*_s,x)$ by 
at most $\gamma^2$. 

With the same observation 
we also get that  with probability at least $1-\exp\left( - \Delta/C''\right)$ it holds that
\begin{equation}\label{eq:Expect-RealChain}
\left| \int^t_0\R(X^*_s,x)n\exp\left[ (s-t)/n\right]ds- \int^t_0\R(X_s,x)n\exp\left[ (s-t)/n\right]ds\right| \leq 2\gamma^2.
\end{equation}

\noindent
Plugging \eqref{eq:RVsExpERTail-RealChain},
\eqref{eq:Expect-RealChain} into \eqref{eq:FixTimeXStarRecurence} and taking appropriate $\gamma,\zeta$ the following is true:
There exists $\hat{C} >0$ such that 
\[
\Prob{ 
\left |  \R(X_t,x) -\exp\left( - \frac{\lambda}{1+\lambda} \sum_{r\in N(x)} 
{\cal E}(r)
\right) \right | \geq 
\frac{\eta \epsilon}{20C_0}} \leq \exp\left[ -\Delta/\hat{C}  \right],
\]
where
\[
{\cal E}(r)=\exp[-t/n] \cdot \R(X_{0},r)+\int^t_{0}\R(X_s,r)n\exp\left[ (s-t)/n\right]ds.
\]

\noindent
At this point, we remark that   the above tail bound holds for a fixed $t\in I$. For  our purpose,  we need a tail bound which
holds for {\em every} $t\in I$.

Consider a partition of the time interval $I$ into subintervals each of length $\frac{\zeta^3}{200\Delta}n$,
where the last part can be of smaller length.  Let $T(j)$ be the $j$-th part.   Each $z \in  S_2(x)$ is updated  during the time period
$T(j)$ with probability  less than $\frac{\zeta^3}{100 \Delta}$, independently of the other vertices. 

Note that $|S_2(x)|\leq \Delta^2$.
 Chernoff's bounds imply that  with probability at least $1-\exp\left( -\Delta \zeta^3 \right)$, the number of
vertices in $S_2(x)$ which are updated more than once during the time interval $T(j)$ is at most $\zeta^3\Delta$. 
Also, changing any $\Delta\zeta^3$
variables in $S_2(x)$ can only change $\R(X_s,x)$ by at most $\zeta^2$.

The lemma follows by taking a union bound over
all $T(j)$ for $j\in \{1,\ldots, \lceil 200|I|\Delta/(\zeta^3 )\rceil\}$
and all vertices  $z\in B_i(x)$.

\subsection{Local Uniformity for the Gibbs Distribution: Proof of Theorem \ref{thrm:ConcentrNoOfBlockedStatic}}

\begin{proof}[Proof of Theorem \ref{thrm:ConcentrNoOfBlockedStatic}]
 Let ${\cal F}$ be the $\sigma$-algebra generated by the configuration of $v$ and the vertices at
distance greater than $2$ from $x$, i.e.  $V\setminus B_2(v)$.
Conditioning on ${\cal F}$, $S_v$ is a sum of $|N(v)|$ many $0$-$1$ independent
random variables. From Azuma's inequality, 
for any  fixed $\gamma>0$, we have that 
\begin{equation}\label{eq:Concent4Sv}
\Prob{\left| S_v-\ExpCond{S_v}{\cal F}\right|>\gamma\Delta}\leq 2\exp\left( -{\gamma^2\Delta}/{ 2} \right).
\end{equation}

\noindent
Working as in the proof of Theorem \ref{thrm:CnvrgLoopyBP2Correct} (i.e. for \eqref{eq:ExUVsRUpperGirth6}, \eqref{eq:ExUVsRLowerGirth6}) we get the following: For each $z\in N(v)$ it holds that
\[
\left| \ExpCond{\I_{z,v}(X)}{\cal F} - \R(X,z)\right|\leq 10e^{e}\lambda.
\]
Note that, given $\cal F$ the quantity $\R(X,z)$ is uniquely specified.

From the above we get that 
\begin{equation}
\ExpCond{S_v}{\cal F} = \sum_{z\in N(v)}\ExpCond{\I_{z,v}(X)}{\cal F} 
\;=  \;\sum_{z\in N(v)}\R(X,z) +\zeta, \label{eq:ExSvVsRUpperGirth6}
\end{equation}
where $|\zeta|\leq e^{5e}$. Furthermore, from Lemma \ref{lemma:RConGirth6} we have that
for every $w\in V$ and every $\theta>0$, there exists $C_0>0$ such that 
\begin{equation}
\Prob{\left| \R(X, w) - \omega^*(w) \right|\leq \theta }
  \geq 1- \exp\left(-\Delta/C_0\right). \label{eq:TailFrom-RConFixGirth6}
\end{equation}
From \eqref{eq:TailFrom-RConFixGirth6}, \eqref{eq:ExSvVsRUpperGirth6} and a simple union
bound we get the following: for every $\gamma'>0$, the exists $C_a>0$ such that 
\begin{equation}
\Prob{\left| \ExpCond{S_v}{\cal F} - \sum_{z\in N(v)}\omega^*(z) \right|\leq \gamma'\Delta }
  \geq 1- \exp\left(  -\Delta/C_a\right). \label{eq:TailFrom-RConGirth6}
\end{equation}
The theorem follows by combining \eqref{eq:Concent4Sv} and \eqref{eq:TailFrom-RConGirth6}.
\end{proof}

\section{Rapid Mixing for Glauber dynamics: Proof of Theorem \ref{thrm:RapidMixingMain}}
\label{sec:thrm:RapidMixingMain}

The following lemma considers a worst case of neighboring independent sets. It states some upper bounds
on the Hamming distance after $Cn$ and $Cn\log\Delta$ steps in the coupling.

\begin{lemma}\label{lemma:ArbitraryPair}
For $\delta>0$,  $0<\epsilon<1$ and $C>10$ let $\Delta\geq \Delta_0$. 
Consider a graph $G=(V,E)$ of maximum degree $\Delta$ and let
 $\lambda\leq (1-\delta)\lambda_c({\Delta})$. 
Let $(X_t), (Y_t)$  be the Glauber dynamics on the hard-core model with fugacity 
$\lambda$ and underlying graphs $G$. Assume that the two chains are maximally coupled.
Then, the following is true:

Assume  $X_0,Y_0$ to be such that $X_0\oplus Y_0=\{v\}$  and    $T=Cn/\epsilon$. Then  it holds that
\begin{enumerate}
\item $\Exp{|X_T \oplus Y_T |}\leq \exp\left( 3C/\epsilon \right)$
\item $\Exp{|X_{T\log \Delta} \oplus Y_{T\log \Delta} |} \leq \Delta^{3C/\epsilon}$
\item Let ${\cal E}_T$ be the event that at some time $t\leq T$, $|X_t\oplus Y_t|>\Delta^{2/3}$. Then
\[
\Exp{|X_T\oplus Y_T|\cdot \mathbf{1}\{{\bf E}_T\}}<\exp\left( -\sqrt{\Delta}\right).
\]
\item Let $S_{T\log\Delta}$ denote the set of disagreements of $(X_{T\log\Delta},Y_{T\log\Delta})$, 
that are 200-suspect for radius $2\Delta^{3/5}$. Then $\Exp{|S_{T\log \Delta}|}\leq \exp\left( -\sqrt{\Delta}\right)$.
\end{enumerate}
\end{lemma}

\noindent
The proof of Lemma \ref{lemma:ArbitraryPair} appears in Section \ref{sec:lemma:ArbitraryPair}.

The above lemma, i.e. Lemma \ref{lemma:ArbitraryPair}.4, shows that from a worst case pair of neighboring
independent sets, after $O(n\log \Delta)$ steps, all the disagreements are likely to be ``nice" in the sense of being
above suspicion. The heart of rapid mixing proof will be the following results, which shows that for a pair
of neighboring independent sets that are ``nice" there is a coupling of $O(n)$ steps of the Glauber dynamics where the
expected Hamming distance decreases. Also, at the end of this $O(n)$ step coupling, it is extremely unlikely
 that there are any disagreements that are not ``nice".

\begin{lemma}\label{lemma:DistanceFromLight}
Let  $C'>10$, $\epsilon>0$ and $\Delta_0=\Delta_0(\epsilon)$. For any graph $G=(V,E)$ on $n$ vertices
and maximum degree $\Delta>\Delta_0$ and girth $g\geq 7$ the following holds:

Suppose that $X_0,Y_0$ differ only at $v$, while $v$ is $400$-above suspicion for $R$, where
$\Delta^{3/5}\leq R\leq 2\Delta^{3/5}$.  For  $T_m=C'n/\epsilon$ we have that
\begin{enumerate}
\item $\Exp{|X_{T_m}\oplus Y_{T_m}|}\leq 1/3$
\item Let  $ {\cal Z}$ denote the event that there exists a 200-suspect disagreement for $R'=R-2\sqrt{\Delta}$
at time $T_m$. Then it holds
\[
\Prob{{\cal Z}}\leq 2\exp(-2\sqrt{\Delta}).
\]
\end{enumerate}
\end{lemma}

\noindent
The proof of Lemma \ref{lemma:DistanceFromLight} appears in Section \ref{sec:lemma:DistanceFromLight}.

\begin{proof}[Proof of Theorem \ref{thrm:RapidMixingMain}]

The proof of the theorem is very similar to the proof of Theorem 1 in \cite{DFHV}. In particular, 
Lemma \ref{lemma:ArbitraryPair} is analogous to Lemma 10 in \cite{DFHV}. Similarly,
Lemma \ref{lemma:DistanceFromLight} is analogous to Theorem 11 in \cite{DFHV}.
Furthermore, working as for  Lemma 12 in \cite{DFHV} we get the following 
result, which ties together Lemma \ref{lemma:DistanceFromLight} and Lemma \ref{lemma:ArbitraryPair}.
It shows that for a worst case initial pairs of independent sets, after $O(n\log\Delta)$  steps, the expected Hamming
distance is small.

Let  $C'>10$, $\epsilon>0$ and $\Delta_0=\Delta_0(\epsilon)$. For any graph $G=(V,E)$ on $n$ vertices
and maximum degree $\Delta>\Delta_0$ and girth $g\geq 7$ the following holds:
Let $X_0,Y_0$ be independent sets which disagree on a single vertex $v$ that is 400-above suspicion
for radius $R=2\Delta^{3/5}$. Let $T=\frac{C'n\log \Delta}{\epsilon}$. Then,
\[
\Exp{|X_T\oplus Y_T|}\leq{1}/{\sqrt{\Delta}}.
\]

\noindent
In  light of the above result, Theorem \ref{thrm:RapidMixingMain} follows using 
the same arguments  as those for the proof of Theorem~1 in \cite{DFHV}.

\end{proof}

\subsection{Proof of Lemma \ref{lemma:ArbitraryPair}}\label{sec:lemma:ArbitraryPair}

\begin{proof}[Proof of Lemma \ref{lemma:ArbitraryPair}.1 and \ref{lemma:ArbitraryPair}.2]

The treatment for both cases are very similar. Note that each vertex  can only become 
disagreeing at time step $t$, if it is updated at time $t$ and it is next to a vertex which is also disagreeing.
Furthermore, for  such vertex the probability  to become disagreeing is at most $e/\Delta$.  
Using the  observations and noting that each disagreeing vertex has at most $\Delta$
non-disagreeing neighbors we get the following: The expected number of disagreements at each 
time step increases by a factor which is at most $(1+\Delta\frac{e}{n\Delta}) \leq \exp\left( 3/n\right)$.

By using induction, it is straightforward that for any $t\geq 0$ it holds
\begin{equation}\label{eq:ExpDis-ArbPair}
\Exp{X_t\oplus Y_t} \leq \exp\left (3t/n\right).
\end{equation}
Then,  the statement 1, follows by plugging into \eqref{eq:ExpDis-ArbPair} $t=Cn/\epsilon$.  
The statement 2 follows by plugging into \eqref{eq:ExpDis-ArbPair} $t=T\log\Delta$.
\end{proof}

\begin{proof}[Proof of Lemma \ref{lemma:ArbitraryPair}.3]
Recall that for any $X_t,Y_t$,  we have that
$
D_t=\{w: X_t\neq Y_t\},
$
while let
\[
D_{\leq t}=\textstyle \bigcup_{t'\leq t}D_{t'}.
\]
Also, let $H_{\leq t}=|D_{\leq t}|$. We prove that for any integer $1\leq \ell\leq n$, for $T=Cn/\epsilon$,
it holds that
\begin{equation}\label{eq:Tail4CumHamming}
\Prob{H_{\leq T}\geq \ell} \leq \exp\left( -\ell e^{-6C/\epsilon}\right).
\end{equation}
For $1\leq i\leq \ell $, let $t_i$ be the time at which the $i$'th disagreement is generated (possibly counting the
same vertex set multiple times). Denote $t_0=0$. Let $\eta_i:=t_i-t_{i-1}$ be the waiting time for the
formation of the $i$'th disagreement. Conditioned on the evolution at all times in $[0,t_i]$, the distribution of
$\eta_i$  stochastically dominates a geometric distribution with success probability $\rho_i$ and range 
$\{1,2,\ldots\}$, where
\[
\rho_i=\frac{e\cdot \min\{i\Delta,n-i\}}{n\Delta}.
\]
This is because at all times prior to $t_i$ we have $H_t\leq i$, while the sets $H_{\leq t}$ increases with probability
at most $\rho_i$ at each time step, regardless of the history. The quantity $\min\{i\Delta,n-i\}$ in the 
numerator in the expression for $\rho_i$ is an upper bound on the number of vertices that are non-disagreeing 
neighbors of the disagreeing vertices. The quantity $e/(n\Delta)$ is an upper bound for the probability of 
a neighbor of a disagreement to be chosen and become a disagreement itself. 

Hence, $\eta_1+\cdots+\eta_{\ell}$ stochastically dominates the sum of independent geometrically distributed
random variables with success probability $\rho_1,\ldots\rho_{\ell}$. For any real $x\geq 0$ it holds that
\[
\Prob{\eta_i\geq x}\geq (1-\rho_i)^{\lceil x\rceil -1}\geq \exp\left [-\frac{\rho_i}{1-\rho_i}x\right]\geq e^{2\rho_ix}.
\]
In the above series of inequalities we used that $1-x>\exp(-\frac{x}{1-x})$ for $0<x<1$ and $\rho_i<1/3$.

The above inequality implies that $\eta_1+\cdots+\eta_{\ell}$ dominates the sum of exponential random
variables with parameters $2\rho_1,2\rho_2,\ldots, 2\rho_{\ell}$. Since $\rho_i\leq i\rho$, where 
$\rho=\frac{e}{n}$, we have that $\eta_1+\cdots+\eta_{\ell}$ stochastically dominates the sum of exponential
random variables $\zeta_1,\zeta_2,\ldots,\zeta_{\ell}$ with parameters $2\rho,4\rho, \ldots, 2\ell \rho$, 
respectively. 

Consider the problem of collecting $\ell$ coupons, assuming that each coupon is generated by 
a Poisson process with rate $2\rho$. The time interval between collecting the $i$'th coupon and the 
$i+1$'st coupon is exponentially distributed with rate $2(\ell-i)\rho$. Hence the time to collect
all $\ell$ coupons has the same distribution as $\zeta_1+\zeta_2+\cdots+\zeta_{\ell}$. But the event that
the total delay is less than $T$ nothing but the intersection of the (independent) events that all coupons are
generated in the time interval $[0,T]$. The probability of this event is
\[
(1-\exp^{-2T\rho})^{\ell}< \exp\left( -\ell\exp\left(-2{Ce}/{\epsilon} \right)\right).
\]
The above completes the proof of \eqref{eq:Tail4CumHamming}.
Then we proceed as follows:
\begin{eqnarray}
\Exp{|X_T\oplus Y_T|\cdot \mathbf{1}\{{\bf E}_T\}} &\leq &\Exp{H_{\leq T}\mathbf{1}\{ {\cal E}_T \}} 
\; \leq \; \sum^n_{\ell=\Delta^{2/3}}\ell \cdot \Prob{H_{\leq T}=\ell} \nonumber \\
&\leq & \Delta^{2/3} \cdot \Prob{H_{\leq T}\geq \Delta^{2/3}}+\sum^n_{\ell=\Delta^{2/3}+1}\Prob{H_{\leq T}\geq \ell}\nonumber \\
&<&\Delta^{2/3}\sum^n_{\ell=\Delta^{2/3}}  \Prob{H_{\leq T}\geq \ell} \nonumber \\
&<&\Delta^{2/3}\sum^n_{\ell=\Delta^{2/3}}  \exp\left(-\ell \exp\left(-6C/\epsilon  \right) \right) \qquad \mbox{[from \eqref{eq:Tail4CumHamming}]} \nonumber \\
&\leq & 2\Delta^{2/3}\exp(-\Delta^{2/3}e^{-6C/\epsilon}) 
\end{eqnarray}
Note  that the above quantity is at most $\exp\left( -\sqrt{\Delta} \right)$, for large $\Delta$. This completes the
proof.
\end{proof}

\begin{proof}[Proof of Lemma \ref{lemma:ArbitraryPair}.4]
For this proof we need to use Lemma \ref{lemma:BurnIn}.
We consider the contribution to the expectation $\Exp{|S_{T\log \Delta}|}$ from the vertices inside the
ball $B_R(v)$ and the vertices outside the ball, i.e. $V\setminus B_R(v)$, where $R=\sqrt{\Delta}$.

First consider the vertices in $B_R(v)$. Lemma \ref{lemma:BurnIn} implies that for some vertex $w\in B_R(v)$
at time $T'=T\log\Delta\leq \exp(\Delta/C )$ is $50$-above suspicion for radius $2\Delta^{3/5}$
with probability at least $1-\exp(-\Delta/C)$.  
This observation implies that
\begin{equation}\label{eq:STlnDeltaInBall}
\Exp{|S_{T\log \Delta}\cap B_R(v)|} \leq \exp(-\Delta/C)  |B_R(v)|\leq \exp\left( - 4\sqrt{\Delta}\right).
\end{equation}

\noindent
To bound the number of disagreements outside $B_R(v)$, we observe that each such disagreement  comes from
a path of disagreements which starts from $v$. Such a path of disagreements is of length at least $R$. 
This observation implies that    $\Exp{|S_{T\log \Delta}\cap \bar{B}_R(v)|}$ is upper bounded by the expected
number of disagreements that start from $v$ and have length at least $R$.

Note that there are at most $\Delta^{\ell}$ many paths of disagreement of length $\ell$ that start from $v$.
Furthermore, so as a fixed path of length $\ell$ to become path of disagreement up to time $T\log \Delta$,
there should be $\ell$ updates which turn its vertices into disagreeing. Each vertex is chosen to be updated
with probability $1/n$, while it becomes disagreeing with probability at most $e/\Delta$.

All the above imply that
\begin{eqnarray}
\Exp{|S_{T\log \Delta}\cap \bar{B}_R(v)|} &\leq& \sum_{\ell\geq R}\Delta^{\ell} {T\log\Delta \choose \ell}
\left(\frac{e}{n\Delta} \right)^{\ell} \nonumber \\
&\leq& \sum_{\ell\geq R} \left( \frac{e^2 T\log \Delta }{\ell n}\right)^{\ell}\qquad \qquad\mbox{[as  ${n \choose s}\leq (ne/s)^s$ ]}
 \nonumber \\
 &\leq& \sum_{\ell\geq R} \left( \frac{e^2 C \log \Delta }{\ell \epsilon}\right)^{\ell}
 \nonumber \\
  &\leq& \left( {1}/{20}\right)^{\sqrt{\Delta}}\leq \exp\left( -10\sqrt{\Delta}\right). \label{eq:STlnDeltaOutBall}
\end{eqnarray}
Summing the bound of $\Exp{|S_{T\log \Delta}\cap B_R(v)|}$ and $\Exp{|S_{T\log \Delta}\cap \bar{B}_R(v)|}$ from \eqref{eq:STlnDeltaInBall} and \eqref{eq:STlnDeltaOutBall}, respectively gives the desired bound for 
$\Exp{|S_{T\log \Delta}|}$.
\end{proof}

\subsection{Proof of Lemma \ref{lemma:DistanceFromLight}}\label{sec:lemma:DistanceFromLight}

Fix $v$ and $R$ as specified in the statement of the theorem. Recall, for $X_t,Y_t$ we let
$D_t=\{w: X_t\oplus Y_t\}$ and denote $H(X_t,Y_t)=|D_t|$. That is, $H(X_t,Y_t)$ 
is the Hamming distance between $X_t,Y_t$.  We let the  accumulative difference be
\[
D_{\leq t}=\textstyle \bigcup_{t'\leq t}D_{t}.
\]
Also, let $H_{\leq t}=|D_{\leq t}|$.
We define the distance between the two chains $X_t,Y_t$ as follows
$$
{\cal D}(X_t,Y_t)=\sum_{v\in X_t \oplus Y_t}\Phi(v),
$$
where $\Phi:V\to [1,12]$ is defined in Theorem \ref{thrm:EigenVector}.
The metric ${\cal D}(X_t,Y_t)$ generalizes the Hamming metric in the following sense:
the disagreement in each vertex $v$ instead of contributing one it contributes
$\Phi(v)$.  Since $\Phi(v)\geq 1$, for every $v\in V$, for any two $X_t,Y_t$ 
we always have 
\begin{equation}\label{eq:HammingVsWeighted}
{\cal D}(X_t,Y_t)\geq H(X_t,Y_t).
\end{equation}

\noindent
For proving the lemma we use the following result.
\begin{lemma}\label{lemma:PathCouplingWithUniformity}
For $\delta>0$, let sufficiently small $\epsilon=\epsilon(\delta)$ and $\Delta\geq \Delta_0$.  Consider a graph $G=(V,E)$ of maximum degree $\Delta$
and let $\lambda\leq (1-\delta)\lambda_c({\Delta})$.
Also,  let $(X_t), (Y_t)$  be the Glauber dynamics on the hard-core model with fugacity $\lambda$ and underlying graphs
$G$.

For some time $t$, assume that $X_t\oplus Y_t=\{v\}$, for some $v\in V$ such that
\begin{eqnarray}\label{eq:PathCouplingGood}
W_{X_t} (v)\leq \sum_{z \in N(v)} \omega^*(z)\cdot\Phi(z) +\epsilon\Delta,
\end{eqnarray}
$W_{X_t}(v)$ is defined in \eqref{eq:DefWXT}.
Then, coupling the chains maximally we have that 
\begin{eqnarray}\nonumber
\Exp{{\cal D}(X_{t+1}, Y_{t+1})-{\cal D}(X_{t}, Y_{t})}<-c/n,
\end{eqnarray}
for appropriate  $c=c(\epsilon,\delta)>0$.
\end{lemma}

\noindent
The proof of Lemma \ref{lemma:PathCouplingWithUniformity} appears in Section \ref{sec:lemma:PathCouplingWithUniformity}.

We start by proving statement 1 of  Lemma \ref{lemma:DistanceFromLight}.

\begin{proof}[Proof of Lemma \ref{lemma:DistanceFromLight}.1]
Let
\[
T_b=\max\{C_bn, C_an\},
\]
where the quantities $C_b, C_a$ are from Lemma \ref{lemma:LightnessOfX} 
and Theorem \ref{thrm:Uniformity}, respectively.
 
Since $T_m\leq n\exp\left( \Delta/(C\log\Delta) \right)$, we can apply Theorem \ref{thrm:Uniformity}
to conclude that the desired local uniformity properties holds with high probability for all 
$t\in I:=[T_b,T_m]$.

For $t\geq T_b$ we define the following {\em bad} events:
\begin{itemize}
\item ${\cal E}(t)$ denotes the event that at some time $s<t$, it holds $H_s>\Delta^{2/3}$
\item ${\cal B}_1(t)$ denotes the event that $D_{\leq t}\not \subseteq B_{\sqrt{\Delta}}(v)$
\item ${\cal B}_2(t)$ denotes the event that there exists a time $T_b\leq \tau\leq t$, 
$z\in B_{\sqrt{\Delta}}(v)$ such that
\[
W_{X_t}(z) >\Theta(z,\epsilon)=\sum_{z \in N(v)} \omega^*(z)\Phi(z)+\epsilon \Delta,
\]
where $\omega^*\in [0,1]^{V}$ is defined in Lemma \ref{lemma:FixPointEquations} and $\Phi:V\to [1,12]$.
is defined in Theorem \ref{thrm:EigenVector}
\end{itemize}
Also, we let  the event
\[
{\cal B}(t)={\cal B}_1(t)\cup {\cal B}_2(t),
\]
while we let the ``good'' event
\[
{\cal G}(t)=\bar{{\cal E}}(t)\cap \bar{\cal B}(t).
\]
We follow the convention that we drop the time $t$, for all the above events when we are referring to
the event at time $T_m$.

We bound the Hamming distance by conditioning on the above event in the following manner,
\begin{eqnarray}
\Exp{H_{T_m}}&=& \Exp{H_{T_m}\mathbf{1}\{\cal E\}}+ \Exp{H_{T_m}\mathbf{1}\{\bar{\cal E}\}\mathbf{1}\{\cal B\}}+
\Exp{H_{T_m}\mathbf{1}\{\bar{\cal E}\}\mathbf{1}\{\bar{\cal B}\}}\nonumber \\
&\leq & \Exp{H_{T_m}\mathbf{1}\{\cal E\}}+\Delta^{2/3}\Prob{\cal B}+\Exp{H_{T_m}\mathbf{1}\{{\cal G}\}}
\nonumber \\
&\leq & \exp(-\sqrt{\Delta})+\Delta^{2/3} \Prob{\cal B}+\Exp{H_{T_m}\mathbf{1}\{{\cal G}\}},\label{eq:DistanceFromLightBasis}
\end{eqnarray}
where in the last inequality we used Lemma \ref{lemma:ArbitraryPair}.3.

For the second term in the \eqref{eq:DistanceFromLightBasis} we prove the following
\begin{equation}\label{eq:DistanceFromLightPrB}
\Prob{\cal B}\leq \exp\left( -\sqrt{\Delta}\right).
\end{equation}
Finally, for the third term in the \eqref{eq:DistanceFromLightBasis} we prove the following
\begin{equation}\label{eq:DistanceFromLightExpHammGood}
\Exp{H_{T_m}\mathbf{1}\{{\cal G}\}}\leq 1/9.
\end{equation}
The part 1 of the theorem follows by plugging into \eqref{eq:DistanceFromLightBasis}, the
bounds in \eqref{eq:DistanceFromLightPrB} and \eqref{eq:DistanceFromLightExpHammGood}.
\end{proof}

\begin{proof}[Proof of \eqref{eq:DistanceFromLightPrB}]
We can bound the probability of the event ${\cal B}_1$ by a standard paths of disagreement
argument. We are looking at the probability of a path of disagreement of length $\ell=\sqrt{\Delta}$,
within $T_m=C'n/\epsilon$ steps, hence:
\begin{eqnarray}
\Prob{{\cal B}_1} &\leq &\Delta^{\ell }{T_m \choose \ell}\left( \frac{e}{n\Delta}\right)^{\ell}
\nonumber \\
&\leq& \left( {e^2 C'}/{\epsilon }\right)^{\ell} \qquad\qquad\qquad \mbox{[as ${N \choose i}\leq (Ne/i)^i$]}
\nonumber \\
&\leq & \exp\left( -2\sqrt{\Delta} \right) \label{eq:Bound4ProbB1}.
\end{eqnarray}
We can bound the probability of the event ${\cal B}_2$ by working as follows:
The assumption is that $v$ is  400-above suspicion for radius $R\geq \Delta^{3/5}$.
Then,  each vertex $z\in B_{\sqrt{\Delta}(v)}$
is 400-above suspicion for the constant radius $R'(\gamma,\delta)$ required for the statement 
for the hypothesis of Theorem \ref{thrm:Uniformity}. Therefore, in the interval $I=[T_b,T_m]$
 the uniformity condition  for each vertex $z$ fails with probability at most $\exp(-\Delta/(C\log\Delta))$. 
 More precisely,  we have that
\begin{equation}\label{eq:Bound4ProbB2}
\Prob{{\cal B}_2}\leq \exp(-\Delta/C) \Delta^{\sqrt{\Delta}+1}\leq\exp\left( -2\sqrt{\Delta}\right).
\end{equation}
Using a simple union bound, we get that $\Prob{\cal B}\leq \Prob{{\cal B}_1}+\Prob{{\cal B}_2}$. Then 
\eqref{eq:DistanceFromLightPrB} follows by plugging \eqref{eq:Bound4ProbB1} and \eqref{eq:Bound4ProbB2}
into the union bound.
\end{proof}

\begin{proof}[Proof of \eqref{eq:DistanceFromLightExpHammGood}]
Recall that for  the two chains $X_t,Y_t$ we defined the following notion of distance
$$
{\cal D}(X_t,Y_t)=\sum_{v\in X_t \oplus Y_t}\Phi(v).
$$
Note that for every $z\in V$ it holds that $1\leq \Phi(z)\leq 12$. This implies that
we always have that  ${\cal D}(X_t,Y_t)\geq H(X_t,Y_t)$.
For showing that \eqref{eq:DistanceFromLightExpHammGood} indeed holds, it suffices to
show that
\begin{equation}\label{eq:WeightDistanceFromLightExpHammGood}
\Exp{{\cal D}(X_{T_m},Y_{T_m})\mathbf{1}\{{\cal G}\}}\leq 1/9.
\end{equation}

Let $Q_0=X_t, Q_1,Q_2,\ldots, Q_h=Y_t$ be a sequence of independent sets
where $h=|X_t\oplus Y_t|$ and $Q_{i+1}$ is obtained from $Q_i$ by changing the
assignment of one vertex $w_i$ from $X_t(w_i)$ to $Y_t(w_i)$. 
We maximally couple $W_i$ and $W_{i+1}$ in one step of the Glauber dynamics to obtain
$W'_i$ and $W'_{i+1}$. More precisely, both chains update the spin of the same vertex
and maximize the probability of choosing the same new assignment for the chosen vertex.

Consider a pair $Q_i, Q_{i+1}$.  Note that  $Q_i, Q_{i+1}$ differ only on the assignment of $w_i$.
With probability $1/n$ both chains update the spin of  vertex $w_i$. Since all
the neighbors of $w_i$ have the same spin, with probability 1 we assign the same
spin on $w_i$ in both chains. Such an update reduces the distance of the two chains by 
$\Phi(w_i)$.

Consider now the update of vertex $w\in N(w_i)$. Also, w.l.o.g. assume that $Q_i(w_i)$ is
occupied while $Q_{i+1}(w_i)$ is unoccupied. It is direct that the worst case is when
$w$ is unblocked in the chain $Q_{i+1}$. Otherwise, i.e. if $w$ is blocked then with
probability 1 we have $Q_{i+1}(w)=Q_{i}(w)=$``unoccupied",
since  in  $Q_i$, we  have $w_i$ blocked.

Assuming that $w_i$ blocked in the chain $Q_i$ and unblocked in the chain $Q_{i+1}$,
we get  $Q'_i(w)\neq Q'_{i+1}(w)$ if the coupling chooses to set  $w_i$ occupied  in 
 $Q'_{i+1}$.  Otherwise,  we have $Q'_i(w)=Q'_{i+1}(w)$. Clearly, 
the disagreement happens with probability at most $\frac{\lambda}{1+\lambda}<e/\Delta$.

Therefore, given $Q_i,Q_{i+1}$, we have that
\begin{equation}\label{eq:PathCuplingGeneral}
\Exp{{\cal D}(Q'_{i+1},Q'_i)- {\cal D}(Q_{i+1},Q_i)}\leq -\frac{\Phi(w_i)}{n}+\frac{e}{n\Delta}\sum_{z\in N(w_i)}\Phi(z).
\end{equation}
Since we have that $1\leq \Phi(z)\leq 12$, for any $z\in V$ and $|N(v)|\leq \Delta$,   we get the trivial bound that
\[
\Exp{{\cal D}(Q'_{i+1},Q'_i)- {\cal D}(Q_{i+1},Q_i)}\leq {35}/{n}.
\]
Therefore, 
\begin{equation}\label{eq:NonContraction}
\Exp{{\cal D}(X_{t+1},Y_{t+1})}\leq \left(1+{35}/{n}\right){\cal D}(X_t,Y_t).
\end{equation}
The above bound is going to be used only for the burn-in phase, i.e. the first $T_b$ steps.
We use a significantly better bound for the remaining $T_m-T_b$ steps.

Since the event ${\cal G}$ holds,  for all $0\leq i\leq h$, $z\in B_R(v)$ and all $t\in [T_b,T_m-1]$,
we have that
\begin{equation}\label{eq:QiVsUniformity}
W(Q_i,z)\leq \Theta(z,\epsilon)+\Delta^{2/3}\leq \Theta(z,2\epsilon).
\end{equation}
The first inequality follows from our assumption that  both event 
$\bar{\cal E}$ and $\bar{\cal B}_2$ occur. The second follows from
the definition of the quantity $\Theta$.

Using Lemma \ref{lemma:PathCouplingWithUniformity}
and get the following: For $Q_i, Q_{i+1}$ which satisfy \eqref{eq:QiVsUniformity} it holds
that 
\[
\Exp{{\cal D}(Q'_{i+1},Q'_{i})}\leq \left(1-{C'}/{n}\right){\cal D}(Q_{i+1},Q_{i}),
\]
for appropriately chosen $C'$.
The above inequality implies the following: Given $X_t,Y_t$ and assuming that
${\cal G}(t)$ holds, we get that
\begin{equation}\label{eq:ContractionCase}
\Exp{{\cal D}(X_{t+1},Y_{t+1})}\leq \left(1-{C}/{n}\right){\cal D}(X_t,Y_t).
\end{equation}
Let $t\in [T_b,T_m-1]$. Then we have that
\begin{eqnarray}
\Exp{{\cal D}(X_{t+1},Y_{t+1})\mathbf{1}\{{\cal G}(t) \} } &=& 
\Exp{\ExpCond{{\cal D}(X_{t+1},Y_{t+1})\mathbf{1}\{{\cal G}(t) \} }{X_0,Y_0,\ldots, X_t,Y_t}}  \nonumber \\
&=& 
\Exp{\ExpCond{{\cal D}(X_{t+1},Y_{t+1})}{X_0,Y_0,\ldots, X_t,Y_t}\mathbf{1}\{{\cal G}(t) \} }  \nonumber \\
&\leq & 
\left(1-C/n \right)\Exp{{\cal D}(X_{t},Y_{t}) \mathbf{1}\{{\cal G}(t) \} }  \nonumber \\
&\leq & 
\left(1-C/n \right)\Exp{{\cal D}(X_{t},Y_{t}) \mathbf{1}\{{\cal G}(t-1) \} }.  \nonumber
\end{eqnarray}
The frist equality is Fubini's Theorem, the second equality is due to the fact that $X_0,Y_0,\ldots X_t,Y_t$
determine uniquely ${\cal G}(t)$ The first inequality uses \eqref{eq:ContractionCase} while the second
inequality uses the fact that ${\cal G}(t)\subset {\cal G}(t-1)$.
By induction, it follows that
\[
\Exp{{\cal D}(X_{T_m},Y_{T_m}) \mathbf{1}\{{\cal G}(T_m)\}}\leq \left( 1-{C}/{n}\right)^{T_m-T_b}
\Exp{{\cal D}(X_{T_b},Y_{T_b}) \mathbf{1}\{{\cal G}(T_b)\}}.
\]
Using the same arguments and \eqref{eq:NonContraction} for 
$\Exp{{\cal D}(X_{T_b},Y_{T_b}) \mathbf{1}\{{\cal G}(T_b)\}}$ we get that
\begin{equation}
\Exp{{\cal D}(X_{T_m},Y_{T_m}) \mathbf{1}\{{\cal G}(T_m)\}}\leq \left( 1-{C}/{n}\right)^{T_m-T_b}
\left(1+{35}/{n}\right)^{T_b}{\cal D}(X_0,Y_0).
\end{equation}
The result follows from the choice of constants and  noting that ${\cal D}(X_0,Y_0)<12$.
\end{proof}

\begin{proof}[Proof of Lemma \ref{lemma:DistanceFromLight}.2]
Recall from the proof of Lemma \ref{lemma:DistanceFromLight}.1 that
${\cal B}_1$  is the event that $D_{\leq T_m\not\subseteq}B_{\sqrt{\Delta}}(v)$.
Also consider  ${\cal B}'_1$ to be the event that  $D_{T_m}\not\subseteq B_{\sqrt{\Delta}(v)}$.
Noting that ${\cal B}'_1\subset {\cal B}_1$, we get that
\[
\Prob{{\cal B}'_1}\leq \Prob{{\cal B}_1}\leq\exp\left( -\sqrt{\Delta}\right),
\]
where the last inequality follows from \eqref{eq:Bound4ProbB1}.

We can assume the disagreements are contained in $B_{\sqrt{\Delta}}(v)$. By the hypothesis of
Lemma \ref{lemma:DistanceFromLight}, 
each vertex $w\in B_{\sqrt{\Delta}}(v)$ is $400$-above 
suspicion for radius $R-\sqrt{\Delta}$ in both $X_0$ and $Y_0$. Therefore, by Lemma \ref{lemma:LightnessOfX},
each vertex $w\in B_{\sqrt{\Delta}}(v)$ is $20$-above suspicion for radius $R-\sqrt{\Delta}-2$ in $X_{T_m}$
and $Y_{T_m}$ with probability at least $1-\exp(-\Delta/C_b)$. Therefore, all 
 $w\in B_{\sqrt{\Delta}}(v)$ is $50$-above suspicion for radius $R-\sqrt{\Delta}-2$ in $X_{T_m}$
and $Y_{T_m}$ with probability at least $1-\exp(-\Delta/C_b)$. 
That is, we have proven that all disagreements between $X_{T_m}$ and $Y_{T_m}$ are 50-above
suspicion for radius $R-\sqrt{\Delta}-2$ with probability at least $1-2\exp(-\Delta/C_b)$.
This proves Lemma \ref{lemma:DistanceFromLight}.2.
\end{proof}

\subsection{Proof of Lemma \ref{lemma:PathCouplingWithUniformity}}\label{sec:lemma:PathCouplingWithUniformity}

\begin{proof}[Proof of Lemma \ref{lemma:PathCouplingWithUniformity}]
Let $\Phi_{\max}=\max_{z\in V}\Phi(z)$,
where $\Phi:V(G)\to \mathbb{R}_{\geq 0}$,  as in Theorem \ref{thrm:EigenVector}.
Each vertex $v\in V$ is called a ``low degree vertex" if $\degree(v)\leq \hat{\Delta}=\frac{\Delta}{e \cdot \Phi_{\max}} $.

If $v$ is a low degree vertex then the following holds
\begin{eqnarray*}
\Exp{{\cal D}(X_{t+1}, Y_{t+1})-{\cal D}(X_{t}, Y_{t})} &\leq &
-\frac{\Phi(v)}{n}+\frac1n\sum_{z\in N(v)}\frac{\lambda}{1+\lambda}\Phi(z).
\end{eqnarray*}
We get the inequality above by working as follows: 
The distance between the two chains changes  when we updated either
$v$ or some vertex $z\in N(v)$.

With probability $1/n$ the  the update involves the vertex $v$. 
Since there is no disagreement at the  neighborhood of $v$ we can couple $X_t$ and $Y_t$ 
such that $X_{t+1}(v)=Y_{t+1}(v)$ with probability 1. That is, the distance between the chain decreases 
by $\Phi(v)$.

We make the (worst case) assumption  that all the vertices in $N(v)$ are unblocked
and unoccupied. 
We have a new disagreement between the two chains,
i.e. an increase in the distance, only if some vertex $z\in N(v)$ is chosen to be updated and one
of the chains sets $z$ occupied.  Since $X_t(v)\neq Y_t(v)$ one of the chains cannot set $z$ occupied.
Each $z\in N(v)$ is chosen with probability  $1/n$ and it is set occupied by one  the two chains 
with probability  $\frac{\lambda}{1+\lambda}$. Then, the  distance between the chains increases  by $\Phi(z)$.
Then we get the following 
\begin{eqnarray}
\Exp{{\cal D}(X_{t+1}, Y_{t+1})-{\cal D}(X_{t}, Y_{t})}  &\leq &
-\frac{\Phi(v)}{n}+\frac1n\sum_{z\in N(v)}\frac{\lambda}{1+\lambda}\Phi(z) \nonumber \\
&\leq & -\frac1n\left( \Phi(v)- \Phi_{\max} \cdot (1-\delta)\lambda_c({\Delta}) \cdot \hat{\Delta} \right) \nonumber \\
&\leq & -\frac1n\left( \Phi(v)- 1 \right) \leq -10/n, \label{eq:LowDegreePathCoupling}
\end{eqnarray}
where the last inequality follows from the fact that $1\leq \Phi(v)\leq 12$, for every $v\in V$,
$\hat{\Delta}=\frac{\Delta}{e \cdot \Phi_{\max}}$  and 
$\lambda\leq e/\Delta$.
For the case where $v$ is a high degree vertex  
we have the following
\begin{eqnarray*}
\Exp{{\cal D}(X_{t+1}, Y_{t+1})-{\cal D}(X_{t}, Y_{t})} &\leq &
-\frac{\Phi(v)}{n}+ \frac1n\sum_{z\in N(v)}\frac{\lambda}{1+\lambda}\omega^*(z)\Phi(z)+\frac{1}{n}\frac{\lambda}{1+\lambda}\epsilon\Delta. 
\end{eqnarray*}
As before, the interesting cases are those where the update involves the vertex
$v$ or  $N(v)$. As we argued above when the vertex $v$ is updated
the distance between the two chains decreases by $\Phi(v)$. 

As far as the neighbors of $v$ are regarded we observe the following:
If some $z\in N(v)$ is blocked, then with probability 1 is set unoccupied in both chains.
This means that $X_{t+1}(z)=Z_{t+1}(z)$, i.e. the distance between the
two chains remains unchanged.  If the update involves an unblocked vertex
$z\in N(v)$, then with probability $\frac{\lambda}{1+\lambda}$ the vertex
$z$ becomes occupied at only one of the two chains and the distance between the chains increases by
$\Phi(z)$.

In the inequality above, we  use also use the fact that  \eqref{eq:PathCouplingGood}
holds for the high degree vertex $v$.
Then we get that
\begin{eqnarray}
\Exp{{\cal D}(X_{t+1}, Y_{t+1})-{\cal D}(X_{t}, Y_{t})}  
&\leq &
-\frac{\Phi(v)}{n}+ \frac1n \frac{\lambda}{1+\lambda}\W_{\sigma}(v)  \nonumber \\
&\leq &
-\frac{\Phi(v)}{n}+ \frac1n\sum_{z\in N(v)}\frac{\lambda}{1+\lambda}\omega^*(z,v)\Phi(z)+
\frac{1}{n}\frac{\lambda}{1+\lambda}\epsilon\Delta. \nonumber \\
&\leq &
-\frac1n \left( 
\Phi(v)-\sum_{z\in N(v)}\frac{\lambda}{1+\lambda}\omega^*(z,v)\Phi(z)+e\epsilon
\right) \; \leq \; -c/n, \quad \;\label{eq:HighDegreePathCoupling}
\end{eqnarray}
where the last inequality follows by taking sufficiently small $\epsilon>0$.

The lemma follows from \eqref{eq:LowDegreePathCoupling} and  \eqref{eq:HighDegreePathCoupling}.
\end{proof}

\section{Random Regular (Bipartite) Graphs: Proof of Theorem \ref{thrm:RapidMixingRandom}}\label{sec:thrm:RapidMixingRandom}

It turns out that the girth restriction of Theorem \ref{thrm:RapidMixingMain} can be relaxed a bit.
The main technical reason why we need   girth at least 7 is for establishing what we
call ``local uniformity property".  
Roughly speaking, local uniformity amounts to showing that the number of unblocked neighbors of a vertex
$v$ is concentrated about the quantity $\sum_{z\in N(v)}\omega^*(z)$, where $\omega^*\in [0,1]^V$
is the fixed points of a BP-like system of equations.
In particular, uniformity amounts to  showing that the number of unblocked neighbors of $v$ is 
$\sum_{w\in N(v)}\omega^*(w)\pm \epsilon \Delta$, with probability that tends to 1
as   $\Delta$ grows.

The analysis of local uniformity could be carried out for graph with  short cycles, i.e. cycles of length less than 7.
The effect of the short cycles is an increase to the  fluctuation of the number of unblocked neighbors of a vertex. 
However, if the number of such cycles is small, i.e.  constant, then the increase in the fluctuation is negligible. 
That is, the proof of Theorem \ref{thrm:RapidMixingMain} 
carries out if, instead of girth at least 7, we have smaller girth but only a {\em constant} number of
cycles of length less than 7 around each vertex $v$. 
The above observation leads to the following corollary from Theorem \ref{thrm:RapidMixingMain}.

For some integers $\ell, g\geq 0$, let ${\cal G}_n(\ell, g)$ denote all the graphs on $n$ vertices 
such that each vertex belongs to at most $\ell$ cycles of length less than $g$.

\begin{corollary}\label{cor:RapidFewCycles}
For all $\delta>0$, there exists $\Delta_0=\Delta_0(\delta)$, $\ell=\ell(\delta)$ and $C=C(\delta)$,
for all $\Delta\geq\Delta_0$, all $\lambda<(1-\delta)\lambda_c(T_\Delta)$,
all graphs $G\in {\cal G}(\ell, 7)$ of maximum degree $\Delta$, all $\eps>0$, 
the mixing time of the Glauber dynamics satisfies:
\[  \Tmix(\eps) \leq Cn\log(n/\eps).
\]
\end{corollary}

\noindent
Using the above corollary we can show the following rapid mixing result for random regular (bipartite) graphs
with sufficiently large degree $\Delta$. 
The theorem  follows by using e.g. the  result from \cite{wormald}. Let $G$ be chosen uniformly at random 
among all $\Delta$ regular (bipartite) graphs with $n$. Then,   with probability that tends to 1 as $n$ tends to infinity 
it holds that $G\in {\cal G}(1,7)$. Then the theorem follows from  Corollary \ref{cor:RapidFewCycles}.

\end{document}